\documentclass[aps,prx,superscriptaddress,twocolumn,longbibliography,floatfix,reprint]{revtex4-2}

\usepackage{amsmath,amssymb,amsthm,mathtools}
\usepackage{booktabs}
\usepackage{graphicx}
\usepackage{needspace}
\ifdefined\pdfminorversion\pdfminorversion=7\else\fi
\usepackage[dvipsnames]{xcolor}
\usepackage{tikz}
\usepackage{quantikz}
\usetikzlibrary{arrows.meta,positioning,shapes.geometric,calc}
\usepackage{enumitem}
\usepackage{microtype}
\usepackage[colorlinks=true,linkcolor=blue, citecolor=blue]{hyperref}
\usepackage{cleveref}
\usepackage[normalem]{ulem}

\hypersetup{colorlinks=true,linkcolor=blue!65!black,citecolor=blue!65!black,urlcolor=blue!65!black,
pdftitle={Syndrome measurements enable fault-tolerant logical T gates},
pdfauthor={Kishor Bharti, Tobias Haug, and Andrew Tanggara}}

\newtheorem{theorem}{Theorem}
\newtheorem{corollary}[theorem]{Corollary}
\newtheorem{proposition}[theorem]{Proposition}
\newtheorem{lemma}[theorem]{Lemma}
\theoremstyle{definition}
\newtheorem{definition}[theorem]{Definition}

\newtheorem{remark}[theorem]{Remark}

\newcommand{\cC}{\mathcal C}
\newcommand{\cD}{\mathcal D}
\newcommand{\cS}{\mathcal S}
\newcommand{\cN}{\mathcal N}
\newcommand{\cP}{\mathcal P}
\newcommand{\bbR}{\mathbb R}
\newcommand{\bbZ}{\mathbb Z}
\newcommand{\syn}{\operatorname{syn}}
\newcommand{\wt}{\operatorname{wt}}
\newcommand{\supp}{\operatorname{supp}}
\newcommand{\rank}{\operatorname{rank}}
\newcommand{\qcode}[3]{[\![#1,#2,#3]\!]}
\newcommand{\ind}{\mathbf 1}
\newcommand{\PC}{P_{\cC}}

\newcommand{\RM}{\mathrm{RM}}

\newcommand{\SMLong}{Appendix}
\newcommand{\SM}{App.}

\newcommand{\CQT}{Centre for Quantum Technologies, National University of Singapore, 3 Science Drive 2, Singapore 117543\looseness=-1}
\newcommand{\NTU}{Nanyang Quantum Hub, School of Physical and Mathematical Sciences, Nanyang Technological University, Singapore 639673\looseness=-1}
\def\QUICS{QuICS, NIST/University of Maryland, College Park, Maryland 20742, USA}
\def\UMIACS{UMIACS, University of Maryland, College Park, Maryland 20742, USA}
\def\TII{Quantum Research Center, Technology Innovation Institute, Abu Dhabi, United Arab Emirates}

\begin{document}

\title{\texorpdfstring{Syndrome measurements enable deterministic fault-tolerant $T$ gates}{Syndrome measurements enable deterministic fault-tolerant T gates}}

\author{Kishor Bharti}
\altaffiliation{Currently at IonQ Inc.}
\affiliation{\QUICS}
\affiliation{\UMIACS}

\author{Tobias Haug}
\affiliation{\TII}

\author{Andrew Tanggara}
\affiliation{\CQT}
\affiliation{\NTU}

\begin{abstract}

Non-Clifford gates are essential for universal quantum computation, yet implementing them fault-tolerantly remains a central challenge for stabilizer codes. Here, we show how a syndrome degree of freedom can mediate a logical non-Clifford gate. Releasing one stabilizer check makes an additional logical qubit available within the encoded data block. Two Pauli rotations, followed by syndrome measurement and Clifford feed-forward, then implement a deterministic logical $T$ gate on every stabilizer code of distance at least two, in a suitable logical basis. During the ideal rotations, the state remains in an intermediate stabilizer code whose distance we determine exactly. For pure codes with a balanced factorization, this distance grows with the original code distance, whereas the maximum weight of the stabilizer generators bounds the intermediate distance from above. We realize the mechanism in two circuits that tolerate a single fault under local stochastic circuit noise: a fixed 22-qubit construction admitting recursive error suppression and a direct Golay-code gate protected by measuring a stabilizer check transported through the non-Clifford rotation, which can recover the unknown encoded state even after
rejection.
Serial implementations with ancilla reuse, reset, and flexible two-qubit connectivity require at most 33 and 32 physical qubits, respectively. These results establish a general mechanism for logical non-Clifford gates and demonstrate how syndrome measurements and recovery can protect the intermediate evolution of encoded information.

\end{abstract}

\maketitle

\makeatletter
\let\oldaddcontentsline\addcontentsline
\renewcommand{\addcontentsline}[3]{%
  \begingroup
  \let\addtocontents\@gobbletwo
  \oldaddcontentsline{#1}{#2}{#3}%
  \endgroup}
\makeatother

\section{Introduction}\label{sec:intro}

Universal fault-tolerant quantum computation %
requires a protected non-Clifford gate. %
Storing quantum information and performing Clifford operations fault-tolerantly are well understood~\cite{shor1996fault,gottesman1998theory,knill2005quantum,aliferis2006quantum}. However, the Eastin--Knill theorem forbids a universal transversal gate set%
~\cite{eastin2009restrictions,zeng2011transversality}, and related no-go results constrain the local implementation of logical non-Clifford gates%
~\cite{bravyi2013classification,pastawski2015fault,jochym2018disjointness}. 
Every fault-tolerant architecture must therefore carefully decide how to generate its non-Clifford resource and, just as importantly, how the encoded information is protected while that resource is being applied.

 Existing methods protect the non-Clifford operation in different ways. Magic-state methods %
 use state preparation: a noisy resource state is injected, distilled, or cultivated until it is good enough, and then consumed by gate teleportation using only Clifford gates, measurements, and feed-forward~\cite{gottesman1999demonstrating,zhou2000methodology,bravyi2005universal,bravyi2012magic,litinski2019magic,chamberland2019fault,gidney2024cultivation,hirano2025cultivation,chen2026cultivation,rosenfeld2025cultivation}. Code switching and gauge fixing %
 can move between related codes with complementary transversal gates~\cite{paetznick2013universal,anderson2014fault,bombin2015gauge,kubica2015universal}; code deformation and
lattice surgery admit the same description~\cite{introVuillot2019Gauge}.
Recent protocols and trapped-ion experiments have made switching between small color codes practical and quantified its cost relative to distillation~\cite{butt2024fault,heussen2024efficient,daguerre2025code,beverland2021cost,pogorelov2025experimental,introDaguerre2025Magic}. Concatenated schemes combine codes with complementary transversal gates~\cite{jochym2014using,nikahd2017nonuniform,chamberland2017overhead,chamberland2016thresholds}. Pieceable fault tolerance divides a nontransversal gate into pieces separated by rounds of correction~\cite{yoder2016universal}.
Hybrid and nonuniform concatenation schemes combine selective encoding with code switching or pieceable fault tolerance to reduce the data-qubit requirements for universal fault-tolerant computation~\cite{lin2020concatenated,nikahd2023low}.
In each case, the protection available \emph{during} a logical gate, and the way faults propagate between correction steps, determine whether the gate is fault tolerant. %

A more recent line of work uses stabilizer projection itself as part of the logical operations. 
Physical rotations followed by syndrome measurement can prepare encoded resource states directly inside the computational code~\cite{gavriel2023transversal,introLiu2026InSitu}, drive analog logical rotations in partially fault-tolerant architectures~\cite{akahoshi2024partially,toshio2025practical,choi2023arbitrary,akahoshi2024compilation,toshio2026mutation,ismail2026transversal}, induce prescribed diagonal logical channels~\cite{introHu2022Channels}, generate random logical unitaries and unitary designs~\cite{cheng2025emergent}, and realize syndrome-resolved logical rotations through adaptive sequences of disjoint Pauli factors
~\cite{yoshioka2026syndrome}. Continuous logical rotations have been analyzed in the surface code~\cite{huang2025robust} and demonstrated in the Steane code~\cite{huang2026continuous}. 

These developments establish syndrome measurement as a tool for 
logical operations. They motivate a complementary question: how can a syndrome degree of freedom mediate a prescribed non-Clifford gate, and what protection is required throughout the resulting evolution?

Here, we use a syndrome degree of freedom to mediate an exact logical $T$ gate on every stabilizer code of distance $d\geq2$, in a suitable logical basis. Temporarily releasing one stabilizer constraint makes an additional logical qubit available within the same physical block. Two Pauli rotations, with overlapping anticommuting factors, followed by syndrome measurement and Clifford feed-forward, implement the same intended gate in both ideal branches. Each rotation acts on at most $\lceil(d+1)/2\rceil$ qubits. We determine the intermediate code distance for the ideal rotations exactly and show that bounded-weight checks impose an upper bound %
on this protection. We further identify fault paths that escape correction despite a large intermediate distance, establishing the additional requirements for a protected implementation.

We meet these requirements in two complementary fault-tolerant constructions. 
{Selective concatenation of a Clifford image of the Steane code, using a Reed--Muller block and a two-qubit repetition block, together with outcome-dependent Pauli measurements, yields a protected logical $T$ gate on 22 data qubits.}
The gate uses fifteen physical $T$-type gates and admits recursive error suppression below a positive threshold. 
On the 23-qubit Golay code, we protect the physical four-qubit rotation within the original block through a
transported stabilizer check
via the circuit-gauge method of Dasu and Criger~\cite{dasu2026flagging}. 
Even if an attempt is rejected, we can restore the unknown encoded input so that the gate can be attempted again.
Both circuits tolerate any single circuit fault under local stochastic noise, with %
maximum simultaneous qubit counts %
of 33 and 32 qubits under serial ancilla reuse with reset and flexible two-qubit connectivity. 
 
Together, they demonstrate how protection tailored to the intermediate evolution enables non-Clifford gates on encoded data.
For the fixed 22-qubit code, the protected Clifford operations of \SM{}~\ref{app:combined22}, together with this logical $T$ gate, give a universal fault-tolerant gate set.

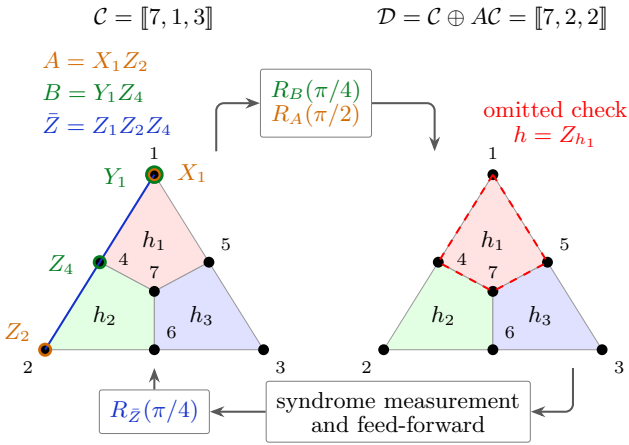
\begin{figure}[!b]
\centering
\resizebox{0.99\columnwidth}{!}{%

\begin{tikzpicture}[x=1cm,y=1cm,font=\small,
  dot/.style={circle,fill=black,inner sep=1.5pt},
  flow/.style={-{Stealth[length=2mm]},thick,draw=black!65},
  box/.style={draw=black!50,rounded corners=1pt,
              inner sep=4pt,align=center}]
\definecolor{acol}{RGB}{210,105,0}
\definecolor{bcol}{RGB}{0,125,30}
\definecolor{lcol}{RGB}{20,50,205}

\node at (1.5,5.45) {$\mathcal C=[\![7,1,3]\!]$};
\node at (6.15,5.45)
  {$\mathcal D=\mathcal C\oplus A\mathcal C=[\![7,2,2]\!]$};
\node[text=acol,anchor=west] at (-.15,4.85) {$A=X_1Z_2$};
\node[text=bcol,anchor=west] at (-.15,4.42) {$B=Y_1Z_4$};
\node[text=lcol,anchor=west] at (-.15,3.99)
  {$\bar Z=Z_1Z_2Z_4$};

\node[box] (rotations) at (3.72,4.28)
  {$\color{bcol}R_B(\pi/4)$\\[-1pt]
   $\color{acol}R_A(\pi/2)$};
\draw[flow,rounded corners=3pt]
  (2.35,3.62) |- (rotations.west);
\draw[flow,rounded corners=3pt]
  (rotations.east) -| (5.35,3.62);

\foreach \s in {0,4.65}{%
 \begin{scope}[xshift=\s cm]
  \coordinate (q1) at (1.5,3.3);
  \coordinate (q2) at (0,0.9);
  \coordinate (q3) at (3,0.9);
  \coordinate (q4) at (.75,2.1);
  \coordinate (q5) at (2.25,2.1);
  \coordinate (q6) at (1.5,.9);
  \coordinate (q7) at (1.5,1.7);
  \filldraw[fill=red!11,draw=black!40]
    (q1)--(q4)--(q7)--(q5)--cycle;
  \filldraw[fill=green!11,draw=black!40]
    (q4)--(q2)--(q6)--(q7)--cycle;
  \filldraw[fill=blue!12,draw=black!40]
    (q5)--(q7)--(q6)--(q3)--cycle;
  \node at (1.5,2.4) {$h_1$};
  \node at (.82,1.37) {$h_2$};
  \node at (2.15,1.37) {$h_3$};
  \foreach \j in {1,...,7}
    {\node[dot] at (q\j) {};}
    \node[font=\scriptsize,above=3pt] at (q1) {{$1$}};
    \node[font=\scriptsize,below left=2pt] at (q2) {{$2$}};
    \node[font=\scriptsize,below right=2pt] at (q3) {{$3$}};
    \node[font=\scriptsize,right=4pt,yshift=1pt] at (q4) {{$4$}};
    \node[font=\scriptsize,above right=2pt] at (q5) {{$5$}};
    \node[font=\scriptsize,above right=2pt] at (q6) {{$6$}};
    \node[font=\scriptsize,above=2pt] at (q7) {{$7$}};
 \end{scope}
}

\draw[lcol,thick] (0,.9)--(1.5,3.3);
\foreach \p in {(1.5,3.3),(0,.9)}
  \draw[acol,line width=1.1pt] \p circle[radius=.077];
\draw[bcol,line width=1.1pt] (1.5,3.3) circle[radius=.108];
\draw[bcol,line width=1.1pt] (.75,2.1) circle[radius=.077];
\node[text=acol,anchor=west] at (1.70,3.32) {$X_1$};
\node[text=bcol,anchor=east] at (1.25,3.26) {$Y_1$};
\node[text=bcol,anchor=east] at (.52,2.07) {$Z_4$};
\node[text=acol,anchor=east] at (-.08,1.14) {$Z_2$};

\draw[red,thick,dashed]
  (6.15,3.3)--(5.4,2.1)--(6.15,1.7)--(6.9,2.1)--cycle;
\node[text=red,align=center] at (7.01,3.99)
  {omitted check\\[-1pt]$h=Z_{h_1}$};

\node[box] (recovery) at (4.85,.05)
  {syndrome measurement\\[-1pt]and feed-forward};
\node[box,text=lcol] (logical) at (1.5,.05)
  {$R_{\bar Z}(\pi/4)$};

\draw[flow,rounded corners=3pt]
  (7.25,.65) |- (recovery.east);
\draw[flow] (recovery.west) -- (logical.east);
\draw[flow,rounded corners=3pt]
  (logical.north) -- (1.5,.65);
\end{tikzpicture}

}
\caption{The two-rotation logical-$T$ gadget for the Steane code. The factors $A=X_1Z_2$ and $B=Y_1Z_4$ satisfy $AB=i\bar Z$, with $\bar Z=Z_1Z_2Z_4$, and have the same nonzero syndrome. The ideal rotations remain in the intermediate code $\mathcal D=\mathcal C\oplus A\mathcal C$, a $[\![7,2,2]\!]$ code obtained by releasing one independent check, $h=Z_{h_1}$. The operators $h$ and $A$ act as logical $Z$ and $X$ on the extra encoded qubit. Syndrome measurement and feed-forward return the state to $\mathcal C$ with $R_{\bar Z}(\pi/4)$ applied. The retained generator basis is described in Sec.~\ref{sec:intermediate}.}
\label{fig:surfacecode}
\end{figure}

\section{Two rotations and one syndrome measurement}\label{sec:identity}

\subsection{Notation}\label{subsec:notation}

We consider an $\qcode{n}{k}{d}$ stabilizer code $\cC$ with stabilizer group $\cS=\langle g_1,\ldots,g_{r_{\cS}}\rangle$, specified by $r_{\cS}=n-k$ independent generators, and normalizer $\cN(\cS)$. The independent generating set fixes the syndrome coordinates; redundant checks, as in many quantum LDPC codes, may also be measured and change neither the stabilizer group nor the code. Up to an overall phase, an $n$-qubit Pauli is a tensor product $E=E_1\otimes\cdots\otimes E_n$ with $E_j\in\{I,X,Y,Z\}$. Its syndrome is the vector $\syn(E)=(s_1,\ldots,s_{r_{\cS}})\in\mathbb F_2^{r_{\cS}}$, where $s_j=1$ exactly when $E$ anticommutes with $g_j$; syndromes add in the vector space $\mathbb F_2^{r_{\cS}}$, $\syn(EF)=\syn(E)+\syn(F)$. Throughout, Pauli cosets, normalizers, and set differences are understood modulo global phase; when a phase matters, we choose an explicit Hermitian representative. For a Hermitian Pauli $Q$ we write $R_Q(\theta)=e^{-i\theta Q/2}=\cos(\theta/2)I-i\sin(\theta/2)Q$, so that $R_L(\pi/4)=e^{-i\pi L/8}$ is the $T$-type gate about $L$ and $R_L(\pi/2)$ is the $S$-type Clifford. ``$T$-type'' always refers to this angle about an arbitrary Pauli axis; a physical $T$-type gate is such a rotation on a single qubit. $\PC$ denotes the projector onto $\cC$, $\wt$ the weight and $\supp$ the support of a Pauli. Barred Pauli symbols denote encoded-qubit operators; $A,B,L,h$ denote explicit physical Pauli representatives, with their action on the relevant code space stated when used.

\subsection{Factoring a logical Pauli}

Let $L\in\cN(\cS)\setminus\cS$ be a Hermitian logical Pauli. A \emph{factorization of $L$} is a pair of Hermitian Paulis $A,B$ with
\begin{equation}
AB=iL.
\label{eq:factor}
\end{equation}
Since $L$ is Hermitian and $AB=iL$ forces $BA=(AB)^\dagger=-iL$, the factors anticommute, $\{A,B\}=0$. Since syndromes add, $\syn(A)+\syn(B)=\syn(AB)=\syn(L)=0$, and hence
\begin{equation}
s:=\syn(A)=\syn(B).
\label{eq:common-syn}
\end{equation}
Unless stated otherwise, $\syn$ denotes the syndrome with respect to the original stabilizer group $\cS$; we write $\syn_{\cS_0}$ for the retained-check syndrome of the intermediate code introduced in Sec.~\ref{sec:intermediate}. We always require $s\neq0$, i.e. $A\notin\cN(\cS)$. If $\wt(A)<d$ this is equivalent to $A\notin\cS$, which holds automatically in a pure code (a code all of whose nonidentity stabilizers have weight at least $d$).

\subsection{The identity}\label{subsec:identity}

To implement the $T$-type logical gate, we apply physical rotation $R_B(\pi/4)$ followed by $R_A(\pi/2)$, then measure the syndrome of the code.
The rotations $R_B(\pi/4)$ and $R_A(\pi/2)$ take the state out of the code space, but because $A$ and $B$ share the syndrome $s$ they connect only the zero-syndrome and syndrome-$s$ sectors. A subsequent syndrome measurement therefore has just two possible outcomes, each of which selects one of two logical rotations, and an outcome-dependent Clifford correction makes their logical actions identical.

\begin{theorem}[Deterministic $T$]\label{cor:det-t}
For every encoded state $|\psi\rangle\in\cC$, perform the following operations in order:
\begin{enumerate}[label=\arabic*.]
\item Apply $R_B(\pi/4)$, then $R_A(\pi/2)$.
\item Measure the stabilizer generators of $\cC$. Only outcomes $0$ and $s$ occur, with probabilities
    \begin{equation}
        \Pr(0)=\Pr(s)=\frac12.
    \label{eq:det-t-probabilities}
    \end{equation}
\item For outcome $0$, apply no correction. For outcome $s$, apply $A$, then $R_L(\pi/2)$.
\end{enumerate}

Every corrected branch implements
\begin{equation}
    |\psi\rangle\longmapsto R_L(\pi/4)|\psi\rangle,
\label{eq:det-t-action}
\end{equation}
up to a global phase.
\end{theorem}

\emph{Proof sketch.} Expanding the two rotations and using $AB=iL$,
\begin{equation}
R_A(\pi/2)R_B(\pi/4)=\tfrac{1}{\sqrt2}R_L(\pi/4)-\tfrac{i}{\sqrt2}\bigl(\cos\tfrac{\pi}{8}\,A+\sin\tfrac{\pi}{8}\,B\bigr).
\label{eq:sketch}
\end{equation}
The first term preserves $\cC$ while the second maps $\cC$ into the syndrome-$s$ sector, so the two terms are the orthogonal, equally likely branches of the syndrome measurement. Multiplying the second term by $A$ returns it to $\cC$ as $-\tfrac{i}{\sqrt2}R_L(-\pi/4)$, and $R_L(\pi/2)R_L(-\pi/4)=R_L(\pi/4)$. The complete proof of Theorem~\ref{cor:det-t} and the identity for arbitrary rotation angles are given in \SM{}~\ref{app:core-proofs}. %

Both $R_A(\pi/2)$ and the outcome-dependent logical Clifford $R_L(\pi/2)$ can be replaced by Pauli measurements, so that every executed correction is a Pauli. If $h\in\cS$ anticommutes with $A$, then $G=iAh$ is a Hermitian Pauli observable. Apply $R_B(\pi/4)$ and measure $G$, obtaining $y=\pm1$. If $y=-1$, measure $M=AhL$, obtaining $r=\pm1$.
Then measure $h$, obtaining $z=\pm1$, and apply $A$ if $z=-1$. Finally, apply $L$ if $y=-1$ and $rz=-1$. For $y=+1$ the state already carries the logical rotation $R_L(\pi/4)$. For $y=-1$ it carries $R_L(-\pi/4)$, and the $M$ measurement followed by the $h$ measurement changes the logical angle from $-\pi/4$ to $-\pi/4+rz\pi/2$: this is $\pi/4$ when $rz=1$, while for $rz=-1$ the final $L$ correction gives the same gate up to phase. \SM{}~\ref{app:pm} proves the identities and gives the complete outcome rule.

\subsection{Optimal factor weights}

The support of each factor in $AB=iL$ determines how many data qubits its rotation touches. 
When a multi-qubit Pauli rotation is compiled into CNOTs and a single-qubit rotation, as in Fig.~\ref{fig:parity-compilation}, a weight-$w$ factor uses $2(w-1)$ CNOTs~\cite{cowtan2020phase}. 
We therefore ask how small the two factors can be; a  balanced factorization 
minimizes the larger support. 

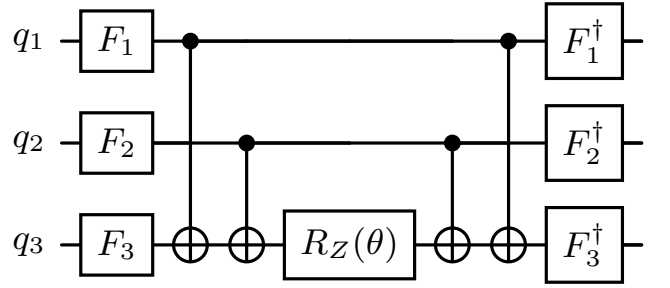
\begin{figure}[!htbp]
\centering
\resizebox{\columnwidth}{!}{%
\begin{quantikz}[row sep=0.24cm,column sep=0.17cm]
\lstick{$q_1$} & \gate{F_1} & \ctrl{2} & \qw & \qw & \qw & \ctrl{2} & \gate{F_1^\dagger} & \qw \\
\lstick{$q_2$} & \gate{F_2} & \qw & \ctrl{1} & \qw & \ctrl{1} & \qw & \gate{F_2^\dagger} & \qw \\
\lstick{$q_3$} & \gate{F_3} & \targ{} & \targ{} & \gate{R_Z(\theta)} & \targ{} & \targ{} & \gate{F_3^\dagger} & \qw
\end{quantikz}}
\caption{Physical compilation of $R_{P_1P_2P_3}(\theta)$, read from left to right. The Clifford $F_j$ obeys $F_jP_jF_j^\dagger=Z_j$: take $F_j$ as $H$, $HS^\dagger$, or $I$ for $P_j$ equal to $X$, $Y$, or $Z$, respectively.
Every wire is one physical qubit. Only the central single-qubit rotation is non-Clifford when $\theta=\pi/4$. With fault-free operations, the displayed circuit implements $R_{P_1P_2P_3}(\theta)$ exactly.  Fault tolerance requires the additional protection described in Secs.~\ref{sec:selective} and~\ref{sec:golay}.
}
\label{fig:parity-compilation}
\end{figure}

\begin{theorem}[Weights]\label{thm:weights}
    Let $L$ be a logical Pauli of weight $\wt(L)$ and $(A,B)$ any factorization $AB=iL$. 
    Then 
    \begin{equation}
        \wt(A)+\wt(B)\geq\wt(L)+1.
    \label{eq:factor-weight-bound}
    \end{equation}
    Equality is attained by the following construction. 
    Choose a shared qubit $q\in\supp(L)$ and partition the other support qubits into disjoint sets $S_A,S_B$, so $\supp(L)=\{q\}\sqcup S_A\sqcup S_B$. Choose single-qubit Paulis $A_q,B_q$ with $A_qB_q=iL_q$, where $L_j$ is the single-qubit Pauli factor of $L$ on qubit $j$, and set 
    \begin{equation}
    \begin{aligned}
    A&=A_q\otimes\bigotimes_{j\in S_A}L_j,\quad
    B&=B_q\otimes\bigotimes_{j\in S_B}L_j,
    \end{aligned}
    \label{eq:split-pair}
    \end{equation}
    with identities elsewhere.
    Thus $S_A=\supp(A)\setminus\{q\}$ and $S_B=\supp(B)\setminus\{q\}$.
    If $\wt(L)=d\geq2$, every  balanced factorization $|S_A|=\lceil(d-1)/2\rceil$, $|S_B|=\lfloor(d-1)/2\rfloor$ has $\syn(A)=\syn(B)\neq0$ and
    \begin{equation}
    \begin{aligned}
    \max\{\wt(A),\wt(B)\}&=\left\lceil\frac{d+1}{2}\right\rceil,\\
    \min\{\wt(A),\wt(B)\}&=\left\lfloor\frac{d+1}{2}\right\rfloor,
    \end{aligned}
    \label{eq:weights}
    \end{equation}
    and its maximum factor weight is optimal among all factorizations of any minimum-weight logical.
\end{theorem}
The proof, by counting supports, is given in \SM{}~\ref{app:core-proofs}.

\emph{State-injection limit.}\label{rem:injection}
 In the  factorization construction of Theorem~\ref{thm:weights}, take $S_B=\varnothing$. Then $B$ is a single-qubit Pauli and $A$ has weight $\wt(L)$, so Theorem~\ref{cor:det-t} transfers one physical $T$-type rotation to the logical qubit through syndrome projection and Clifford feed-forward. This induces the same ideal logical channel as state injection and gate teleportation~\cite{gottesman1999demonstrating,zhou2000methodology,bravyi2005universal,gavriel2023transversal}; the physical circuits differ because conventional injection prepares and consumes an ancillary magic state, whereas the unbalanced gadget rotates a data qubit and uses the complementary weight-$\wt(L)$ Clifford rotation. We say ``the same ideal logical channel'' deliberately: the noisy channels and resource costs of the two methods are not equivalent. The  balanced factorization used below trades this weight-one non-Clifford factor for a smaller largest rotation support.

 Explicit factorizations for small codes, the Golay code, and a BCH code, together with their intermediate-code parameters, are collected in \SM{}~\ref{app:intermediate-details}, Table~\ref{tab:split}; the Steane example of Fig.~\ref{fig:surfacecode} uses two weight-two rotations and has a $\qcode{7}{2}{2}$ intermediate code.

\section{The intermediate code}\label{sec:intermediate}

The rotations $R_B(\pi/4)$ and $R_A(\pi/2)$ are not logical operations of $\cC$, so after the rotations the state leaves the code space. Not all protection is lost, however: the checks that commute with both rotation axes can still be measured, and we will see that half of the stabilizer group indeed survives. To identify the surviving checks, write $A\cC=\{A|\psi\rangle:|\psi\rangle\in\cC\}$ for the syndrome-$s$ image of the code space. Expanding the first rotation, $R_B(\theta)|\psi\rangle=\cos(\theta/2)|\psi\rangle-i\sin(\theta/2)B|\psi\rangle$, where $B|\psi\rangle\in A\cC$ because $B=iAL$ and $L$ preserves $\cC$. Both $A$ and $B$ exchange $\cC$ and $A\cC$, so both rotations, including $R_A(\pi/2)$, preserve the direct sum $\cC\oplus A\cC$. This is the space preserved throughout the ideal rotations, and we now determine its stabilizers and distance.

Let $\cS_0=\{g\in\cS:[g,A]=0\}$ be the subgroup of stabilizers that commute with $A$. These are precisely the checks that remain fixed while the state belongs to $\cC\oplus A\cC$. Because $A$ has nonzero syndrome, at least one stabilizer anticommutes with it; since the commutation sign with $A$ is a homomorphism from $\cS$ to $\{\pm1\}$, exactly half of the stabilizers commute with $A$, and $\cS_0$ has $r_{\cS}-1$ independent generators. Choose any $h\in\cS\setminus\cS_0$ as the omitted check.
In Fig.~\ref{fig:surfacecode}, $A$ and $B$ anticommute with $X_{h_2}$ and $Z_{h_1}$. Replacing $X_{h_2}$ by $X_{h_2}Z_{h_1}$ gives a generator that commutes with both factors, leaving $h=Z_{h_1}$ as the single omitted check.

\begin{definition}\label{def:intermediate}
The \emph{intermediate code} of the factorization $AB=iL$ is the stabilizer code $\cD$ of $\cS_0$. %

\end{definition}

Removing one independent stabilizer increases the code-space dimension from $2^k$ to $2^{k+1}$. 
Thus, the intermediate code $\cD$ encodes $k+1$ qubits: the $k$ logical qubits of $\cC$, and one additional logical qubit.
This additional logical qubit, labeled $p$, has logical Pauli operators $\Bar{X}_p=A$ and $\Bar{Z}_p=h$, where $p$ labels the additional logical qubit encoded in the same $n$ physical qubits.

The operator $A$ exchanges the two syndrome sectors, while $h$ distinguishes them, so they act as logical $X$ and $Z$ on this additional qubit. 
For a logical Pauli representative $Q$ of $\cC$, define $\widetilde Q=Qh^{\epsilon_Q}$, where $\epsilon_Q=0$ if $[Q,A]=0$ and $\epsilon_Q=1$ otherwise. These representatives commute with $A$ and $h$, preserve the original logical commutation relations, and act as $Q$ on $\cC$.
For $|\Bar{\psi}\rangle\in\cC$, the identification
$\cD\simeq\cC\otimes\mathbb C^2$ is specified by
\begin{equation}
\begin{aligned}
    |\Bar{\psi}\rangle
    \longleftrightarrow|\Bar{\psi}\rangle\otimes|0\rangle_p,\quad
    A|\Bar{\psi}\rangle
    \longleftrightarrow|\Bar{\psi}\rangle\otimes|1\rangle_p.
\end{aligned}
\end{equation}
The representatives $\widetilde Q$ act on the first factor, while $A$ and $h$ act on qubit $p$.

Under this logical tensor-product identification, the two rotations act about $\bar L_h\otimes\bar Y_p$ and $\bar X_p$, where $\bar L_h$ is represented by $Lh$ and acts on the original $k$ logical qubits; the complete logical-operator description of $\cD$ is given in
\SM{}~\ref{app:sector}, Eq.~\eqref{eq:path-logical-basis}.
Now, let us define
\begin{align}
    \mu(s)&=\min\{\wt(E):E\in\cP_n,\ \syn(E)=s\},\label{eq:mu}\\
    \nu(A)&=\min\{\wt(g):g\in\cS\setminus\cS_0\}.\label{eq:nu}
\end{align}
Here $\mu(s)$ is the smallest weight of a Pauli with syndrome $s$, and $\nu(A)$ is the smallest weight of a stabilizer that anticommutes with $A$.
The distance of the intermediate code determines which Pauli errors the remaining checks can detect.

\begin{theorem}[Intermediate distance]\label{thm:intermediate-distance}
    $\cD$ is an $\qcode{n}{k+1}{\delta}$ code with
    \begin{equation}
        \delta=\min\{d,\mu(s),\nu(A)\}.
    \label{eq:delta}
    \end{equation}
\end{theorem}
The proof is given in \SM{}~\ref{app:intermediate-details}; the terms are the minimum weights of the three kinds of nontrivial logical operators of $\cD$, namely the logicals of $\cC$, the stabilizers of $\cC$ that anticommute with $A$, and the Paulis with syndrome $s$. Equation~\eqref{eq:delta} shows two ways in which the intermediate distance can be smaller than $d$: a low-weight Pauli can have syndrome $s$, giving small $\mu(s)$, or a low-weight stabilizer can anticommute with $A$ and become a logical operator of $\cD$, giving small $\nu(A)$.

\subsection{Pure codes and bounded checks}

A distance-$d$ stabilizer code is \emph{pure} if every nonidentity stabilizer has weight at least $d$~\cite{calderbank1998gf4}. Purity controls $\nu(A)$, and for a  balanced factorization the %
 factor weight $\wt(B)$ 
controls $\mu(s)$, so the intermediate distance of a pure code grows with $d$. Bounded-weight checks give the opposite limit.
\begin{theorem}[Pure codes and bounded-weight checks]\label{thm:pure}\label{prop:bounded-checks}
    (a)
    Let $\cC$ be a pure stabilizer code of distance $d$, let $L$ be a minimum-weight logical, and let $(A,B)$ be the  balanced factorization of Theorem~\ref{thm:weights} with $\wt(B)=\lfloor(d+1)/2\rfloor\leq\wt(A)$. Then
    \begin{equation}
        \left\lfloor\frac d2\right\rfloor\leq\delta=\mu(s)\leq\left\lfloor\frac{d+1}{2}\right\rfloor.
    \label{eq:pure-window}
    \end{equation}
    (b) If $\cS$ is generated by stabilizer checks of weight at most $w$, then for every factorization with $s\neq0$ the intermediate distance satisfies $\delta\leq\nu(A)\leq w$.
\end{theorem}
{The proofs of parts~(a) and~(b) are given in \SM{}~\ref{app:proof-pure-bounds}.} 
Part~(b) gives an upper bound on the intermediate distance: 
for any family whose stabilizer generators have weight at most a constant $w$, including every quantum LDPC family, such as surface and color codes~\cite{kitaev2003fault,bombin2006topological}, 
this direct factorization has $\delta\leq w$ no matter how large $d$ becomes. %

The Golay and BCH examples (\SM{}~\ref{app:golay-details} and~\ref{app:bch}, respectively; see also Table~\ref{tab:split}) use higher-weight checks and still obey $\delta\leq w$.
Selective concatenation in Sec.~\ref{sec:selective} also obeys this bound for the concatenated code while increasing the protection of the qubits touched by the rotations. 
Purity is sufficient but not necessary for a  balanced factorization to satisfy Eq.~\eqref{eq:pure-window}.
For example, in \SM{}~\ref{app:purity-remark} we give an impure example with the same $\delta$ as the pure code from which it is obtained by appending a qubit in $|0\rangle$.
In \SM{}~\ref{app:bch}, Proposition~\ref{cor:bch}, we give a pure BCH family with encoding rate approaching one, dense checks, and intermediate distance three or four, after a balanced factorization.

\section{Why intermediate distance is not sufficient for fault tolerance}\label{sec:bare}

An intermediate distance $\delta$ guarantees detection of Pauli errors of weight below $\delta$ and correction of arbitrary errors on at most $\lfloor(\delta-1)/2\rfloor$ qubits at a fixed stage of the gate. 
Theorem~\ref{thm:pure} permits $\delta\geq3$ and even growth of $\delta$ with $d$, so it does not by itself show that the gadget fails to be one-fault tolerant. 
A single faulty location can nevertheless produce an error on the full support of a rotation axis, as Proposition~\ref{prop:bare-failure} and its compiled-circuit example below demonstrate.
Fault tolerance additionally requires correction after the errors propagate through all subsequent operations and measurements. Individual basis changes and CNOTs in the compiled rotation need not preserve $\cD$ during a compiled rotation; the circuit proof therefore controls propagated faults at the specified correction and measurement steps.
 When a Pauli error anticommutes with the axis of a later rotation, commuting the error through that rotation reverses the remaining angle, and an angle reversal can become an undetectable logical error.

We use a common circuit model for both protected constructions: single-qubit Clifford gates, CNOTs, $T$ or $T^\dagger$, preparations, measurements, resets and idles. Each elementary operation is a fault location, including waits during feed-forward. Classical processing is reliable. 
If $\mathcal F$ is the random set of faulty locations, then
$\Pr(R\subseteq\mathcal F)\leq p^{|R|}$ for every specified set $R$ (see Eq.~\eqref{eq:local-stochastic-budget}); a fault acts
arbitrarily on the participating qubits. %
Thus any specified $r$ elementary locations are all faulty with probability at most $p^r$; the faults need not be independent. Each fault may act arbitrarily on the qubits of that operation.
Figure~\ref{fig:parity-compilation} gives the standard decomposition of a multi-qubit Pauli rotation into single-qubit basis changes, CNOTs and a single-qubit rotation. 
The protected implementations use this elementary gate set with their specified circuits: a transversal non-Clifford layer for the fixed 22-qubit code, transversal CNOT layers with inner correction in Remark~\ref{pm:ft} and Corollary~\ref{pm:support}, and the restricted supports and ordering of Theorem~\ref{thm:golay-complete} for Golay.
The specified protected implementations preserve the verification, corrections and gate order in Theorems~\ref{thm:main22} and~\ref{thm:golay-complete}.
The Golay rotation can use the decomposition of Fig.~\ref{fig:parity-compilation} because all of its gates act within the four known rotation qubits, whose faults are covered by the monitor and recovery analysis.

The fixed 22-qubit code retains its transversal layer of fifteen physical $T^\dagger$ gates; replacing this layer by one physical $T$ gate and Clifford operations would lose the single-fault guarantee (Proposition~\ref{pm:oneT}).
We call a gadget one-fault tolerant if, whenever the number of incoming physical errors plus internal faulty locations is at most one, it implements the ideal logical gate up to a correctable output error or, for a rejected attempt, recovers the unknown input for a restart.

\begin{proposition}[Failure of the unprotected gadget]\label{prop:bare-failure}
The unprotected two-rotation gadget is not one-fault tolerant. A Pauli error $B$ occurring after $R_B(\pi/4)$ and before $R_A(\pi/2)$ 
leaves the syndrome outcome unchanged but produces the logical error $L$. 
Under a noise model in which this event has probability $\Theta(p)$, the {logical gate infidelity} is $\Theta(p)$ for any distances of $\cC$ and $\cD$.
\end{proposition}

The mechanism is the angle reversal just described: $B$ commutes with $R_B(\pi/4)$ and anticommutes with $A$, so $R_A(\pi/2)B=BR_A(-\pi/2)=iBA\,R_A(\pi/2)$, and $BA=-iL$. 

Equation~\eqref{eq:pauli-angle-propagation} gives this angle-reversal mechanism for an error occurring between arbitrary partial rotation angles.
The proof of Proposition~\ref{prop:bare-failure} is given in \SM{}~\ref{app:proof-bare-failure}.
\par

In the compiled circuit, a $Z$ fault immediately after the single-qubit $R_Z(\pi/4)$ propagates through the remaining CNOTs and basis changes to become $B$ (Fig.~\ref{fig:parity-compilation}). 
After the second rotation, this produces the undetected logical error $L$ of Proposition~\ref{prop:bare-failure}.
The gadget therefore needs a way to detect or filter faults on the support of the non-Clifford rotation before the Clifford rotation converts them into logical errors.
Section~\ref{sec:golay} gives such a filter, and the constructions of Sec.~\ref{sec:selective} avoid the exposed rotation altogether.

{In Proposition~\ref{prop:bare-failure}, the fault-free path and the $B$-fault path have the same syndrome but differ by the nontrivial logical $L$, so they cannot share any recovery channel and thus result in a logical error (see the outcome-conditioned Knill--Laflamme criterion in Theorem~\ref{thm:record-kl} of \SM{}~\ref{app:bare-details}).}
An intermediate-distance bound controls errors of limited weight at one stage, but a single circuit fault may produce an error of greater weight as it propagates.

\section{Fault-tolerant implementations by selective concatenation}\label{sec:selective}

Concatenation enhances protection by encoding a code within another code. While standard concatenation encodes each physical qubit of the original code, here we employ \emph{selective concatenation} to encode only the qubits requiring additional protection during the rotations.

In particular, we choose different inner codes according to the Pauli operators that require protection. The fixed-code construction below uses a $\qcode{15}{1}{3}$ Reed--Muller block, whose transversal $T$-type gate supplies the non-Clifford layer~\cite{steane1999quantum,anderson2014fault,paetznick2013universal}, and a two-qubit $Z$-basis repetition block; the remaining outer qubits are unencoded.

For the Steane code of Fig.~\ref{fig:surfacecode} the rotations touch outer qubits $1$, $2$, and $4$. 
\SM{}~\ref{app:selective} gives the selective concatenation and the exact parameters of the resulting  codes used before, during, and after the gate, and the protected primitives (verified-cat extraction and verified stabilizer-state teleportation~\cite{shor1996fault,divincenzo1996fault,steane1997active,gottesman1998theory,knill2005quantum,chamberland2018flag,chao2018quantum}) from which every circuit in this section is built. 
 The code used before and after the gate is specified for each construction; selective concatenation replaces chosen outer qubits by inner code blocks.
All results are single-fault guarantees under local stochastic circuit noise, in the sense of Sec.~\ref{sec:bare}.

\begin{figure}[!t]
\centering
\begin{tikzpicture}[x=\columnwidth,y=1cm,font=\small,
  box/.style={draw=black!60,rounded corners=2pt,align=center,
    minimum height=0.95cm,text width=0.265\columnwidth,
    minimum width=0.29\columnwidth,inner sep=2.5pt,
    fill=black!3},
  nc/.style={box,fill=orange!18,draw=orange!75!black},
  ec/.style={box,fill=blue!8,draw=blue!55!black},
  ms/.style={box,fill=green!12,draw=green!45!black},
  pf/.style={box,fill=red!8,draw=red!55!black},
  arr/.style={-{Stealth[length=2mm]},thick,draw=black!65}]
\node[ec] (ec0) at (0,0)
  { EC\\on $\cC'_{22}$};
\node[nc] (t) at (0.34,0)
  {$T^{\dagger\otimes15}$ on the\\Reed--Muller block};
\node[ec] (ec1) at (0.68,0)
  {retained EC on\\$\cD'_{22}$ (20 checks)};
\draw[arr] (ec0)--(t);
\draw[arr] (t)--(ec1);

\node[ms] (g) at (0,-2.2)
  {measure $\widehat G'$ (wt 6)\\$\to y=\pm1$};
\node[ms,dashed] (m) at (0.34,-2.2)
  {if $y=-1$: measure\\$\widehat M'$ (wt 10) $\to r$};
\node[ms] (h) at (0.68,-2.2)
  {measure $\widehat h'$ (wt 10)\\$\to z=\pm1$};
\draw[arr] (ec1.south) -- ++(0,-0.5)
  -- ++(-0.68,0) -- (g.north);
\draw[arr] (g)--(m);
\draw[arr] (m)--(h);

\node[pf] (pf) at (0,-4.5)
  {Pauli frame update:\\$\widehat A'$ if $z=-1$,\\
   $\widehat L'$ if $y=-1$ and $rz=-1$};
\node[ec] (ec2) at (0.34,-4.5)
  { EC\\on $\cC'_{22}$};
\node[box,fill=white] (out) at (0.68,-4.5)
  {$R_{Z_{\rm log}}(\pi/4)\,|\psi\rangle$\\on $\cC'_{22}$};
\draw[arr] (h.south) -- ++(0,-0.5)
  -- ++(-0.68,0) -- (pf.north);
\draw[arr] (pf)--(ec2);
\draw[arr] (ec2)--(out);
\end{tikzpicture}
\caption{
The protected fixed-code gate of Sec.~\ref{subsec:gate22} (\SM{}~\ref{app:combined22}). The 22-qubit data block consists of a $\qcode{15}{1}{3}$ Reed--Muller block at outer {qubit} $1$, a two-qubit $Z$-basis repetition block at {qubit} $2$, and five unencoded qubits.
Hats denote physical Pauli representatives after selective encoding, and primes denote conjugation of the outer operators by $C_B$; the measured representatives are given in Eq.~\eqref{c22:measurement-lifts}.
The only non-Clifford operation is one transversal $T$-type layer; everything after it is a protected Pauli measurement or a Pauli-frame update, and the dashed measurement is executed only when $y=-1$. Each logical measurement consists of three  measurements using verified cat states interleaved with error correction (EC) using only the stabilizer checks of the intermediate code 
and a majority vote. With at most ten cat qubits and one verification ancilla, at most 33 qubits are occupied in the serial schedule.
}
\label{fig:gate22}
\end{figure}
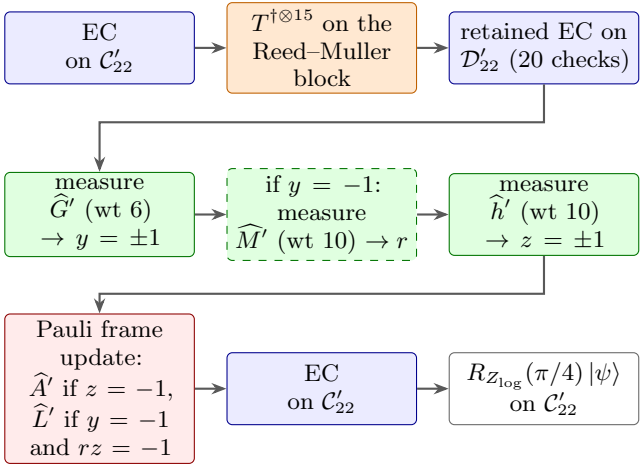

\subsection{A fixed-code gate on 22 data qubits}\label{subsec:gate22}

Our most compact construction {encodes} outer {qubit} $1$ of the Steane code, which is shared by both factors, with a Reed--Muller block, {encodes} {qubit} $2$ with a two-qubit $Z$-basis repetition code, and leaves {qubits} $3$, $4$, $5$, $6$, and $7$ unencoded.
The qubit numbering is shown in Fig.~\ref{fig:surfacecode} (see also Eq.~\eqref{pm:steane}).
The repetition check raises the distance of the retained intermediate code $\cD'_{22}$ to three. The fixed code $\cC'_{22}$ before and after the gate, and the codes obtained by fixing the measured logical operators, also have distance three, even though the repetition code alone has distance one (\SM{}~\ref{app:rep}). Two further choices remove the second Clifford rotation and the conditional logical Clifford rotation. %
First, we fix the encoding once and for all as the Clifford image $C_B\cC_0$ of the outer code in which the non-Clifford factor is $B'=Z_1$, so that the non-Clifford rotation is a single transversal $T$-type layer on the Reed--Muller block and no CNOTs are needed for that rotation. Second, we replace the  second rotation $R_A(\pi/2)$ and the conditional logical Clifford by the adaptive Pauli measurements of Sec.~\ref{subsec:identity}, so that every executed correction is a Pauli. Figure~\ref{fig:gate22} shows the resulting circuit; 
{the signed operators and explicit stabilizer generators are given in \SM{}~\ref{app:c22-signed}, and the serial cat-state schedule is specified in \SM{}~\ref{app:c22-correction}.}

\begin{theorem}[Protected fixed-code gate]\label{thm:main22}
\label{c22:ft}\label{c22:peak}
The circuit of Fig.~\ref{fig:gate22} implements the logical $T$ gate $R_{Z_{\rm log}}(\pi/4)$ on the fixed $\qcode{22}{1}{3}$ code $\cC'_{22}$, up to at most one correctable output error, whenever the number of incoming physical errors plus faulty circuit locations is at most one. It uses fifteen physical $T$-type gates, no CNOT layers for the non-Clifford rotation, and only Pauli measurements and Pauli corrections after the transversal layer. 

In a serial schedule with reset and flexible two-qubit connectivity it occupies at most $33$ qubits, including verified ancillas, and the same bound suffices to prepare a verified logical $|+\rangle$ and hence an encoded $T$ state.
\end{theorem}
The complete proof of Theorem~\ref{thm:main22}, including the 33-qubit bound and state preparation, is given in \SM{}~\ref{app:gate_single_fault_proof}.

The  code used before and after the gate, the intermediate code, and every code obtained by fixing one of the measured logical operators have distance three (Proposition~\ref{c22:distances}).
A finite set of protected Clifford gates, preparations, measurements and error-correction circuits gives arbitrarily small logical error below a positive, implementation-dependent threshold by concatenation~\cite{aliferis2006quantum}.
The proof groups each logical operation with its error-correction steps and controls failures recursively using the local child-failure rule of \SM{}~\ref{app:ft-budget} (Corollary~\ref{c22:recursion}). At level $\ell$, there are $22^\ell$ data qubits and at most $15^\ell$ physical $T$-type gates before ancilla and Clifford overhead. A logical failure requires at least $2^\ell$ faults, whereas the distance of the code before and after the gate is at least $3^\ell$.
At the first encoding level, 22 data qubits and at most 11 ancillas are used. Additional ancillas are required for the encoded operations at higher levels. 

Further protected selective-concatenation constructions are given in Remark~\ref{pm:ft} and Corollary~\ref{pm:support}; the paragraph following the proof of Corollary~\ref{cor:49} identifies its 49-qubit  code used before and after the gate with the earlier nonuniform construction~\cite{nikahd2017nonuniform,chamberland2017overhead}.

\begin{table}[!htbp]
\centering
\footnotesize
\setlength{\tabcolsep}{3pt}
\begin{tabular}{p{0.55\columnwidth}rrr}
\toprule
Construction & Data & Max. & $T$-gates\\
\midrule
Fixed code, $\qcode{22}{1}{3}$ (Sec.~\ref{subsec:gate22}, App.~\ref{app:combined22}) & 22 & 33 & 15\\
Protected Golay, $\qcode{23}{1}{7}$ (Sec.~\ref{sec:golay}, App.~\ref{app:golay-completion}) & 23 & 32 & 22\\
\bottomrule
\end{tabular}
\caption{Resource summary of the two protected circuits: data qubits (Data), maximum number of simultaneously occupied qubits including verified ancillas (Max.), and single-qubit $T$-type gates per attempt for the fully compiled circuit. 
Both rows satisfy the single-fault guarantee, including the incoming-error budget stated in Sec.~\ref{sec:bare}. The Golay count consists of one $T$-type gate for the initial rotation and 21 for its protected check measurements. The qubit counts assume serial preparation, reset, and flexible two-qubit connectivity.
}

\label{tab:main-resources}
\end{table}

\section{Protected rotations on the Golay code}\label{sec:golay}

An alternative to further encoding is to detect faults using {the retained stabilizer checks and the transported check $G_B=U_BhU_B^\dagger$, where $U_B=R_B(\pi/4)$ and $h$ is the omitted stabilizer}. 
We implement this construction on the Golay code. %
We keep one pure distance-$7$ $\qcode{23}{1}{7}$ Golay block~\cite{paetznick2012fault} and insert check measurements after the rotation;  we call the circuit for these measurements a \emph{monitor}; its check measurements follow the initial rotation, while the two cat states used to measure $G_B$ are prepared and verified beforehand. The rotation is decomposed into the elementary gates of Fig.~\ref{fig:parity-compilation}.
 
The protected gate tolerates any single circuit fault and uses at most 32 qubits and 22 single-qubit $T$-type gates per attempt in a fully compiled serial schedule (Theorem~\ref{thm:golay-complete}). 

For a  balanced factorization, we have $\mu(s)=4$ and $\nu(A)=8$.
So, the intermediate code is $\qcode{23}{2}{4}$.
The explicit Paulis $L,A,B$, the omitted check $h$, the generators, and the syndrome convention are given in \SM{}~\ref{app:golay-details}, Eq.~\eqref{eq:golay-pair}.

\subsection{The transported check}

Before the rotation, the omitted stabilizer $h$ has eigenvalue $+1$. Conjugating it by the rotation gives an observable with the same known eigenvalue on the rotated state. We call this observable the \emph{transported check} $G_B$. Together with the retained stabilizer checks, measuring $G_B$ tests membership in the rotated code space.

\begin{lemma}[Transported check]\label{lem:transported}
Let $U_B=R_B(\pi/4)$. The operator $G_B=U_BhU_B^\dagger=(h-iBh)/\sqrt2$ is Hermitian, squares to the identity, and satisfies $G_BU_B\PC=U_B\PC$: the state after the non-Clifford rotation is a $+1$ eigenstate of $G_B$, while the retained checks $\cS_0$ are unchanged.
\end{lemma}
{See the \hyperref[app:proof-transported]{proof of Lemma~\ref{lem:transported}} in \SM{}~\ref{app:golay-extension}.} %
It is a Hermitian Clifford operator that is not a Pauli: the conjugate of a Pauli by the $\pi/4$ Pauli rotation $U_B=R_B(\pi/4)$, so measuring it is equivalent to measuring $h$ in the rotated frame.

\begin{lemma}[Exact angle filter]\label{lem:angle-filter}
Let $\Pi_+=(I+G_B)/2$. For every $\varepsilon$ and $|\psi\rangle\in\cC$,
\begin{equation}
\Pi_+R_B(\pi/4+\varepsilon)|\psi\rangle=\cos(\varepsilon/2)R_B(\pi/4)|\psi\rangle.
\label{eq:angle-filter}
\end{equation}
When the measured value of $G_B$ is $+1$, the state is projected onto $R_B(\pi/4)|\psi\rangle$, with probability $\cos^2(\varepsilon/2)$.
\end{lemma}

{See the \hyperref[app:proof-angle-filter]{proof of Lemma~\ref{lem:angle-filter}} in \SM{}~\ref{app:golay-extension}.}
Among faults commuting with every retained check, Theorem~\ref{thm:local-filter} below shows that only $I$ and $B$ remain on the rotation support $\Omega_B:=\supp(B)$.
For these ideal measurements, we accept if all retained checks and $G_B$ give $+1$. The identity leaves the $G_B$ outcome at $+1$, while a $B$ error changes it to $-1$. After rejection, the unknown input is recovered and the gate is attempted once more, as described below.

\subsection{Local filter and the monitor}

The transported check detects the $B$ component of an error on the rotation support, while the retained stabilizers detect every other nonidentity Pauli component.
The next theorem establishes this division of error detection for a pure code whenever the rotation support has size less than the code distance.
\par

\begin{theorem}[Local filter]\label{thm:local-filter}
Let $(A,B)$ be a factorization of a logical Pauli with $\syn(A)=\syn(B)\neq0$. Let $\cC$ be pure with distance $d$, and let $B$ have weight $w<d$ and support $\Omega_B:=\supp(B)$. Then the restriction of $\cS_0$ to $\Omega_B$ has binary rank $2w-1$ in the symplectic Pauli representation, and the only Paulis on $\Omega_B$ commuting with every element of $\cS_0$ are $I$ and $B$. Consequently, after the rotation $R_B(\pi/4)$, every nonidentity Pauli fault supported on $\Omega_B$ is detected with certainty by measuring the retained checks and $G_B$.
\end{theorem}

{See the \hyperref[app:proof-local-filter]{proof of Theorem~\ref{thm:local-filter}} in \SM{}~\ref{app:golay-extension}.} 
Because the accepted projector annihilates each nonidentity Pauli on the rotation support, the normalized accepted state is ideal for arbitrary coherent errors and quantum channels on that support whenever acceptance has nonzero probability, and the success probability is independent of the encoded input; \SM{}~\ref{app:golay-extension} proves this extension and stronger support and recovery bounds for Golay.

The protected circuit uses all 21 retained checks. After rejection, ideal full-syndrome extraction and a decoder restricted to the modeled support recover the input even after arbitrary errors on all four qubits in the rotation support and one unknown qubit outside that support (Theorem~\ref{ge:restart}).

\subsection{A protected Golay circuit}\label{subsec:golayprotected}

The monitor protects the non-Clifford rotation, but two further ingredients are needed for a fault-tolerant gate. 
The transported check $G_B$ requires a protected extraction circuit: its ideal implementation applies $U_B^\dagger$, measures $h$, and reapplies $U_B$, which contains additional non-Clifford rotations. 
An anticommuting single-qubit fault on $\supp(A)$ between the monitor and the  second rotation $R_A(\pi/2)$ creates a fault-path ambiguity that the final syndrome cannot resolve (see Lemma~\ref{lem:bare-fault-paths}, item~\ref{fp:between-a} in \SM{}~\ref{app:proof-bare-failure}). 

 We %
 supply both ingredients without changing the 23-qubit data block.
Two verified four-qubit cat states control the transported-check circuit. A second control detects correlated faults that could otherwise change both the data and the first check outcome. After the adjacent $ZZ$ parities and the product of the $X$ outcomes of each cat state have been checked, all 21 retained stabilizers are measured on the data.
All non-Clifford blocks act on data only within the four-qubit rotation support; they may also act on the cat controls. Each terminal controlled rotation touches one data qubit. This construction applies the recursive circuit-gauge method of Dasu and Criger~\cite{dasu2026flagging} to the released Golay syndrome. If a check rejects, every single-fault branch lies in the span of Pauli errors on $\Omega_B$ and at most one qubit outside it. The original Golay syndrome distinguishes all $14\,848$ such Pauli hypotheses, so recovery restores the unknown input, including its entanglement with another system. After acceptance, error correction using the retained checks and the Pauli-measurement rule of Sec.~\ref{subsec:identity} complete the logical gate; this replaces the unprotected second rotation and removes the corresponding single-fault ambiguity.

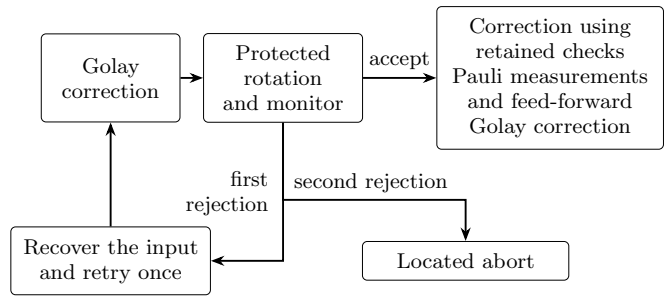
\begin{figure}[!htbp]

\centering
\resizebox{\columnwidth}{!}{%
\begin{tikzpicture}[font=\small,box/.style={draw,rounded corners=2pt,align=center,inner sep=5pt},flow/.style={-{Stealth[length=2mm]},thick}]
\node[box,text width=1.7cm,minimum height=1.3cm] (lead) at (0,0) {Golay correction};
\node[box,text width=2.0cm,minimum height=1.3cm] (monitor) at (2.55,0) {Protected rotation and monitor};
\node[box,text width=3.0cm,minimum height=1.3cm] (finish) at (6.5,0) { Correction using retained checks\\Pauli measurements\\and feed-forward\\Golay correction};
\node[box,text width=2.6cm] (recover) at (0,-2.7) {Recover the input\\and retry once};
\node[box,text width=2.7cm] (abort) at (5.25,-2.7) {Located abort};
\draw[flow] (lead) -- (monitor);
\draw[flow] (monitor) -- node[above]{accept} (finish);
\draw[flow] (monitor.south) |- (recover.east);
\node[anchor=east,align=right] at (2.45,-1.7) {first\\rejection};
\draw[flow] (recover.north) -- (lead.south);
\draw[flow] (monitor.south) -- (2.55,-1.8) -| (abort.north);
\node[above] at (3.85,-1.8) {second rejection};
\end{tikzpicture}%
}
\caption{Overview of the protected Golay logical $T$ gate on an encoded input. Acceptance is followed by  error correction using the retained checks, protected adaptive Pauli measurements and Pauli feed-forward, and final Golay correction. The first rejection triggers recovery of the original input and one retry; the second rejection gives a located abort. Exhausting an ancilla or correction retry cap, or finding no recovery hypothesis, also gives a located abort as specified in \SM{}~\ref{app:gc-control}. Figure~\ref{fig:golay-complete} gives the complete circuit and monitor sequence.}
\label{fig:golay-overview}
\par
\end{figure}

The Golay circuit summarized in Fig.~\ref{fig:golay-overview} and detailed in Fig.~\ref{fig:golay-complete}, consisting of  leading correction, the rotation and protected monitor, recovery and a possible retry,  completion by protected Pauli measurements and corrections, and trailing correction, corrects one incoming error or tolerates one internal fault (Theorem~\ref{thm:golay-complete}). 
Two gate attempts suffice in this one-fault regime, and two rejected attempts are a located abort.
A located abort is a reported unsuccessful gate invocation, as defined in Definition~\ref{def:located-failure}; its flag is computed from the recorded measurement outcomes and their circuit-step and attempt labels using the stopping rules in \SM{}~\ref{app:gc-control}. The flag identifies the failed gate invocation, not a faulty elementary operation.
With a schedule consisting of an ordered execution of elementary circuit operations, bounded by two gate attempts, two candidates per requested cat state and four rounds per correction procedure (\SM{}~\ref{app:gc-recovery}), logical failure or abort requires at least two faulty locations under local stochastic circuit noise. 
Resetting and serially reusing the eight cat-state qubits and their verification ancilla gives a maximum of 32 physical qubits and 22 physical $T$-type gates per attempt when fully compiled, or at most 44 for two attempts. 
\SM{}~\ref{app:golay-completion} gives the full circuit, retained generators, recovery, and proof. 
The monitor is protected for an encoded input in $\cC$; a protected measurement of $G_B$ on an arbitrary state of $\cD$ is not claimed. For an ideally encoded input, the finite circuit gives the analytical bound $\Pr(\text{logical failure or located abort})\leq\binom{N_G}{2}p^2$, where $N_G$ is the total number of elementary-operation and idle slots in a fixed padded schedule for the entire two-attempt circuit, as defined in \SM{}~\ref{app:gc-failure} 
(see Eq.~\eqref{gc:failure-bound}).

\section{Discussion}\label{sec:discussion}

We have shown that a syndrome degree of freedom can realize a deterministic logical $T$ gate. Releasing one stabilizer check provides an auxiliary logical qubit within the existing physical block. Two Pauli rotations couple this qubit to the encoded data, and syndrome measurement with Clifford feed-forward completes the gate and returns the state to the original code space.
During the ideal rotations, the intermediate code
determines the available error detection and correction 
by specifying the checks that remain measurable and the errors they detect or distinguish for correction. 
This links the synthesis of a logical gate directly to the protection of the spaces it visits.

The two Paulis for the rotation are given by a factorization of a logical operator. The maximally unbalanced factorization implements the same ideal logical channel as magic-state injection. 
When one factor has weight one, the gadget rotates a single data qubit by a $T$-type angle and uses syndrome projection and feed-forward to transfer that rotation to the logical qubit, which is the same ideal logical channel that injection and gate teleportation produce~\cite{gottesman1999demonstrating,zhou2000methodology,bravyi2005universal,gavriel2023transversal}. 
In contrast, a balanced factorization distributes the two rotations so that each acts on at most $\lceil(d+1)/2\rceil$ qubits.

We fully characterize the  distance of the intermediate code  via $\delta=\min\{d,\mu(s),\nu(A)\}$. For pure codes with balanced factorizations, the intermediate distance grows with $d$. 
For quantum LDPC codes, including surface and color codes~\cite{kitaev2003fault,bombin2006topological,bhardwaj2026mitten,zheng2026canonical,yang2026gala,hong2026rate} it cannot exceed the maximum weight in a stabilizer generating set, which means that increasing the code distance alone never yields an unboundedly protected intermediate code; a weight-nine check family, for instance, gives $\delta\leq9$ for every  factorization with nonzero syndrome. %
Obtaining a larger intermediate distance requires changing the implementation, for example by concatenation, as we do here, or through auxiliary encoded degrees of freedom. The high-rate BCH family shows a different trade-off: encoding rate approaching one, but dense checks and an intermediate distance that stays at three or four. Growing intermediate distance, high rate, and cheap syndrome extraction are three separate demands that no single code family in this paper meets at once. \SM{}~\ref{app:rate-design} gives the rate cost of protecting selected logical rotations and the dense-check extraction cost.

Fault tolerance also requires controlling how errors propagate through the physical rotations. 
The clearest example is a Pauli error on the axis of the first rotation, which survives to the second rotation, reverses its angle, and produces a logical error that the final syndrome cannot see. %
We use the outcome-conditioned Knill--Laflamme criterion to test whether the propagated faults admit a common recovery.
The same separation between ``detectable in the intermediate code'' and ``correctable after the whole operation'' appears in pieceable fault tolerance and in error-transparent gates~\cite{yoder2016universal,vy2013error,kapit2018error,ma2020error}, and we expect it to be useful in analyzing syndrome-resolved and continuous-angle logical operations more generally~\cite{introHu2022Channels,cheng2025emergent,yoshioka2026syndrome,huang2025robust,huang2026continuous}.

Based on our two-rotation construction, we introduce two protected gates robust against a single fault, demonstrating that our approach can be made fault-tolerant. %
The fixed $\qcode{22}{1}{3}$ gate admits arbitrarily small logical error through recursive concatenation of a finite set of protected gate circuits, with failed child procedures handled locally as specified in \SM{}~\ref{app:recursive-failures}.
Its 33-qubit serial implementation establishes a space bound for one level; reusing ancillas reduces simultaneous space at the cost of additional time. A numerical threshold or cost at a target logical error requires the complete elementary circuit schedule. The same circuit can prepare an encoded $T$ state by applying the gate to an encoded $|+\rangle$, if a resource state for gate teleportation is desired.

The Golay construction implements a protected logical $T$ gate directly on the 23-qubit data block, using at most 32 qubits and 22 single-qubit $T$-type gates per compiled attempt. A rejected single-fault attempt recovers the unknown encoded input for a restart.
The transported check is a Hermitian Clifford observable, equivalent to $(Z+X)/\sqrt2$ after a Clifford change of basis. Its accepted eigenspace contains the rotated encoded input, rather than a prescribed resource state~\cite{gidney2024cultivation,gottesman2026surviving}. With ideal, correctly calibrated extraction, acceptance removes arbitrary errors on the rotation support, including coherent over-rotation. Rejected attempts use a decoder whose allowed Pauli errors act on the four known rotation qubits and at most one unknown qubit outside them. The protected circuit handles one incoming physical error or one internal fault, including faults in preparation, measurement, reset and waits. 
The proof uses the known $G_B=+1$ value on the ideally rotated encoded input. If the rotation and its check circuit share the same angle error, the check can accept a state with the wrong rotation angle 
(\SM{}~\ref{app:golay-extension}).

Two future directions seem most promising. At the smallest scale, our two-rotation identity suggests a concrete experiment for the seven-qubit Steane code, using two weight-two rotations and one syndrome round on the same platforms that have demonstrated continuous-angle logical rotations~\cite{huang2026continuous}; such an experiment would directly observe coherent evolution between syndrome sectors followed by a deterministic logical action.
A second direction is to identify code and circuit conditions for a protected implementation with low overhead. For quantum LDPC codes, gauge degrees of freedom or measurement-based logical operations may allow constructions outside the assumptions of the direct intermediate-distance bound~\cite{cohen2022low,quintavalle2023partitioning,huang2023homomorphic,he2025extractors,xu2024fast,williamson2024low}.
Phantom codes implement in-block logical CNOTs through qubit relabeling~\cite{koh2026phantom} and potentially could reduce the cost of the CNOT network that accumulates the Pauli parity in the compiled rotations.
Theorem~\ref{thm:record-kl} determines when the measurement record permits recovery from a specified set of fault paths. It remains to identify code and circuit properties that guarantee this condition and allow the recovery itself to be implemented fault-tolerantly. Such criteria could guide selective concatenation and transported-check constructions, accounting for protection throughout the gate rather than only in the original code.

\section*{Acknowledgments}

A.T. is supported by the CQT PhD Scholarship, the Google PhD Fellowship Program and the CQT Young Researcher Career Development Grant. K.B.~is supported by a Hartree Fellowship from the Joint Center for Quantum Information and Computer Science (QuICS) at the University of Maryland, College Park. Generative artificial intelligence tools were used to assist with %
developing ideas,
language editing, organization, and the preparation of portions of the manuscript. 

\section*{Author contributions}
All authors contributed equally to the conceptualization, methodology,
formal analysis, verification, and interpretation of the results. All
authors contributed equally to drafting, reviewing, and editing the
manuscript and approved its final version.

\bibliography{references}

\clearpage
\newpage

\let\addcontentsline\oldaddcontentsline

\renewcommand{\appendixname}{\SMLong{}}
\appendix

\onecolumngrid
\newpage

\setcounter{secnumdepth}{3}
\setcounter{equation}{0}
\setcounter{figure}{0}
\setcounter{section}{0}

\renewcommand{\thesection}{\Alph{section}}
\makeatletter
\def\p@subsection{}
\def\p@subsubsection{}
\makeatother
\renewcommand{\thesubsection}{\thesection.\arabic{subsection}}
\renewcommand{\thesubsubsection}{\thesubsection.\arabic{subsubsection}}
\setcounter{tocdepth}{2}
\renewcommand*{\theHsection}{\thesection}

\clearpage
\begin{center}
\textbf{\large \SMLong{}}
\end{center}
\setcounter{equation}{0}
\setcounter{figure}{0}
\setcounter{table}{0}

\makeatletter
\renewcommand{\thefigure}{S\arabic{figure}}
\renewcommand{\thetable}{S\arabic{table}}
\renewcommand*{\theHfigure}{S\arabic{figure}}
\renewcommand*{\theHtable}{S\arabic{table}}
\begingroup
\let\l@subsubsection\@gobbletwo
\@starttoc{toc}
\endgroup
\makeatother

\section*{Appendix Guide}

The appendix follows the development of the main text: ideal gate identities and elementary rotation circuits, intermediate codes, fault propagation and recovery, selective concatenation and the protected fixed-code gate, the direct Golay gate, and finally resource bounds.
The introductory overview in Sec.~\ref{sec:intro} announces later results; the technical order follows their use in Secs.~\ref{sec:identity}--\ref{sec:discussion}, with definitions and lemmas placed before the arguments that need them.
Ideal identities and code distances are kept distinct from the circuit-level single-fault guarantees.
The table identifies the appendix material supporting each main-text result or claim.

\begin{center}
\small
\begin{tabular}{@{}p{0.36\textwidth}p{0.60\textwidth}@{}}
\toprule
Main-text result or claim & Appendix details\\
\midrule
Two-rotation identity and optimal factor weights (Sec.~\ref{sec:identity}) & \SM{}~\ref{app:core-proofs}: the general two-rotation identity and proofs of Theorems~\ref{cor:det-t} and~\ref{thm:weights}.\\
Adaptive completion by Pauli measurements (Sec.~\ref{subsec:identity}) & \SM{}~\ref{app:pm}: branch identities, outcome probabilities, Pauli corrections and the inverse-gate rule (Theorems~\ref{pm:identity} and~\ref{pm:adaptive}).\\
Elementary implementation of Pauli rotations (Sec.~\ref{sec:identity}) & \SM{}~\ref{subsec:rotation-compilation}: the parity-circuit identity, elementary gate counts and phase conventions.\\
Intermediate-code distance and code families (Sec.~\ref{sec:intermediate}) & \SM{}~\ref{app:intermediate-details}: proofs of Theorems~\ref{thm:intermediate-distance} and~\ref{thm:pure}, purity and examples; \SM{}~\ref{app:sector}: the sector representation; \SM{}~\ref{app:golay-details} and~\ref{app:bch}: Golay code data and the BCH family.\\
Failure of the unprotected circuit and the circuit-noise model (Sec.~\ref{sec:bare}) & \SM{}~\ref{app:fault-analysis}: local stochastic noise and compiled-fault bounds; \SM{}~\ref{app:bare-details}: fault paths, the outcome-conditioned recovery criterion (Theorem~\ref{thm:record-kl}) and the limitation with one unprotected $T$ gate (Proposition~\ref{pm:oneT}).\\
Selective concatenation and the protected 22-qubit gate (Sec.~\ref{sec:selective}) & \SM{}~\ref{app:selective-details}: the distance formula; \SM{}~\ref{app:component-codes}: Steane and Reed--Muller conventions; \SM{}~\ref{app:combined22}: the fixed code, protected primitives, gate proof and 33-qubit schedule ({proof of Theorem~\ref{thm:main22} in \SM{}~\ref{app:gate_single_fault_proof}}).\\
Error suppression by recursive concatenation (Sec.~\ref{subsec:gate22}) & \SM{}~\ref{app:ft-budget}: the finite set of gate circuits, truncation convention, error recurrence and resource counts (Corollary~\ref{c22:recursion}).\\
Golay filtering and input recovery (Sec.~\ref{sec:golay}) & \SM{}~\ref{app:golay-details}: code data and retained checks; \SM{}~\ref{app:golay-extension}: proofs of Lemmas~\ref{lem:transported} and~\ref{lem:angle-filter} and Theorem~\ref{thm:local-filter}, together with exact channel filtering, restart recovery and the calibration limitation.\\
Protected Golay circuit and its costs (Sec.~\ref{subsec:golayprotected}) & \SM{}~\ref{app:golay-completion}: control preparation, completion by Pauli measurements and corrections, resource counts and rejected-branch recovery, followed by the single-fault proof and failure-probability bound (Theorem~\ref{thm:golay-complete}).\\
High rate and extraction costs (Sec.~\ref{sec:discussion}) & \SM{}~\ref{app:rate-design}: the rate cost of protecting selected logical supports and the cost of directly measuring dense checks.\\
\bottomrule
\end{tabular}
\end{center}

\par

\section{Logical identities and elementary circuits}\label{app:identities}

Here we give the proofs supporting Sec.~\ref{sec:identity}.
The syndrome branches for arbitrary rotation angles give the deterministic gate of Theorem~\ref{cor:det-t}, and adaptive Pauli measurements provide its alternative completion.
A support-counting argument proves Theorem~\ref{thm:weights}; elementary rotation circuits and products of Paulis with a common syndrome complete the ideal construction.
\par

\subsection{Proofs of the two-rotation identity and factor-weight bound}\label{app:core-proofs}

\subsubsection{Branch operators and the logical rotation}

We use $\operatorname{atan2}(y,x)=\arg(x+iy)\in(-\pi,\pi]$.
For $(x,y)\neq(0,0)$, $\arg(x+iy)$ is the unique angle
$\phi\in(-\pi,\pi]$ such that
\[
x+iy=\sqrt{x^2+y^2}\,e^{i\phi}.
\]
Equivalently, $\cos\phi=x/\sqrt{x^2+y^2}$ and
$\sin\phi=y/\sqrt{x^2+y^2}$. The argument is undefined at
$(x,y)=(0,0)$; below, this case corresponds to a syndrome
outcome of probability zero.

\begin{theorem}[Two rotations]\label{thm:two-rotations}
    Let $(A,B)$ be a factorization of $L$ with $s=\syn(A)\neq0$ and let $\alpha,\beta\in\bbR$.
    For every encoded state $|\psi\rangle\in\cC$,
    \begin{equation}
    e^{-i\alpha A}e^{-i\beta B}|\psi\rangle=K_0|\psi\rangle+K_s|\psi\rangle,
    \label{eq:two-rot}
    \end{equation}
    where {$K_0$} maps $\cC$ to itself and {$K_s$} maps $\cC$ to the syndrome-$s$
    subspace, with
    \begin{equation}
        K_0=\cos\alpha\cos\beta\,I-i\sin\alpha\sin\beta\,L,\quad
        K_s=-i(\sin\alpha\cos\beta\,A+\cos\alpha\sin\beta\,B).
    \label{eq:two-rot-branches}
    \end{equation}
    Measuring the syndrome yields:
    \begin{enumerate}[label=\arabic*.]
        \item syndrome outcome $0$ with probability, independent of $|\psi\rangle$,
        \begin{equation}
            p_0=\cos^2\alpha\cos^2\beta+\sin^2\alpha\sin^2\beta.
        \label{eq:two-rot-p0}
        \end{equation}
        When $p_0>0$, the normalized post-measurement state is, up to global phase, $R_L(\theta_0)|\psi\rangle$, where
        \begin{equation}
            \theta_0=2\operatorname{atan2}
            (\sin\alpha\sin\beta,\cos\alpha\cos\beta).
        \label{eq:two-rot-theta0}
        \end{equation}
        Equivalently, 
        \begin{equation}
            \tan(\theta_0/2)=\tan\alpha\tan\beta,
        \label{eq:two-rot-tan0}
        \end{equation}
        whenever this expression is defined.
        \item syndrome outcome $s$ with probability $1-p_0$. When $1-p_0>0$, applying $A$ returns the normalized state to $\cC$ as $R_L(\theta_s)|\psi\rangle$, up to global phase, where
        \begin{equation}
            \theta_s=2\operatorname{atan2}
            (-\cos\alpha\sin\beta,\sin\alpha\cos\beta).
        \label{eq:two-rot-thetas}
        \end{equation}
        Equivalently, 
        \begin{equation}
            \tan(\theta_s/2)=-\tan\beta/\tan\alpha,
        \label{eq:two-rot-tans}
        \end{equation}
        whenever this expression is defined.
    \end{enumerate}
\end{theorem}

\begin{proof}
    Expand both exponentials,
    $e^{-i\alpha A}e^{-i\beta B}
    =(\cos\alpha-i\sin\alpha A)(\cos\beta-i\sin\beta B)$,
    and use $AB=iL$ for the last term; this gives
    Eq.~\eqref{eq:two-rot}:
    \[
    \begin{aligned}
        e^{-i\alpha A}e^{-i\beta B}
        &=(\cos\alpha\,I-i\sin\alpha\,A) (\cos\beta\,I-i\sin\beta\,B)\\
        &=\cos\alpha\cos\beta\,I -i\sin\alpha\cos\beta\,A -i\cos\alpha\sin\beta\,B -\sin\alpha\sin\beta\,AB\\
        &=\cos\alpha\cos\beta\,I-i\sin\alpha\sin\beta\,L -i(\sin\alpha\cos\beta\,A +\cos\alpha\sin\beta\,B)\\
        &=K_0+K_s.
    \end{aligned}
    \]
    This identifies the branch operators in Eq.~\eqref{eq:two-rot-branches}.
    
    The operator $L$ preserves $\cC$, whereas $A$ and $B$
    map $\cC$ to the syndrome-$s$ subspace, so the
    {vectors $K_0|\psi\rangle$ and $K_s|\psi\rangle$}
    are orthogonal.
    Indeed, $AB=iL$ and $\syn(L)=0$ imply
    $\syn(A)+\syn(B)=0$, so $\syn(B)=s$.
    For every stabilizer generator $g_j$,
    \[
    \begin{aligned}
        g_jA|\psi\rangle&=(-1)^{s_j}A|\psi\rangle,\quad
        g_jB|\psi\rangle&=(-1)^{s_j}B|\psi\rangle.
    \end{aligned}
    \]
    Since $s\neq0$, some $g_j$ has eigenvalue $-1$ on the
    syndrome-$s$ subspace and $+1$ on $\cC$. These eigenspaces
    are orthogonal, giving
    $\langle\psi|K_0^\dagger K_s|\psi\rangle=0$.
    
    For state $|\psi\rangle$, the probability of outcome $0$ is $p_0=\langle\psi|K_0^\dagger K_0|\psi\rangle$.
    Since $L^\dagger=L$ and $L^2=I$,
    \[
    \begin{aligned}
        K_0^\dagger K_0
        &=(\cos^2\alpha\cos^2\beta
          +\sin^2\alpha\sin^2\beta)I -i\cos\alpha\cos\beta\sin\alpha\sin\beta\,L +i\sin\alpha\sin\beta\cos\alpha\cos\beta\,L\\
        &=(\cos^2\alpha\cos^2\beta
          +\sin^2\alpha\sin^2\beta)I.
    \end{aligned}
    \]
    Thus the two mixed terms cancel, giving the stated $p_0$ in Eq.~\eqref{eq:two-rot-p0} independently of $|\psi\rangle$. Orthogonality of the branches and unitarity of the rotations give the other outcome probability
    \[
        1-p_0=\sin^2\alpha\cos^2\beta+\cos^2\alpha\sin^2\beta.
    \]
    
    For $p_0>0$, the definition of $\theta_0$ in Eq.~\eqref{eq:two-rot-theta0} gives
    \[
    \begin{aligned}
        \cos(\theta_0/2)&=\frac{\cos\alpha\cos\beta}{\sqrt{p_0}},\quad
        \sin(\theta_0/2)&=\frac{\sin\alpha\sin\beta}{\sqrt{p_0}}.
    \end{aligned}
    \]
    Therefore
    \[
    \begin{aligned}
        \frac{K_0}{\sqrt{p_0}}
        &=\cos(\theta_0/2)I-i\sin(\theta_0/2)L
        =R_L(\theta_0),
    \end{aligned}
    \]
    so the normalized outcome-$0$ state is $R_L(\theta_0)|\psi\rangle$.
    Taking the ratio of the sine and cosine gives Eq.~\eqref{eq:two-rot-tan0} whenever the stated expression is defined.
    The coefficient $\cos\alpha\cos\beta$ can vanish even when $p_0>0$; using $\operatorname{atan2}$ avoids dividing by this coefficient.
    
    For outcome $s$, use
    $AK_s=-i\sin\alpha\cos\beta\,I+\cos\alpha\sin\beta\,L$.
    This follows from $A^2=I$ and $AB=iL$, and gives the branch operator after the Pauli correction $A$, which returns the state to $\cC$. 
    For $1-p_0>0$, the definition of $\theta_s$ in Eq.~\eqref{eq:two-rot-thetas} gives
    \[
    \begin{aligned}
        \cos(\theta_s/2)
        &=\frac{\sin\alpha\cos\beta}{\sqrt{1-p_0}},\quad
        \sin(\theta_s/2)
        &=-\frac{\cos\alpha\sin\beta}{\sqrt{1-p_0}}.
    \end{aligned}
    \]
    Hence
    \[
    \begin{aligned}
        \frac{AK_s}{\sqrt{1-p_0}}
        &=-i\bigl[\cos(\theta_s/2)I-i\sin(\theta_s/2)L\bigr]
        =-iR_L(\theta_s).
    \end{aligned}
    \]
    The normalized corrected state is therefore
    $-iR_L(\theta_s)|\psi\rangle$, with $-i$ an irrelevant
    global phase. 
    Taking the ratio of the sine and cosine gives Eq.~\eqref{eq:two-rot-tans} whenever the stated expression is defined.
    Again, $\operatorname{atan2}$ allows
    $\sin\alpha\cos\beta=0$ whenever this outcome has
    nonzero probability.
\end{proof}

\subsubsection{{Proof of Theorem~\ref{cor:det-t}}}\label{app:proof_cor:det-t}

\begin{proof}
    For $\alpha=\pi/4$, Theorem~\ref{thm:two-rotations} gives
    $\theta_0=2\beta$ and $\theta_s=-2\beta$
    {modulo $4\pi$} for every $\beta$, and $p_0=1/2$.
    Equation~\eqref{eq:two-rot-p0} therefore gives the two probabilities in Eq.~\eqref{eq:det-t-probabilities}.
    More explicitly, for $\alpha=\pi/4$ and $\beta=\pi/8$,
    \[
    \begin{aligned}
        K_0&=\frac{1}{\sqrt2}R_L(\pi/4),\quad
        AK_s&=-\frac{i}{\sqrt2}R_L(-\pi/4).
    \end{aligned}
    \]
    Thus outcome $0$ yields the normalized state
    $R_L(\pi/4)|\psi\rangle$. Outcome $s$, after applying $A$,
    yields $-iR_L(-\pi/4)|\psi\rangle$.
    Applying the logical Clifford $R_L(\pi/2)$ to this
    latter state gives
    \[
    \begin{aligned}
        &R_L(\pi/2)\bigl[-iR_L(-\pi/4)|\psi\rangle\bigr]
        =-iR_L(\pi/4)|\psi\rangle,
    \end{aligned}
    \]
    since rotations about the same Pauli add their angles.
    Both outcomes therefore implement the same logical
    $T$-type gate, up to global phase.
    This proves the corrected logical action in Eq.~\eqref{eq:det-t-action}.
\end{proof}

\subsubsection{Adaptive Pauli measurements}\label{app:pm}

Here we will show that Pauli measurements and outcome-dependent Pauli corrections replace the  second Clifford rotation and the conditional logical Clifford introduced in Sec.~\ref{subsec:identity}.
After the non-Clifford rotation, these measurements return the state to the  code used before and after the gate with the desired logical gate applied.
These are ideal identities; circuit protection requires separate correction and measurement procedures.
\par

We use the Pauli factors $A,B$ and stabilizer group $\cS$ of Sec.~\ref{sec:identity}, and the retained group $\cS_0$ and intermediate code $\cD$ of Sec.~\ref{sec:intermediate}.
Choose $h\in\cS$ anticommuting with $A$, and put $G=iAh$, $\widetilde L=hL$, and $M=A\widetilde L=AhL$. These are Hermitian Paulis in $\cN(\cS_0)$. In logical coordinates on $\cD$,
\begin{equation}
 A=\bar X_p,\quad h=\bar Z_p,\quad G=\bar Y_p,\quad B=\bar Y_p\widetilde L,\quad M=\bar X_p\widetilde L.
\label{pm:coordinates}
\end{equation}
Indeed, $A$ anticommutes with both $h$ and $L$, so $\widetilde L$ commutes with $A,h$, and $G\widetilde L=iAL=B$. 
The input state lies in $|0\rangle_p\otimes\cC$. 
For a Hermitian Pauli $P$, define the projector $\Pi_u(P)=(I+uP)/2$, for $u\in\{+1,-1\}$.

\begin{theorem}[Measurement identity]\label{pm:identity}
Applying $R_B(\theta)$, measuring $G$ with outcome $y$, measuring $h$ with outcome $z$, and applying $A^{(1-z)/2}$ gives the operator
\begin{equation}
     A^{(1-z)/2}\Pi_z(h)\Pi_y(G)R_B(\theta)\PC
     =\frac{(iy)^{(1-z)/2}}{2}R_L(y\theta)\PC.
\label{pm:branch}
\end{equation}
Each pair $(y,z)$ has probability $1/4$, independently of the input.
\end{theorem}
\begin{proof}
Let $|\psi\rangle$ be a state in $\cC$, and write
$c=\cos(\theta/2)$ and $s=\sin(\theta/2)$.
Since $\bar Y_p|0\rangle_p=i|1\rangle_p$,
Eq.~\eqref{pm:coordinates} gives
\[
\begin{aligned}
    R_B(\theta)|0\rangle_p|\psi\rangle
    &=(cI-is\bar Y_p\widetilde L)|0\rangle_p|\psi\rangle
    =c|0\rangle_p|\psi\rangle
      +s|1\rangle_p\widetilde L|\psi\rangle.
\end{aligned}
\]
The normalized eigenstate of $G$ with eigenvalue $y$ is
$|Y_y\rangle_p=(|0\rangle_p+iy|1\rangle_p)/\sqrt2$.
Hence projection onto outcome $y$ gives the unnormalized state
\[
\begin{aligned}
    \Pi_y(G)R_B(\theta)|0\rangle_p|\psi\rangle
    &=\frac{|Y_y\rangle_p}{\sqrt2}
      (cI-iys\widetilde L)|\psi\rangle
    =\frac{|Y_y\rangle_p}{\sqrt2}
      R_{\widetilde L}(y\theta)|\psi\rangle.
\end{aligned}
\]
The last equality uses $[h,L]=0$ and $h^2=L^2=I$, so
$\widetilde L^2=(hL)^2=I$. Since $y=\pm1$,
\[
\begin{aligned}
    R_{\widetilde L}(y\theta)
    =e^{-iy\theta\widetilde L/2}
    =\cos(y\theta/2)I-i\sin(y\theta/2)\widetilde L
    =cI-iys\widetilde L.
\end{aligned}
\]
Its squared norm is $1/2$, so $\Pr(y)=1/2$.
The subsequent $h=\bar Z_p$ projection gives
\[
    \Pi_z(h)|Y_y\rangle_p=
    \begin{cases}
    |0\rangle_p/\sqrt2,&z=+1,\\
    iy|1\rangle_p/\sqrt2,&z=-1.
    \end{cases}
\]
Applying $A^{(1-z)/2}=\bar X_p^{(1-z)/2}$ therefore gives
\[
\begin{aligned}
&A^{(1-z)/2}\Pi_z(h)\Pi_y(G)R_B(\theta)
  |0\rangle_p|\psi\rangle
    =\frac{(iy)^{(1-z)/2}}{2}
  |0\rangle_pR_{\widetilde L}(y\theta)|\psi\rangle.
\end{aligned}
\]
On this $h=+1$ sector, $\widetilde L=hL$ agrees with $L$.
Since the equality holds for every input in $\cC$, we obtain Eq.~\eqref{pm:branch}.
Its scalar prefactor has squared modulus $1/4$. 
\end{proof}

\begin{theorem}[A deterministic $T$ gate with Pauli corrections]\label{pm:adaptive}
Consider the following protocol:
\begin{enumerate}
\item Apply $R_B(\pi/4)$.
\item Measure $G$, obtaining $y\in\{+1,-1\}$.
\item If $y=-1$, measure $M$, obtaining $r\in\{+1,-1\}$.
If $y=+1$, skip this measurement.
\item Measure $h$, obtaining $z\in\{+1,-1\}$, and apply
$A^{(1-z)/2}$.
\item If $y=-1$, also apply $L$ when $rz=-1$.
For $y=+1$, no additional $L$ correction is applied.
\end{enumerate}
Then, every corrected branch implements $R_L(\pi/4)$. 
The two branches with $y=+1$ have probability $1/4$ each; the four branches with $y=-1$ have probability $1/8$ each.
\end{theorem}

\begin{proof}
Let $|\psi\rangle\in\cC$.
{
{For $y=+1$, Eq.}~\eqref{pm:branch} with $\theta=\pi/4$ gives
$R_L(\pi/4)$ after $A^{(1-z)/2}$, up to a scalar factor.
Thus $A$ is still applied if $z=-1$, but no additional $L$ correction is needed.}
Conditioned on the first outcome $y$, the normalized state after measuring $G$ is $|Y_y\rangle_p|\phi_y\rangle$, where
$|\phi_y\rangle=R_{\widetilde L}(y\pi/4)|\psi\rangle$.
To calculate the optional $M$ measurement, use
$M=\bar X_p\widetilde L$ and
$\Pi_r(M)=(I+r\bar X_p\widetilde L)/2$:
\[
\begin{aligned}
    &\Pi_r(M)|Y_y\rangle_p|\phi_y\rangle
    =\frac{1}{2\sqrt2}
    \bigl(|0\rangle_p(I+iry\widetilde L) + |1\rangle_p(iyI+r\widetilde L)\bigr)
    |\phi_y\rangle.
\end{aligned}
\]
The $h$ measurement selects the $|0\rangle_p$ term for $z=+1$
and the $|1\rangle_p$ term for $z=-1$; $A^{(1-z)/2}$ then
returns that qubit to $|0\rangle_p$.
The corresponding operators on $|\phi_y\rangle$ are therefore
\begin{equation}\label{eq:pm-conditional-operators}
\begin{gathered}
     z=+1:\quad \frac{I+iry\widetilde L}{2\sqrt2}
     =\frac12 R_{\widetilde L}(-ry\pi/2)
     \\ \text{and} \\
     z=-1:\quad \frac{iyI+r\widetilde L}{2\sqrt2}
     =\frac{iy}{2}R_{\widetilde L}(ry\pi/2).
\end{gathered}
\end{equation}

For $y=-1$ the total angle is $-\pi/4+rz\pi/2$. It equals $\pi/4$ when $rz=1$. Otherwise $LR_L(-3\pi/4)=iR_L(\pi/4)$ supplies the required Pauli correction. 
The first measurement has $\Pr(y)=1/2$, as shown in the proof
of Theorem~\ref{pm:identity}.
For $y=+1$, $M$ is skipped and each $h$ outcome has conditional
probability $1/2$, since
$|\langle0|Y_+\rangle|^2=|\langle1|Y_+\rangle|^2=1/2$.
For $y=-1$, each operator in Eq.~\eqref{eq:pm-conditional-operators} is one half
of a unitary, up to a phase, so each pair $(r,z)$ has
conditional probability $1/4$.
Consequently
\[
\begin{aligned}
    \Pr(y=+1,z)
    &=\Pr(y=+1)\Pr(z\mid y=+1) =\frac12\cdot\frac12=\frac14,\\
    \Pr(y=-1,r,z)
    &=\Pr(y=-1)\Pr(r,z\mid y=-1)=\frac12\cdot\frac14=\frac18.
\end{aligned}
\]
The first equality holds for each $z\in\{+1,-1\}$,
and the second for each $(r,z)\in\{+1,-1\}^2$.
These probabilities are independent of the logical input state $|\psi\rangle$, and Pauli corrections do not change them.
\end{proof}

For the inverse logical gate, replace $R_B(\pi/4)$ by $R_B(-\pi/4)$ and apply the additional $L$ correction when $y=-1$ and $rz=+1$. All other measurements and Pauli corrections are unchanged. Indeed, the $y=-1$ branch then has angle $\pi/4+rz\pi/2$: it equals $-\pi/4$ when $rz=-1$, while multiplication by $L$ gives the same rotation up to phase when $rz=+1$. This implements $R_L(-\pi/4)$ in the same fixed encoding.

The protocol above shares a feature with gate teleportation:
the use of an auxiliary qubit, measurements, and outcome-dependent corrections to obtain the same logical gate in every branch~\cite{gottesman1999demonstrating,zhou2000methodology}.
Here the auxiliary degree of freedom is the additional logical
qubit $p$ of $\cD$, and the non-Clifford resource is supplied by
applying $R_B(\pi/4)$ directly to the data.

The final return to $\cC$ can be seen as a gauge-fixing step: regard $p$ as a gauge qubit, with stabilizer $\cS_0$ and gauge group $\langle\cS_0,A,h\rangle$.
Measuring $h$ and applying $A^{(1-z)/2}$ fixes $h=+1$, giving stabilizer $\langle\cS_0,h\rangle=\cS$ and returning to $\cC$; this is a gauge-fixing step~\cite{bombin2015gauge}.
The optional $M=\bar X_p\widetilde L$ measurement also acts on the original logical subsystem and supplies the logical angle correction derived above.

\subsubsection{Proof of Theorem~\ref{thm:weights}}

\begin{proof}
    On every qubit $L_j\propto A_jB_j$, so $\supp(L)\subseteq\supp(A)\cup\supp(B)$. 
    Two Pauli strings anticommute only if their single-qubit factors anticommute at an odd number of sites in their common support, so $|\supp A\cap\supp B|\geq1$. 
    Hence
    \[
        \wt(A)+\wt(B)=|\supp A\cup\supp B|+|\supp A\cap\supp B|
        \geq\wt(L)+1.
    \]
    This proves Eq.~\eqref{eq:factor-weight-bound}.
     The factorization pair $A=A_q\otimes\bigotimes_{j\in S_A}L_j$,
    $B=B_q\otimes\bigotimes_{j\in S_B}L_j$ from Eq.~\eqref{eq:split-pair} of Theorem~\ref{thm:weights} attains equality.
    Here $L_q$ is the single-qubit Pauli factor of $L$ on physical qubit $q$,
    $A_qB_q=iL_q$, and
    $\supp(L)=\{q\}\sqcup S_A\sqcup S_B$.
    Both factors act nontrivially on $q$, and their other
    supports are disjoint, so
    \[
    \begin{aligned}
        \wt(A)+\wt(B)
        &=(1+|S_A|)+(1+|S_B|) =\wt(L)+1.
    \end{aligned}
    \]

    Now take $L$ to be minimum-weight, with $\wt(L)=d\geq2$.
    A  balanced factorization is the above construction with
    \[
        |S_A|=\left\lceil\frac{d-1}{2}\right\rceil,\qquad
        |S_B|=\left\lfloor\frac{d-1}{2}\right\rfloor.
    \]
    Thus
    \[
    \begin{aligned}
        \wt(A)&=\left\lceil\frac{d+1}{2}\right\rceil,\quad
        \wt(B)&=\left\lfloor\frac{d+1}{2}\right\rfloor<d,
    \end{aligned}
    \]
    with $\wt(B)\leq\wt(A)$; the weights may be equal.

    Suppose for contradiction that the common syndrome is zero.
    Then $B\in\cN(\cS)$. By the definition of distance,
    \[
        d=\min\{\wt(E):E\in\cN(\cS)\setminus\cS\}.
    \]
    Since $\wt(B)<d$, $B$ cannot belong to
    $\cN(\cS)\setminus\cS$, and hence $B\in\cS$ modulo global phase.
    Then $A\in\cN(\cS)$ as well
    {because $\syn(A)=\syn(B)=0$},
    but every stabilizer commutes with every normalizer Pauli,
    contradicting $\{A,B\}=0$.
    Thus the common syndrome is nonzero, proving $\syn(A)=\syn(B)\neq0$.
    Finally, $\max\{\wt A,\wt B\}\geq\lceil(d+1)/2\rceil$ whenever $\wt A+\wt B\geq d+1$.
    Together with Eq.~\eqref{eq:factor-weight-bound}, the attaining weights in Eq.~\eqref{eq:weights} prove the stated optimality.
\end{proof}

\subsubsection{\texorpdfstring{Elementary rotation circuits}{Elementary rotation circuits}}
\label{subsec:rotation-compilation}

Here, we describe the compiled Pauli rotation gate as illustrated in Figure~{\ref{fig:parity-compilation}}.
For a weight-$w$ Pauli $P$, choose a product of single-qubit Cliffords $F$ such that $FPF^\dagger=Z_1\cdots Z_w$. Apply $F$, CNOTs from qubits $1,\ldots,w-1$ into qubit $w$, the single-qubit rotation $R_{Z_w}(\theta)$, the CNOTs in reverse order, and $F^\dagger$. This implements $R_P(\theta)$ with $2(w-1)$ CNOTs and one single-qubit rotation~\cite{cowtan2020phase,toshio2025practical}. Thus a factorization with $\wt(A)+\wt(B)=d+1$ uses $2(d-1)$ CNOTs in the two rotation circuits, before syndrome extraction and protection. 
The single-qubit $R_Z$ gates in the circuits for $R_B$ and $R_A$ have angles $\pi/4$ and $\pi/2$, respectively.

Writing $C$ for the CNOT parity map, $C^\dagger Z_wC=Z_1\cdots Z_w$ gives
\[
R_P(\theta)=F^\dagger C^\dagger R_{Z_w}(\theta)CF.
\]
For the angles $\pm\pi/4$ and $\pi/2$ used here, the elementary rotation is a $T$-type or Clifford gate up to global phase. A general angle requires the corresponding elementary $R_Z(\theta)$; the identity does not assert exact finite Clifford+$T$ synthesis for every angle. A negative sign of $P$ is absorbed into $\theta$. 
For a unitary $U$ and a phase $\phi$, the controlled operation on control qubit $a$ is
\[
C_a(e^{i\phi}U)=|0\rangle\langle0|_a\otimes I+e^{i\phi}|1\rangle\langle1|_a\otimes U.
\]
Thus the factor $e^{i\phi}$ changes the phase between the two control branches and cannot be discarded as a global phase of the controlled operation.
For example, $R_Z(\pi)=-iZ$, but a controlled-$R_Z(\pi)$ applies the factor $-i$ only when the control is $|1\rangle$. It therefore equals controlled-$Z$ followed by $S^\dagger$ on the control. These phases are retained in Eqs.~\eqref{gc:W} and~\eqref{gc:phase-polynomial}.

\subsubsection{{Products of Paulis with the same syndrome}}

\begin{remark}\label{rem:several}
    {If $m$ Paulis share a nonzero syndrome $s$, their product has syndrome $(m\bmod2)s$. Thus an odd number cannot have a logical Pauli as their product. Expanding a product of rotations about these Paulis still produces only syndrome sectors $0$ and $s$. This observation does not exclude products of an even number of Paulis; the identity and factor-weight bounds proved here concern two Paulis.}
\end{remark}

\section{Intermediate-code bounds and examples}\label{app:intermediate-details}

Here we prove the distance statements of Sec.~\ref{sec:intermediate} and specify the code examples used there.

\subsection{Logical operators and proof of Theorem~\ref{thm:intermediate-distance}}

The space $A\cC$ is stabilized by $\cS_0$ and has syndrome $s$ with respect to $\cS$, so $\cC\oplus A\cC$ is the joint $+1$ eigenspace of $\cS_0$. 
Recall that $h\in\cS\setminus\cS_0$ is the omitted check chosen in Sec.~\ref{sec:intermediate}; it anticommutes with $A$ and satisfies $\cS=\langle\cS_0,h\rangle$.
The operator $A$ swaps the two summands, while $h$ acts as $+1$ on $\cC$ and $-1$ on $A\cC$. 
Thus $\bar X_p=A$ and $\bar Z_p=h$ act on the additional logical qubit. 
For any logical Pauli representative $Q$ of $\cC$, the representative $\widetilde Q=Qh^{\epsilon_Q}$ introduced in Sec.~\ref{sec:intermediate} commutes with $A$: the factor $h^{\epsilon_Q}$ cancels precisely the commutation sign of $Q$. 
It also commutes with $h$, and multiplying original logical representatives by powers of $h$ preserves their pairwise commutation signs because $h$ commutes with every $Q\in\cN(\cS)$. Since $h=I$ on $\cC$, their actions there agree with the original logical operators. 
Let $\bar L_h$ be represented by $Lh$. It commutes with $A$ and $h$ and equals $L$ on $\cC$. {In particular, $\epsilon_L=1$ and $B=(iAh)(Lh)$, so $B$ represents $\bar Y_p\bar L_h$.} The two rotations act within $\cD$ about $\bar L_h\otimes\bar Y_p$ and $\bar X_p$, respectively.

\begin{proof}[{Proof of Theorem~\ref{thm:intermediate-distance}}]
    Every stabilizer $g\in\cS$ has the form $g=g_0h^b$, with $g_0\in\cS_0$ and $b\in\mathbb F_2$.
    Suppose that a Pauli operator $E$ commutes with all of $\cS_0$.
    If $E$ also commutes with $h$, it commutes with all of $\cS$ and has syndrome $\syn(E)=0$.
    If $E$ anticommutes with $h$, it has the same commutation sign as $A$ with every $g\in\cS$, and hence $\syn(E)=s$.
    Conversely, a Pauli with syndrome $0$ commutes with all of $\cS$, hence with $\cS_0$. 
    A Pauli $E$ with syndrome $\syn(E)=s=\syn(A)$ has the same commutation signs as $A$ with every element of $\cS$, since these signs are determined by those of the stabilizer generators. 
    It therefore commutes with every element of $\cS_0$, as $A$ does.
    Therefore
    \[
        E\in\cN(\cS_0)
        \quad\Longleftrightarrow\quad
        \syn(E)\in\{0,s\}.
    \]
    
    Since $\cN(\cS)=\ker\syn$ and $\syn(A)=s$, we can write the set of all Paulis with syndrome $s$ as {a coset of $\cN(\cS)$ in the Pauli group $\mathcal{P}_n$ modulo global phases}
    \begin{equation}
        A\cN(\cS)=\{E\in\cP_n:\syn(E)=s\}.
    \label{eq:syndrome-s-coset}
    \end{equation}
    Thus $\cN(\cS_0)=\cN(\cS)\cup A\cN(\cS)$, where $\cN(\cS), A\cN(\cS)$ are disjoint. 
    The nontrivial logical Pauli representatives of $\cD$ are the elements of $\cN(\cS_0)\setminus\cS_0$.
    Since $\cS_0\subseteq\cS\subseteq\cN(\cS)$, subtracting
    $\cS_0$ gives
    \begin{equation}
    \begin{aligned}
        \cN(\cS_0)\setminus\cS_0
        ={}&(\cN(\cS)\setminus\cS) \sqcup(\cS\setminus\cS_0) \sqcup A\cN(\cS).
    \end{aligned}
    \label{eq:intermediate-logical-sets}
    \end{equation}
    
    The nontrivial logical Paulis of $\cD$ fall into three sets: $\cN(\cS)\setminus\cS$, with minimum weight $d$; $\cS\setminus\cS_0$, with minimum weight $\nu(A)$; and the coset $A\cN(\cS)$ in Eq.~\eqref{eq:syndrome-s-coset}, with minimum weight $\mu(s)$.
    Also, each set in Eq.~\eqref{eq:intermediate-logical-sets}
    is nonempty, since
    \[
    \begin{aligned}
        L&\in\cN(\cS)\setminus\cS,\quad
        h\in\cS\setminus\cS_0,\quad
        A\in A\cN(\cS).
    \end{aligned}
    \]
    The three minimum weights are $d$, $\nu(A)$, and $\mu(s)$ by the definition of the distance of $\cC$ and Eqs.~\eqref{eq:nu} and \eqref{eq:mu}, respectively.
    Their minimum weight is therefore $\min\{d,\mu(s),\nu(A)\}$, proving Eq.~\eqref{eq:delta}.
\end{proof}

\subsection{Structure of the intermediate code and the sector representation}\label{app:sector}

Let $\widetilde g_1,\ldots,\widetilde g_{r_{\cS}-1}$ be independent generators of $\cS_0$. 
Choose Hermitian destabilizers $\widetilde D_1,\ldots,\widetilde D_{r_{\cS}-1}$ satisfying $\{\widetilde D_j,\widetilde g_i\}=0$ iff $i=j$, and choose them to commute with logical pairs $\{(\bar X_\ell,\bar Z_\ell)\}_{\ell=1}^k$ and $(A,h)$ of $\cD$. 
For each logical Pauli representative $Q$ of the original code $\cC$, use $Q$ if it commutes with $A$, and $hQ$ otherwise. These representatives commute with $A,h$ and define the logical pairs $(\bar X_\ell,\bar Z_\ell)$ on $\cD$. In particular, $L$ is replaced by $hL=\widetilde L$. 
For $e\in\mathbb F_2^{r_{\cS}-1}$  let $T_e=\prod_j\widetilde D_j^{e_j}$. Every physical state has a unique decomposition
\begin{equation}
|\Psi\rangle=\sum_{e\in\mathbb F_2^{r_{\cS}-1}}T_e|v_e\rangle_{\cD},
\qquad v_e\in\mathbb C^{2^{k+1}}.
\label{eq:sector-decomposition}
\end{equation}
For a Pauli $E$, set $e'=e+\syn_{\cS_0}(E)$. Then $N=T_{e'}^\dagger ET_e\in\cN(\cS_0)$ and, up to an element of $\cS_0$,
\begin{equation}
N=i^\eta\prod_{\ell=1}^k\bar X_\ell^{a_\ell}\bar Z_\ell^{b_\ell}A^ch^d,
\label{eq:path-logical-basis}
\end{equation}
where $\eta\in\bbZ_4$ and the binary exponents follow from the symplectic commutation relations with the chosen logical pairs.
There are $2^{r_{\cS}-1}$ sectors, each of dimension $2^{k+1}$, so their total dimension is $2^{r_{\cS}-1}2^{k+1}=2^n$.

\subsection{Proofs of the pure-code and bounded-check bounds}\label{app:proof-pure-bounds}

\begin{proof}[Proof of Theorem~\ref{thm:pure}(a)]
    Every element of $\cS\setminus\cS_0$ is a nonidentity stabilizer, so purity and Eq.~\eqref{eq:nu} give $\nu(A)\geq d$.
    Equation~\eqref{eq:delta} therefore reduces to
    \[
        \delta=\min\{d,\mu(s),\nu(A)\}=\min\{d,\mu(s)\},
    \]
    since $\nu(A)$ cannot lower the minimum below $d$.
    Moreover, $\syn(B)=s$ by Eq.~\eqref{eq:common-syn}, so $B$ is one of the Paulis over which the minimum in Eq.~\eqref{eq:mu}
    is taken. Hence
    \[
        \mu(s)\leq\wt(B)
        =\left\lfloor\frac{d+1}{2}\right\rfloor,
    \]
    where the equality follows from the  balanced factorization in
    Eq.~\eqref{eq:weights}.
    Consequently $\delta=\mu(s)\leq\lfloor(d+1)/2\rfloor$.

    Now, consider Pauli $E$ with syndrome $\syn(E)=s$ that differs from $B$ modulo phase.
    Syndrome additivity gives
    \[
        \syn(EB)=\syn(E)+\syn(B)=s+s=0,
    \]
    so $EB\in\cN(\cS)$.
    Since $E\neq B$ modulo phase and $B^2=I$, then $EB\neq I$ modulo phase.
    If $EB\in\cN(\cS)\setminus\cS$, then $\wt(EB)\geq d$ by the definition of the code distance.
    If $EB\in\cS$, then $\wt(EB)\geq d$ by purity, because $EB$ is a nonidentity stabilizer modulo phase.
    
    Since $\wt(EB)\leq\wt(E)+\wt(B)$, we obtain $\wt(E)\geq d-\lfloor(d+1)/2\rfloor=\lfloor d/2\rfloor$. 
    If $E$ equals $B$ modulo phase, then $\wt(E)=\wt(B)$, which satisfies the same lower bound.
    Thus every Pauli $E$ with syndrome $s$ has weight at least $\lfloor d/2\rfloor$, so Eq.~\eqref{eq:mu} gives $\mu(s)\geq\lfloor d/2\rfloor$.
    Combining this with $\delta=\mu(s)$ and the upper bound above yields
    \[
        \left\lfloor\frac d2\right\rfloor
        \leq\delta=\mu(s)
        \leq\left\lfloor\frac{d+1}{2}\right\rfloor,
    \]
    which is Eq.~\eqref{eq:pure-window}.
\end{proof}

\begin{proof}[Proof of Theorem~\ref{prop:bounded-checks}(b)]
    Since $s\neq0$, some chosen stabilizer generator $g$ anticommutes with $A$.
    Note that $g$ belongs to $\cS\setminus\cS_0$ and has weight at most $w$.
    Since $g\in\cS\setminus\cS_0$, it is one of the stabilizers
    over which the minimum in Eq.~\eqref{eq:nu} is taken, so
    $\nu(A)\leq\wt(g)$.
    Also, Eq.~\eqref{eq:delta} gives
    $\delta=\min\{d,\mu(s),\nu(A)\}\leq\nu(A)$.
    Together with $\wt(g)\leq w$, these inequalities give
    \[
        \delta\leq\nu(A)\leq\wt(g)\leq w,
    \]
    which completes the proof.
\end{proof}

\subsection{Purity and the intermediate distance}\label{app:purity-remark}

Theorem~\ref{thm:intermediate-distance} gives the following necessary and sufficient condition for the intermediate code to correct every error on at most $t$ qubits:
\begin{equation}
d\geq2t+1,\qquad \mu(s)\geq2t+1,\qquad \nu(A)\geq2t+1.
\label{eq:intermediate-design-criterion}
\end{equation}
This concerns the full error space on each such support; fault propagation through the gate requires a separate analysis. 
For a pure code with $d\geq2t+1$, every nonidentity stabilizer has weight at least $d$, so $\nu(A)\geq d\geq2t+1$. More generally, Eq.~\eqref{eq:nu} minimizes only over $\cS\setminus\cS_0$, the stabilizers that anticommute with $A$.
A stabilizer commuting with $A$ belongs to $\cS_0$, so it is excluded from this minimization and remains a stabilizer of $\cD$.
For example, append $|0\rangle$ to a pure code with $d\geq2$ and let $A,B$ act trivially on the appended qubit. The resulting code is impure because $Z_{n+1}$ is a weight-one stabilizer, whereas its intermediate code is $\cD\otimes|0\rangle$ and has the same distance as $\cD$. Any normalizer Pauli has $I$ or $Z$ on that qubit, and multiplying by $Z_{n+1}$ removes the latter without changing the logical action.

\subsection{Golay code and retained checks}\label{app:golay-details}

Here, we consider the CSS code of the self-orthogonal $[23,11,8]$ code spanned by the eleven octads
\begin{align*}
O_1&=\{1,2,9,14,15,16,17,21\},\\
O_2&=\{1,5,7,8,10,17,20,23\},\\
O_3&=\{6,9,12,13,17,19,20,22\},\\
O_4&=\{6,8,11,14,15,16,17,18\},\\
O_5&=\{1,2,3,4,15,17,20,23\},\\
O_6&=\{2,7,8,11,12,15,18,22\},\\
O_7&=\{3,5,9,10,12,15,16,21\},\\
O_8&=\{2,3,6,9,13,16,21,22\},\\
O_9&=\{4,5,8,10,11,14,17,19\},\\
O_{10}&=\{5,15,16,17,18,20,21,23\},\\
O_{11}&=\{1,6,16,17,18,19,21,22\},
\end{align*}
with checks $X_{O_j}$ and $Z_{O_j}$, $j=1,\ldots,11$; its dual is the $[23,12,7]$ Golay code. 
Let $H$ be the $11\times23$ binary matrix with $H_{jq}=1$ if $q\in O_j$ and $0$ otherwise. Its rows generate the self-orthogonal $[23,11,8]$ code $C$, and $H$ is a parity-check matrix of its dual $C^\perp$, the $[23,12,7]$ Golay code. The CSS check matrices are $H_X=H_Z=H$. Since $HH^{\mathsf T}=0$, the checks
\[
X_{O_j}=\prod_{q\in O_j}X_q,\qquad Z_{O_j}=\prod_{q\in O_j}Z_q
\]
commute. Their $22$ independent generators encode $23-22=1$ logical qubit. A generator matrix of $C^\perp$ is a parity-check matrix of $C$.

Let $u=(1,\ldots,1)\in\mathbb F_2^{23}$ be a row vector.
Each octad has even size, so $Hu^{\mathsf T}=0$, whereas $u\notin C$ because every string in $C$ has even weight.
Thus $C^\perp=C+\langle u\rangle$, and a parity-check matrix of $C$ is
\[
H_C=\begin{pmatrix}H\\u\end{pmatrix}.
\]
The quantum code is the simultaneous $+1$ eigenspace of all $X_{O_j}$ and $Z_{O_j}$, with encoded basis
\[
|b\rangle_{\rm Golay}=\frac{1}{\sqrt{|C|}}\sum_{c\in C}|c+bu\rangle,\qquad b\in\{0,1\}.
\]
The $Z$ checks restrict computational-basis strings to $C^\perp$, while the $X$ checks require equal amplitudes within each coset of $C$.
Since $C^\perp$ has minimum distance seven and $C$ has minimum distance eight, the minimum weight in $C^\perp\setminus C$ is seven, giving the quantum parameters $\qcode{23}{1}{7}$.
\par

We take $L=Z_{\{1,10,12,13,14,15,21\}}$ (a weight-$7$ nontrivial logical Pauli), and the  balanced factorization
\begin{equation}
\begin{aligned}
A&=X_1Z_{10}Z_{12}Z_{13}, & B&=Y_1Z_{14}Z_{15}Z_{21},
\end{aligned}
\label{eq:golay-pair}
\end{equation}
so that $AB=iL$, where $A,B$ has a common syndrome $s=(01000101100\mid11001000001)$ (bits ordered $X_{O_1},\ldots,X_{O_{11}},Z_{O_1},\ldots,Z_{O_{11}}$). Exhaustive search gives $\mu(s)=4$ (attained by $A$) and $\nu(A)=8$, so $\cD$ is a $\qcode{23}{2}{4}$ code: the retained checks detect every Pauli of weight at most three and permit correction of one physical error between completed operations. The omitted check and rotation are
\begin{equation}
\begin{aligned}
h&=Z_1Z_2Z_9Z_{14}Z_{15}Z_{16}Z_{17}Z_{21},\\
U&=R_B(\pi/4),
\end{aligned}
\label{gc:operators}
\end{equation}
for $\Omega_B=\supp(B)=\{1,14,15,21\}$.
Thus $\{h,B\}=\{h,A\}=0$ and $h\PC=\PC$. 
Table~\mbox{\ref{tab:gc-retained}} gives one set of $21$ independent retained generators, all of weight eight. 
The protected circuit uses this complete retained syndrome.

\begin{table}[t]
\centering
\caption{An independent generating set for the full retained Golay stabilizer. Each row denotes the indicated Pauli on its listed support, with positive sign. Adding $h$ from Eq.~\eqref{gc:operators} gives a generating set for the original Golay code.}
\label{tab:gc-retained}
\begin{tabular}{cl}
\hline
Pauli&Support\\\hline
$X$&$1,2,9,14,15,16,17,21$\\
$Y$&$1,5,7,8,10,17,20,23$\\
$X$&$6,9,12,13,17,19,20,22$\\
$Z$&$6,9,12,13,17,19,20,22$\\
$X$&$6,8,11,14,15,16,17,18$\\
$Z$&$6,8,11,14,15,16,17,18$\\
$X$&$1,2,3,4,15,17,20,23$\\
$Z$&$3,4,9,14,16,20,21,23$\\
$Z$&$2,3,4,5,7,8,10,15$\\
$Z$&$2,7,8,11,12,15,18,22$\\
$Y$&$1,6,7,9,12,15,21,22$\\
$X$&$3,5,9,10,12,15,16,21$\\
$Z$&$3,5,9,10,12,15,16,21$\\
$Z$&$2,3,6,9,13,16,21,22$\\
$Y$&$1,3,6,13,14,15,17,22$\\
$Z$&$4,5,8,10,11,14,17,19$\\
$X$&$1,4,7,11,14,19,20,23$\\
$X$&$5,15,16,17,18,20,21,23$\\
$Z$&$5,15,16,17,18,20,21,23$\\
$X$&$1,6,16,17,18,19,21,22$\\
$Z$&$2,6,9,14,15,18,19,22$\\\hline
\end{tabular}
\end{table}

\subsection{High-rate BCH family}\label{app:bch}

Here, we establish the high-rate example stated at the end of Sec.~\ref{sec:intermediate}.
The construction gives a pure CSS family with asymptotic encoding rate approaching one at fixed distance $d=7$ (see Proposition~{\ref{cor:bch}}); applying Theorem~\ref{thm:pure}(a) to a  balanced factorization of a minimum-weight logical Pauli gives $\delta\in\{3,4\}$.
The weights of its checks and the cost of syndrome extraction are discussed in \SM{}~\mbox{\ref{app:rate-design}}.
This code-family construction does not by itself supply a protected circuit or a family-wide logical gate set.

We consider the classical binary BCH code family $\{C_m\}_m$~\cite{aly2007quantum} to construct this CSS family, where the vectors in $C_m^\perp$ specify both its $X$- and $Z$-type stabilizer checks.
We say that a classical binary BCH code family $\{C_m\}_m$ is \emph{primitive} when the block length is $n=2^m-1$.
Let $\alpha$ be a primitive element of $\mathbb F_{2^m}$. A primitive narrow-sense binary BCH code of designed distance $d_{\mathrm{des}}$ is the cyclic code whose generator polynomial is the least common multiple over $\mathbb F_2$ of the minimal polynomials of $\alpha,\alpha^2,\ldots,\alpha^{d_{\mathrm{des}}-1}$.
Here $d_{\mathrm{des}}=7$, so $C_m$ consists of the binary coefficient vectors of
\[
c(z)=\sum_{i=0}^{n-1}c_{\alpha^i}z^i,
\qquad c(\alpha^j)=0\quad(j=1,\ldots,6).
\]
The designed distance specifies the prescribed
consecutive zeros, whereas the BCH bound guarantees $d(C_m)\geq d_{\mathrm{des}}$, but equality need not hold in general. 
For this family, equality will follow from an explicit weight-seven string, as we will establish in Proposition~{\ref{cor:bch}} below.
For the BCH examples, the coordinate convention uses the primitive polynomial $x^5+x^2+1$ for $m=5$ and $x^6+x+1$ for $m=6$, with coordinate $j$ corresponding to $\alpha^{j-1}$.

    The defining zeros impose $c(\alpha^j)=\sum_x c_xx^j=0$
    for $j=1,\ldots,6$, where the sum is over $\mathbb F_{2^m}^{\times}=\mathbb F_{2^m}\setminus\{0\}$.
    These are linear constraints
    on the bits $c_x$, with sums evaluated in $\mathbb F_{2^m}$.
    Since $c_x^2=c_x$, squaring the condition at exponent $j$
    gives the condition at exponent $2j$. Thus the conditions
    at exponents $2,4,6$ follow by squaring those at $1,2,3$,
    respectively, and membership in $C_m$ is equivalent to
    \begin{equation}
        \sum_x c_xx=\sum_x c_xx^3=\sum_x c_xx^5=0.
        \label{eq:bch-field-checks}
    \end{equation}
    \par
    For $m\geq5$, the dimension formula $\dim C_m=n-3m$ and dual containment
    $C_m^\perp\subseteq C_m$ follow from Theorem~10 and Corollary~6 of Ref.~\cite{aly2007quantum},
    respectively, with $q=2$ and designed distance $7$.
    For $m\geq5$, their designed-distance conditions hold because
    $7\leq2^{\lceil m/2\rceil}-1$. The dimension formula specializes
    to $n-m\lceil6(1-1/2)\rceil=n-3m$.
To construct the CSS code parity check matrices $H_X,H_Z$ explicitly, fix an $\mathbb F_2$-basis
of $\mathbb F_{2^m}$ and let $\mathbf v(x)\in\mathbb F_2^m$ be
the coordinate column of $x$ in this basis. Define
$H_m\in\mathbb F_2^{3m\times n}$ by its column indexed by
$x\in\mathbb F_{2^m}\setminus\{0\}$:
\[
    (H_m)_{\cdot,x}=
    \begin{pmatrix}
        \mathbf v(x)\\
        \mathbf v(x^3)\\
        \mathbf v(x^5)
    \end{pmatrix}.
\]
The field parity conditions in Eq.~\eqref{eq:bch-field-checks} give $C_m=\ker H_m$, so the binary row space of $H_m$ is $C_m^\perp$. Take
\[
    H_X=H_Z=H_m.
\]
For each row $u=(u_x)_x$ of $H_m$, the corresponding stabilizer
generators are $\prod_xX_x^{u_x}$ and $\prod_xZ_x^{u_x}$.
The dual containment $C_m^\perp\subseteq C_m$ established above gives
 $H_XH_Z^T=H_mH_m^T=0$, so these generators commute.
\par

\begin{proposition}[BCH family]\label{cor:bch}
    For every $m\geq5$, the primitive narrow-sense binary BCH code of length $2^m-1$ and designed distance $7$ contains its dual, and the associated CSS code is a pure $\qcode{2^m-1}{2^m-1-6m}{7}$ code. 
    Every minimum-weight logical Pauli admits a  balanced factorization into two weight-$4$ factors with intermediate distance $\delta\in\{3,4\}$, and the rate approaches one while the distance of the CSS code before and after the gate remains seven.
\end{proposition}

\begin{proof}
    Write $n=2^m-1$ and index the binary coordinates by $x\in\mathbb F_{2^m}^{\times}$ 
    , where $\mathbb F_{2^m}^{\times}=\mathbb F_{2^m}\setminus\{0\}$ denotes the nonzero field elements. 
    Since $C_m=\ker H_m$, we have $\rank H_m=n-\dim C_m=3m$.
    The $X$- and $Z$-type generators impose $3m$ independent
    constraints each, so the number of encoded qubits is
    \[
        K=n-\rank H_X-\rank H_Z=n-6m.
    \]
    \par
    
    Choose any three-dimensional binary subspace $U\subseteq\mathbb F_{2^m}$ and let $c$ be the indicator of $U\setminus\{0\}$.
    Explicitly, $c_x=1$ when $x\in U\setminus\{0\}$ and $c_x=0$ otherwise.
    Choose an $\mathbb F_2$-basis $u_1,u_2,u_3$ of $U$.
    Each $x\in U$ is uniquely
    \[
        x=t_1u_1+t_2u_2+t_3u_3,
        \qquad(t_1,t_2,t_3)\in\{0,1\}^3.
    \]
    These three bits are the binary coordinates. Since
    \[
        x^2=\sum_{i=1}^3t_i u_i^2,\qquad
        x^4=\sum_{i=1}^3t_i u_i^4,
    \]
    the expressions $x$, $x^3=x\,x^2$, and $x^5=x\,x^4$
    are polynomials of total degree at most two in $t_1,t_2,t_3$,
    with coefficients in $\mathbb F_{2^m}$.
    \par
    Every monomial omits at least one of the three bits.
    In the sum over all eight triples in $\{0,1\}^3$, pair
    triples differing only in an omitted bit. Each pair
    contributes the same field element twice, which is zero
    in characteristic two. Consequently
    \[
        \sum_{x\in U}x=\sum_{x\in U}x^3=\sum_{x\in U}x^5=0.
    \]
    The $x=0$ term vanishes, so the indicator vector satisfies
    Eq.~\eqref{eq:bch-field-checks}.
    \par
    Therefore $c\in C_m$ has weight seven
    , since $|U|=2^3$ and the indicator of $U\setminus\{0\}$ has $|U|-1=7$ entries equal to one.
    The BCH bound gives $d(C_m)\geq7$, hence $d(C_m)=7$.
    
    Every dual string is the evaluation on $\mathbb F_{2^m}^{\times}$ of a Boolean function $f(x)=\operatorname{Tr}(ax+bx^3+cx^5)$ with $a,b,c\in\mathbb F_{2^m}$, where $\operatorname{Tr}$ is the field trace to $\mathbb F_2$.
    To see this, recall that $C_m^\perp$ is the binary row space of $H_m$. 
    By nondegeneracy of the trace pairing, every $\mathbb F_2$-linear functional on $\mathbb F_{2^m}$ has the form $y\mapsto\operatorname{Tr}(ay)$ for some $a\in\mathbb F_{2^m}$. 
    Indeed, the trace pairing is
    $(a,y)\mapsto\operatorname{Tr}(ay)$. The trace polynomial
    $z+z^2+\cdots+z^{2^{m-1}}$ is nonzero and has degree below
    $2^m$, so it cannot vanish on the entire field. Hence some
    $z$ has $\operatorname{Tr}(z)=1$. For any $a\ne0$, taking
    $y=a^{-1}z$ gives $\operatorname{Tr}(ay)=1$. Thus the map
    sending $a$ to the functional
    $y\mapsto\operatorname{Tr}(ay)$ has trivial kernel.
    Its domain and the space of binary linear functionals both
    have dimension $m$ over $\mathbb F_2$, so this map is also
    surjective.
    Thus the binary row combinations of the $m$ coordinate checks for exponent $j$ are precisely the vectors $(\operatorname{Tr}(ax^j))_x$. 
    They impose the binary parity constraints
    \[
        \sum_x c_x\operatorname{Tr}(ax^j)
        =\operatorname{Tr}\!\left(a\sum_x c_xx^j\right)=0.
    \]
    Adding the row combinations for $j=1,3,5$ gives exactly the stated evaluation vectors of $f$.
    \par
    In any binary field basis, $f$ has degree at most two and $f(0)=0$. 
    Here the variables are the $m$ binary coordinates of $x$; the Frobenius maps and trace are $\mathbb F_2$-linear, while $x^3$ and $x^5$ are products of two linear expressions.
    A nonzero degree-at-most-two Boolean function on $m$ bits has weight at least $2^{m-2}$, the minimum-distance bound for $\RM(2,m)$~\cite[Ch.~13, Theorem~3]{macwilliams1977theory}.
    The evaluation vector of $f$ on all $2^m$ field elements belongs to $\RM(2,m)$. For a nonzero dual string, this extended vector is nonzero and therefore has weight at least $2^{m-2}$.
    Removing its coordinate $f(0)=0$ does not change that weight.
    Taking the minimum over nonzero dual strings gives
    \[
        d(C_m^\perp)\geq2^{m-2}\geq8,
    \]
    where the last inequality uses $m\geq5$.
    \par
    For binary vectors $u,v$, the Pauli $\prod_xX_x^{u_x}Z_x^{v_x}$ has support $\supp(u)\cup\supp(v)$. Since $H_X=H_Z=H_m$, it normalizes the CSS stabilizer exactly when $u,v\in C_m$, and is a stabilizer exactly when $u,v\in C_m^\perp$, with global phases ignored. 
    A nonidentity normalizer therefore has
    weight at least $d(C_m)=7$, and a nonidentity stabilizer
    has weight at least $d(C_m^\perp)\geq8$. The weight-seven
    vector $c$ constructed above belongs to
    $C_m\setminus C_m^\perp$, because no nonzero dual string
    has weight seven. Hence $\prod_{x:c_x=1}Z_x$ is a
    weight-seven nontrivial logical Pauli.
    The quantum code is therefore pure with distance seven. 
    Theorem~\ref{thm:pure} then gives $3\leq\delta\leq4$ for a  balanced factorization.
\end{proof}

This is a code-family statement. The intermediate-distance bound applies to a chosen minimum-weight logical Pauli; it does not supply a protected circuit or a family-wide logical gate set.  
These distance statements alone do not bound the number of gates for syndrome extraction or the cost of encoding the supports of several chosen logical Pauli rotations; \SM{}~\ref{app:rate-design} analyzes these costs.
More precisely, Eq.~\eqref{eq:bch-check-interactions} bounds direct data--ancilla couplings, while Eq.~\eqref{eq:active-union-rate} and Proposition~\ref{prop:logical-access-rate} quantify the additional data qubits needed to encode several logical supports.

\subsection{Examples}

Table~\mbox{\ref{tab:split}} compares the Steane and Golay examples with the impure Shor code and the length-$31$ BCH code. The Shor example has intermediate distance one; the Steane example has intermediate distance two~\cite{shor1995scheme,steane1996error,steane1996simple}.
The table records the intermediate-code parameters for the chosen factorization; the family-wide BCH argument does not depend on a finite search.
The BCH family proof above applies to every $m\geq5$, independently of these finite checks.

\begin{table}[h]
\centering
\caption{Factorizations $AB=iL$. Here $\mu(s)$ is the minimum weight of a Pauli with syndrome $s$, $\nu(A)$ is the minimum weight of a stabilizer anticommuting with $A$, and $\delta$ is the intermediate-code distance.}
\label{tab:split}
\small
\begin{tabular}{@{}l l l l l c c c@{}}
\toprule
Code & $L$ & $A$ & $B$ & $\cD$ & $\mu(s)$ & $\nu(A)$ & $\delta$\\
\midrule
$\qcode{7}{1}{3}$ Steane & $Z_1Z_2Z_4$ & $X_1Z_2$ & $Y_1Z_4$ & $\qcode{7}{2}{2}$ & 2 & 4 & 2\\
$\qcode{9}{1}{3}$ Shor & $X_1X_2X_3$ & $Z_1X_2$ & $-Y_1X_3$ & $\qcode{9}{2}{1}$ & 1 & 2 & 1\\
$\qcode{23}{1}{7}$ Golay & $Z_{\{1,10,12,13,14,15,21\}}$ & $X_1Z_{10}Z_{12}Z_{13}$ & $Y_1Z_{14}Z_{15}Z_{21}$ & $\qcode{23}{2}{4}$ & 4 & 8 & 4\\
$\qcode{31}{1}{7}$ BCH & $Z_{\{1,10,12,14,15,17,20\}}$ & $X_1Z_{10}Z_{12}Z_{14}$ & $Y_1Z_{15}Z_{17}Z_{20}$ & $\qcode{31}{2}{4}$ & 4 & 8 & 4\\
\bottomrule
\end{tabular}
\end{table}

\section{Fault propagation and recovery}\label{app:fault-analysis}

Here we give the circuit-noise model and the fault-propagation and recovery arguments supporting Sec.~\ref{sec:bare}.
The outcome-conditioned recovery criterion precedes its application to the unprotected gadget.
We also state the support and probability bounds needed when the protected constructions are compiled into elementary gates.
\par

\subsection{{Local circuit noise}}
We consider local stochastic circuit noise: for every specified set $R$ of elementary locations,
\begin{equation}
 \Pr(R\subseteq\mathcal F)\leq p^{|R|},
\label{eq:local-stochastic-budget}
\end{equation}
where $\mathcal F$ is the set of faulty locations. Each fault acts only on the qubits participating at that location. Preparations, measurements, gates, resets, and idles are included; classical processing is reliable. All-to-all two-qubit connectivity is assumed. 

\subsection{Fault paths and outcome-dependent recovery}\label{app:bare-details}

\subsubsection{Propagation through a Pauli rotation}

If a Pauli error $E$ occurs after angle $\theta_1$ and before the remaining angle $\theta_2$ of a rotation about $Q\in\{A,B\}$, then
\begin{equation}
    R_Q(\theta_2)E R_Q(\theta_1)=
    \begin{cases}
    E R_Q(\theta_1+\theta_2), &[E,Q]=0,\\
    E R_Q(\theta_1-\theta_2), &\{E,Q\}=0.
    \end{cases}
\label{eq:pauli-angle-propagation}
\end{equation}
A commuting Pauli error preserves the rotation angle in the sense of Refs.~\cite{vy2013error,kapit2018error,ma2020error}, whereas an anticommuting error reverses the remaining angle.

\subsubsection{Outcome-conditioned Knill--Laflamme criterion}

The relevant recovery condition concerns the propagated fault paths sharing a complete measurement record, rather than only error weights at an intermediate stage.
The outcome-conditioned variant of the Knill--Laflamme criterion~\cite{knill1997theory} tests whether those paths admit a common recovery.

\begin{theorem}[Outcome-conditioned Knill--Laflamme]\label{thm:record-kl}
Let $U:\cC\to\mathcal H_{\rm out}$ be the desired isometry, so $U^\dagger U=\PC$ when extended by zero outside $\cC$.
For a gate returning to $\cC$, $U$ is its logical unitary; for input recovery, $U$ is the identity on $\cC$.
Let $\sigma$ denote the complete classical measurement record.
Each unnormalized branch operator $F_j^{(\sigma)}$ is the time-ordered product of ideal gates, the chosen fault Kraus operators and measurement operators along a fault path with record $\sigma$, including the prescribed feed-forward.
The index $j$ labels the fault path and any unrecorded Kraus choices.
A completely positive trace-preserving recovery channel $\mathcal R_\sigma$ restores $U$ on their linear span when \[\mathcal R_\sigma(F\omega F^\dagger)=\lambda_F U\omega U^\dagger\] for every $F\in\operatorname{span}\{F_j^{(\sigma)}\}$ and every density operator $\omega$ supported on $\cC$, with $\lambda_F\geq0$ independent of $\omega$.
Such channels exist if and only if, for every record $\sigma$ and all $j,k$, \[\PC F_j^{(\sigma)\dagger}F_k^{(\sigma)}\PC=c_{jk}^{(\sigma)}\PC\] for scalars $c_{jk}^{(\sigma)}$.
\end{theorem}

\begin{proof}
Fix an outcome $\sigma$ and choose a finite spanning set of branch operators. If the displayed conditions hold, the standard Knill--Laflamme recovery channel corrects the branch operators $F_j^{(\sigma)}$ on $\cC$, and hence their linear span~\cite{knill1997theory}. Composing this channel with the target isometry $U$ gives $\mathcal R_\sigma$.

Conversely, suppose that $\mathcal R_\sigma$ satisfies the stated recovery property. Compose it with a CPTP channel that inverts $U$ on its range; such a channel exists for every isometry. This recovers the input for the completely positive map $\mathcal E_\sigma(\omega)=\sum_jF_j^{(\sigma)}\omega F_j^{(\sigma)\dagger}$, up to its input-independent trace. If this trace is nonzero, normalize $\mathcal E_\sigma$ on the input code and apply the necessity of the Knill--Laflamme conditions. If it is zero, every $F_j^{(\sigma)}\PC$ vanishes and the conditions hold trivially. Distinct classical records allow their recovery channels to be chosen independently.
\end{proof}

\subsubsection{Additional single-fault paths and proof of Proposition~\ref{prop:bare-failure}}\label{app:proof-bare-failure}

Besides the $B$ fault stated in the main text, two further ambiguities are useful.
\begin{lemma}[Single-fault paths of the unprotected gadget]\label{lem:bare-fault-paths}
Write $V=R_A(\pi/2)R_B(\pi/4)$ for the ideal rotation sequence. Here $L=-iAB$ is the nontrivial logical Pauli being rotated, and $|\psi\rangle\in\cC$ is the encoded input state.
\begin{enumerate}[label=\arabic*.]
    \item\label{fp:between-b} A $B$ fault between the two rotations produces $LV|\psi\rangle$ and leaves the syndrome outcome unchanged.
    \item\label{fp:between-a} A single-qubit fault $E$ with $\{E,A\}=0$ between the rotations produces $EAV|\psi\rangle$ up to phase. The same syndrome outcome can arise from $E$ after the second rotation, but the two paths require different logical corrections.
    \item\label{fp:pre-b} A single-qubit Pauli $E$ with $[E,A]=0$ and $\{E,B\}=0$ immediately before the first rotation changes the non-Clifford angle to $-\pi/4$. For a factorization of Eq.~\eqref{eq:split-pair} one may take $E=A_q$ on the qubit shared by both factors $q$. After removing the shifted Pauli error and applying the prescribed correction, outcomes $0$ and $s$ implement $R_L(-\pi/4)$ and $R_L(3\pi/4)$, respectively.
\end{enumerate}
\end{lemma}
\par

\begin{proof}[Proof of Lemma~\ref{lem:bare-fault-paths}]
For item 1, $B$ commutes with $R_B(\pi/4)$ and anticommutes with $A$, so $R_A(\pi/2)B=BR_A(-\pi/2)=iBAR_A(\pi/2)$. Since $BA=-iL$, the state is $LV|\psi\rangle$. 
The logical Pauli $L$ commutes with every stabilizer, so the syndrome is unchanged. The prescribed corrections leave $L$ as a logical error up to a global phase. 

For item 2, similarly $R_A(\pi/2)E=ER_A(-\pi/2)=iEAR_A(\pi/2)$. The final syndrome is $\syn(E)+s$ on one ideal branch and $\syn(E)$ on the other; the same outcomes arise from $E$ after $R_A(\pi/2)$ on the opposite branches, with logical actions differing by $R_L(\pi/2)$. 
To see the incompatibility explicitly, set $e=\syn(E)$ and let $F_{\rm before}^{(e)}$ and $F_{\rm after}^{(e)}$ denote the unnormalized branch operators for outcome $e$ when $E$ occurs immediately before or after $R_A(\pi/2)$, respectively.
The ideal branches are $R_L(\pi/4)\PC/\sqrt2$ and $-iAR_L(-\pi/4)\PC/\sqrt2$, so
\[
F_{\rm before}^{(e)}\PC=\frac{1}{\sqrt2}ER_L(-\pi/4)\PC,\qquad
F_{\rm after}^{(e)}\PC=\frac{1}{\sqrt2}ER_L(\pi/4)\PC.
\]
Consequently,
\[
\PC F_{\rm after}^{(e)\dagger}F_{\rm before}^{(e)}\PC=\frac12R_L(-\pi/2)\PC,
\]
which is not a scalar multiple of $\PC$ because $L$ is a nontrivial logical Pauli.
The same prescribed outcome-dependent correction on both paths leaves this cross product unchanged.
Theorem~\ref{thm:record-kl} therefore rules out one outcome-dependent recovery for both paths. 

For item 3, the commutation assumptions give $R_B(\pi/4)E=ER_B(-\pi/4)$ and $R_A(\pi/2)E=ER_A(\pi/2)$. Theorem~\ref{thm:two-rotations} with $\beta=-\pi/8$ gives $\theta_0=-\pi/4$ and $\theta_s=\pi/4$, and the prescribed outcome-$s$ correction produces $R_L(3\pi/4)$.
\end{proof}

The proof of Proposition~\ref{prop:bare-failure} below uses item~\ref{fp:between-b} of Lemma~\ref{lem:bare-fault-paths}.
Item~\ref{fp:between-a} motivates replacing the unprotected second rotation by protected Pauli measurements in Sec.~\ref{sec:golay}, while item~\ref{fp:pre-b} illustrates how a single-qubit error before the non-Clifford rotation changes the resulting logical angles even after its Pauli part is corrected.

Now, we are ready to prove Proposition~\ref{prop:bare-failure}.

\begin{proof}[Proof of Proposition~\ref{prop:bare-failure}]
Lemma~\ref{lem:bare-fault-paths}, item~\ref{fp:between-b}, shows that a $B$ fault between the rotations produces the nontrivial logical error $L$, preserves the syndrome outcome and leaves $L$ after the prescribed correction.
\par
In the elementary rotation circuit of Fig.~\ref{fig:parity-compilation}, the central rotation acts on qubit $w$.
A $Z_w$ fault immediately after the central rotation propagates through the remaining gates to $P$. Taking $P=B$ therefore realizes the malignant $B$ fault of Proposition~\ref{prop:bare-failure} at a single elementary circuit location. This makes the unprotected circuit's failure independent of whether a many-qubit rotation is a hardware operation.
Thus a single elementary fault can cause an uncorrectable logical error, proving that the unprotected gadget is not one-fault tolerant.
Under the noise assumption of Proposition~\ref{prop:bare-failure}, this gives logical gate infidelity $\Theta(p)$.
Indeed, when this is the sole fault, the nontrivial logical Pauli has a nonzero gate infidelity independent of $p$, so this event supplies the $\Omega(p)$ lower bound, while the total probability of any fault in a fixed finite circuit is $O(p)$.
\end{proof}

\subsubsection{\texorpdfstring{Limitation of a single unprotected $T$ gate}{Limitation of a single unprotected T gate}}

The adaptive measurements of Theorem~\ref{pm:adaptive} determine the logical gate in every ideal branch, but their correctness does not by itself protect the physical non-Clifford rotation.
For a circuit containing only one physical $T$-type gate, the next proposition shows that adding stabilizer ancillas, Clifford gates and Pauli measurements cannot both retain an exact non-Clifford logical action with positive success probability and correct or certainly reject the phase fault immediately after that gate.
This explains why stabilizer operations alone cannot repair the exposed single-$T$ implementation of Proposition~\ref{prop:bare-failure}.
\par

\begin{proposition}[One unprotected $T$ cannot be repaired by stabilizer operations]\label{pm:oneT}
Let $m\in\mathcal M$ denote a complete record of the circuit's measurement outcomes and classical random choices, and let $\mathcal M_{\mathrm{acc}}\subseteq\mathcal M$ be the fixed set of records for which the circuit declares success.
Write $p_{\mathrm{acc}}$ for its fault-free acceptance probability.
\par
A finite adaptive circuit using stabilizer ancillas, Clifford gates, Pauli measurements, stabilizer recovery and encoding, and exactly one physical $T$-type gate cannot implement an exact non-Clifford logical unitary $U$ on $k$ logical qubits if (i) a fault-free run is accepted with probability $p_{\mathrm{acc}}>0$, and (ii) a Pauli fault immediately after the gate along its rotation axis is either corrected on every accepted run or rejected with certainty.
For $T$ or $T^\dagger$, this is a $Z$ fault on the gate qubit.
Acceptance means that $m\in\mathcal M_{\mathrm{acc}}$ and the circuit returns a logical output; conditioned on acceptance, the output must equal $U\rho U^\dagger$ for every logical input density operator $\rho$.
The same set $\mathcal M_{\mathrm{acc}}$ is used in the fault-free and faulty circuits.
No other nonstabilizer resource is allowed.
\end{proposition}

\begin{proof}
We argue by contradiction.
Suppose such a circuit implements a non-Clifford unitary $U$ on $k$ logical qubits.
Thus $U$ acts on $(\mathbb C^2)^{\otimes k}$, the fault-free acceptance probability is non-zero, and every accepted output, both without the specified fault and with it, has the exact logical action $U$.
We will show that these assumptions force $U$ to be Clifford.
A single-qubit Clifford change of basis maps the gate's rotation axis and its specified fault to $Z$.
Absorbing this change into the surrounding Clifford operations, it suffices to consider a unique $T$ or $T^\dagger$ gate and a $Z$ fault immediately after it on the same qubit.

An accepted record is an element $m\in\mathcal M_{\mathrm{acc}}$ for which the circuit declares success and returns a logical output.
The outputs of rejected runs are discarded.
For each accepted run, we record all measurement outcomes and any classical random choices in $m$.

For each qubit that would otherwise be discarded, also record the outcome of a computational-basis measurement immediately before discarding it, and collect these additional outcomes in a string $d$.
These added outcomes do not affect the original acceptance rule or subsequent gates; summing over them reproduces the original partial trace.
Explicitly, if $D$ is a discarded register and $\tau$ is its joint state with the retained register, then $\operatorname{Tr}_D\tau=\sum_d(I\otimes\langle d|)\tau(I\otimes|d\rangle)$, where $d$ runs over the computational basis of $D$.
The discarded qubits need not be in stabilizer states.
The partial-trace identity holds for every $\tau$, including entangled states and states produced using the $T$ gate.
Include the Clifford encoding, all outcome-dependent corrections, and ideal Clifford decoding in the description of each record.
All ancillary output qubits of the decoding are included among the discarded qubits; the remaining input and output registers each contain $k$ qubits.

The refined record is $\omega=(m,d)$, where $m$ contains the original measurement outcomes and classical random choices, and $d$ contains only the added computational-basis outcomes on discarded qubits.
The accepted refined records form
\begin{equation}
\Omega_{\mathrm{acc}}=\{(m,d):m\in\mathcal M_{\mathrm{acc}}\},
\label{pm:oneT-records}
\end{equation}
where $d$ ranges over the computational-basis outcomes on the qubits discarded in that run.
Acceptance depends on $m$ alone.
\par
For a fixed $\omega$, let $K_\omega(O)$ be the $2^k\times2^k$ linear operator obtained by replacing the sole physical $T$ or $T^\dagger$ gate by a single-qubit operator $O$ acting on the same physical qubit at the same circuit location.
Fixing all recorded outcomes in $\omega=(m,d)$ fixes every adaptive choice and every measurement operator, so the encoded circuit followed by decoding gives one unnormalized output vector $K_\omega(O)|\psi\rangle$ for each logical input $|\psi\rangle$.
Without fixing $d$, the original record $m$ would in general specify a sum of Kraus branches, rather than a single operator on the retained logical register.

Let $c_\ell$ be the $\ell$th classical random choice in $m$, let $m_{<\ell}$ be its preceding original record, and let $r_\ell(c_\ell\mid m_{<\ell})$ be the probability assigned to that choice by the circuit.
If $F_\omega(O)$ denotes the fixed-outcome operator with these classical probability factors omitted, then
\begin{equation}
r_m=\prod_\ell r_\ell(c_\ell\mid m_{<\ell}),\qquad K_\omega(O)=\sqrt{r_m}\,F_\omega(O),
\label{pm:oneT-random-weights}
\end{equation}
with $r_m=1$ when there are no classical random choices.
The measurement probabilities are already contained in the norms of the projected vectors and require no additional factors.
Consequently, summing the refined accepted branches gives the original accepted map:
\begin{equation}
\begin{aligned}
&\sum_{m\in\mathcal M_{\mathrm{acc}}}r_m\sum_d F_{(m,d)}(O)\rho F_{(m,d)}(O)^\dagger
=\sum_{\omega\in\Omega_{\mathrm{acc}}}K_\omega(O)\rho K_\omega(O)^\dagger.
\end{aligned}
\label{pm:oneT-refined-map}
\end{equation}
\par

For a normalized input $|\psi\rangle$, the record has probability $\|K_\omega(O)|\psi\rangle\|^2$ whenever $O$ is the gate used in that circuit.
Define $K_{I,\omega}=K_\omega(I)$ and $K_{Z,\omega}=K_\omega(Z)$; these two replacements leave only stabilizer operations.
Since the replaced gate appears once, $K_\omega(O)$ is linear in $O$.
Up to an irrelevant global phase, write the gate as $aI+bZ$, where
\[
a=\cos(\pi/8),\qquad b=
\begin{cases}
-i\sin(\pi/8),&T,\\
+i\sin(\pi/8),&T^\dagger.
\end{cases}
\]
The fault-free branch operator is therefore
\[
K_{0,\omega}=K_\omega(aI+bZ)=aK_{I,\omega}+bK_{Z,\omega}.
\]
An immediately following $Z$ fault changes the gate to $Z(aI+bZ)=bI+aZ$, giving 
\[
K_{1,\omega}=K_\omega(bI+aZ)=bK_{I,\omega}+aK_{Z,\omega}.
\]
Thus $K_{1,\omega}$ describes the same recorded outcomes and corrections when the physical $T$-type gate is followed by the specified $Z$ fault.

Thus $K_{0,\omega}$ and $K_{1,\omega}$ use the same refined record $\omega$ and the same conditional corrections; the index $0$ means no fault, and the index $1$ means that the specified $Z$ fault follows the gate.
Let $\mathcal E_0$ and $\mathcal E_1$ be the accepted decoded maps, summed over all refined accepted records, without and with this fault.
They are trace-nonincreasing maps from $k$-qubit density operators to unnormalized $k$-qubit outputs, with $\mathcal E_j(\rho)=\sum_{\omega\in\Omega_{\mathrm{acc}}}K_{j,\omega}\rho K_{j,\omega}^\dagger$, where $j=0$ denotes no fault and $j=1$ denotes the specified fault.
Exact implementation and correction or certain rejection require 
\[
\mathcal E_j(\rho)=p_{\mathrm{acc},j}U\rho U^\dagger,\qquad j\in\{0,1\},\qquad p_{\mathrm{acc},0}=p_{\mathrm{acc}}>0,\quad p_{\mathrm{acc},1}\geq0.
\]
This equality holds for every $k$-qubit density operator $\rho$, and $p_{\mathrm{acc},j}=\operatorname{Tr}\mathcal E_j(\rho)$ is the probability of acceptance.
The acceptance probabilities $p_{\mathrm{acc},j}$ are input-independent by linearity and the exact action on every logical input.

The equality for $j=0$ is exact fault-free implementation, and assumption~(i) supplies $p_{\mathrm{acc},0}=p_{\mathrm{acc}}>0$.
Assumption~(ii) gives the equality for $j=1$: accepted faulty outputs have the same logical action $U$, while certain rejection is the case $p_{\mathrm{acc},1}=0$.
To make input-independence explicit, temporarily write $p_j(\rho)=\operatorname{Tr}\mathcal E_j(\rho)$.
For two distinct density operators $\rho,\sigma$, linearity and the exact conditional output give
\[
p_j\!\left(\frac{\rho+\sigma}{2}\right)(\rho+\sigma)=p_j(\rho)\rho+p_j(\sigma)\sigma.
\]
Distinct unit-trace density operators are linearly independent, so comparison of their coefficients gives $p_j(\rho)=p_j(\sigma)$.
Thus the constants $p_{\mathrm{acc},j}$ used above are independent of the logical input, including when the accepted map is zero.
\par

For an operator $K$ on the logical qubits, define its Choi vector by
\[
|J_K\rangle=(I\otimes K)|\Phi_k\rangle,\qquad |\Phi_k\rangle=2^{-k/2}\sum_{x\in\{0,1\}^k}|x\rangle|x\rangle.
\]
Writing $K_{j,\omega}$ for the branch operator of refined accepted record $\omega$, the Choi operators of the two maps obey 
\[
\sum_{\omega\in\Omega_{\mathrm{acc}}}|J_{K_{j,\omega}}\rangle\langle J_{K_{j,\omega}}|=p_{\mathrm{acc},j}|J_U\rangle\langle J_U|.
\]
Taking the expectation in any vector orthogonal to $|J_U\rangle$ gives a sum of nonnegative squared moduli equal to zero.
For $\langle v|J_U\rangle=0$, this expectation is 
\[
\sum_{\omega\in\Omega_{\mathrm{acc}}}|\langle v|J_{K_{j,\omega}}\rangle|^2=p_{\mathrm{acc},j}|\langle v|J_U\rangle|^2=0.
\]
Every summand is nonnegative, so $\langle v|J_{K_{j,\omega}}\rangle=0$ for every $\omega$ and every such $v$.
Thus every branch Choi vector is proportional to $|J_U\rangle$.
The branch Choi vector is simply $|J_{K_{j,\omega}}\rangle=(I\otimes K_{j,\omega})|\Phi_k\rangle$, namely the output when that branch acts on half of $k$ Bell pairs.
Since $K\mapsto|J_K\rangle$ is injective, each refined accepted record satisfies 
\[
K_{0,\omega}=\alpha_\omega U,\qquad K_{1,\omega}=\beta_\omega U,
\]
where $\alpha_\omega,\beta_\omega\in\mathbb C$ depend on the record, not on the logical input.
Their squared moduli are that record's fault-free and faulty probabilities.
For every normalized $|\psi\rangle$, unitarity of $U$ gives
\[
\begin{aligned}
\Pr(\omega\mid\text{no fault})&=\|K_{0,\omega}|\psi\rangle\|^2=|\alpha_\omega|^2,\\
\Pr(\omega\mid Z\text{ fault})&=\|K_{1,\omega}|\psi\rangle\|^2=|\beta_\omega|^2.
\end{aligned}
\]
In particular, $p_{\mathrm{acc},0}=\sum_{\omega\in\Omega_{\mathrm{acc}}}|\alpha_\omega|^2$ and $p_{\mathrm{acc},1}=\sum_{\omega\in\Omega_{\mathrm{acc}}}|\beta_\omega|^2$.
Certain rejection of the fault means $p_{\mathrm{acc},1}=0$, so every accepted faulty branch is zero and $\beta_\omega=0$.
By assumption~(i), $\sum_{\omega\in\Omega_{\mathrm{acc}}}|\alpha_\omega|^2=p_{\mathrm{acc}}>0$, so there is an accepted refined record $\omega$ with $\alpha_\omega\ne0$; fix this $\omega$ for the remainder of the proof, retaining its label.

The coefficient matrix of the two equations for $K_{0,\omega},K_{1,\omega}$ is $\left(\begin{smallmatrix}a&b\\b&a\end{smallmatrix}\right)$.
Explicitly, the two linear combinations above can be written as 
\[
    \begin{pmatrix}K_{0,\omega}\\
    K_{1,\omega}\end{pmatrix}
    =
    \begin{pmatrix}a&b\\
    b&a\end{pmatrix}
    \begin{pmatrix}K_{I,\omega}\\
    K_{Z,\omega}\end{pmatrix}.
\]
Here $a$ and $b$ are the coefficients of the physical $T$-type gate defined above, and $(\pm i)^2=-1$ gives $b^2=-\sin^2(\pi/8)$.
For either sign of $b$, its determinant is
\[
a^2-b^2=\cos^2(\pi/8)+\sin^2(\pi/8)=1.
\]
The inverse matrix is therefore
\[
\begin{pmatrix}a&b\\b&a\end{pmatrix}^{-1}
=
\frac{1}{a^2-b^2}\begin{pmatrix}a&-b\\-b&a\end{pmatrix}
=
\begin{pmatrix}a&-b\\-b&a\end{pmatrix}.
\]
Inverting this matrix gives
\[
\begin{aligned}
K_{I,\omega}&=aK_{0,\omega}-bK_{1,\omega}=(a\alpha_\omega-b\beta_\omega)U,\\
K_{Z,\omega}&=-bK_{0,\omega}+aK_{1,\omega}=(-b\alpha_\omega+a\beta_\omega)U.
\end{aligned}
\]
At least one of $a\alpha_\omega-b\beta_\omega$ and $-b\alpha_\omega+a\beta_\omega$ is nonzero: if both vanished, then $K_{I,\omega}=K_{Z,\omega}=0$ and hence $K_{0,\omega}=aK_{I,\omega}+bK_{Z,\omega}=0$, contrary to our choice $K_{0,\omega}=\alpha_\omega U$ with $\alpha_\omega\ne0$.
Choose a nonzero operator $K_{\star,\omega}=\gamma_\omega U\in\{K_{I,\omega},K_{Z,\omega}\}$, thus $\gamma_\omega\in\{a\alpha_\omega-b\beta_\omega,-b\alpha_\omega+a\beta_\omega\}$.
It is realized by replacing {the sole physical $T$ or $T^\dagger$ gate} by $I$ or $Z$ and retaining only the fixed refined record $\omega=(m,d)$, so it is a postselected branch made entirely of stabilizer operations.

Applying this branch to the second half of the stabilizer Bell state gives 
\[
(I\otimes K_{\star,\omega})|\Phi_k\rangle=\gamma_\omega(I\otimes U)|\Phi_k\rangle=\gamma_\omega|J_U\rangle.
\]
Stabilizer ancillas, Clifford gates, and Pauli measurements with specified outcomes take a pure stabilizer state to zero or to a scalar multiple of a pure stabilizer state.
Since $\gamma_\omega\ne0$, $|J_U\rangle$ must therefore be a stabilizer state.

To see why this is impossible for a non-Clifford $U$, note that both reduced states of $|J_U\rangle$ are maximally mixed.
Its stabilizer therefore contains no nonidentity element supported on only one half.
Its $4^k$ elements consequently have distinct Pauli labels on the first half, covering all $k$-qubit Pauli labels modulo phase.
Hence, for every Hermitian Pauli $P$, there is a Hermitian Pauli $Q$ such that
\[
(P^{\mathsf T}\otimes Q)|J_U\rangle=|J_U\rangle,
\]
where the transpose is in the computational basis.
The Bell-state identity $(P^{\mathsf T}\otimes I)|\Phi_k\rangle=(I\otimes P)|\Phi_k\rangle$ then gives $QUP=U$, and hence $Q=UPU^\dagger$.
Thus $U$ conjugates every Pauli to a Pauli and is Clifford, contradicting the assumed non-Clifford action.
\end{proof}

\subsection{Faults in compiled blocks}

We show how a fault in a compiled unitary circuit can spread within its support, and relate the probability of a faulty compiled block to the elementary noise parameter. Lemma~\ref{lem:compiled-support} is used for the Golay rotation and controlled blocks in Theorem~\ref{thm:golay-complete}.

\begin{lemma}[{Fault support after a compiled block}]\label{lem:compiled-support}
Let $U=U_m\cdots U_1$ be a unitary circuit whose gates act only on a set $S$ of qubits.
Each $U_j$ acts only on $S$.
A single internal fault acting only on the qubits participating in its faulty location is equivalent at its output to an arbitrary error on $S$ after the ideal $U$, including for an input entangled with a reference.
The fault may affect both qubits of a two-qubit gate; it need not be a single-qubit error.
\end{lemma}

\begin{proof}
{Represent a faulty gate by an error Kraus operator $E_a$ immediately after its ideal gate $U_j$, and put $U_{>j}=U_m\cdots U_{j+1}$. Then
\[
U_{>j}E_aU_j\cdots U_1
=\bigl(U_{>j}E_aU_{>j}^\dagger\bigr)U.
\]
Every suffix gate acts within $S$, so the effective error is supported on $S$. The equality holds for every Kraus operator and acts trivially on an external reference.}
\end{proof}

A single-fault guarantee therefore transfers to this compilation only when the proof handles the full operator space on $S$. 
Here single-fault tolerance has the meaning defined in Sec.~\ref{sec:bare}, including its incoming-error budget.
An elementary two-qubit fault may spread over all of $S$. Individual basis changes and CNOT gates need not preserve the static intermediate code. Preparations and measurements retain their stated verified implementations; the specified supports and operation order remain part of the circuit proof.

Let $D_i$ be compiled blocks with disjoint sets of elementary locations, containing $m_i$ locations each, and call a block faulty when at least one of its locations is faulty.
The index $i$ labels circuit blocks. Their elementary locations are disjoint, but their qubit supports may overlap.
For any specified $r$ blocks, a union bound over one faulty location in each block and Eq.~\eqref{eq:local-stochastic-budget} give
\begin{equation}
    \Pr(\text{all selected blocks are faulty})
    \le p^r\prod_{i=1}^r m_i\le(mp)^r,
\label{eq:compiled-block-rate}
\end{equation}
where $m=\max_i m_i$.
The parameter $p$ is the elementary local-stochastic noise parameter in Eq.~\eqref{eq:local-stochastic-budget}.
The selected elementary locations are distinct, so no independence assumption is needed. This bound concerns fault occurrences; correctness also requires the stated support conditions. The logical-failure bounds count the complete elementary schedules, including conditional operations, recovery and waits.

\section{Shared primitives and the fixed 22-qubit construction}\label{app:selective}
Here we give the code constructions and protected circuits supporting Sec.~\ref{sec:selective}.
We first establish the distance formula for selective concatenation and specify its component codes before applying it to examples.
The fixed 22-qubit gate is then built from protected correction, Pauli measurement, and circuits for logical Clifford gates and transfer between encodings, followed by its finite schedule and recursive error bound.
\par

\subsection{Selective concatenation}\label{app:selective-details}

Here we define selective concatenation, meaning that only chosen outer-code qubits are encoded in inner codes, and derive the resulting distance formula.
{Let $\mathcal{S}^{\mathrm{out}}$ be the stabilizer group of the outer code on $n$ qubits, and let $\mathcal{S}_i^{\mathrm{in}}$ be the stabilizer group of the one-logical-qubit inner code at outer qubit $i$. 
Different blocks may use different inner codes.
An unencoded qubit is treated as a one-qubit inner code with $\mathcal{S}_i^{\mathrm{in}}=\{I\}$.}
Then, we lift each outer Pauli by replacing its factors with compatible inner logical representatives. 
We denote the physical lift of the outer Pauli $P$ by $\widehat P$ and retain the overall phase of the outer Pauli $P$.
{The resulting stabilizer group $\mathcal{S}^{\rm enc}$ is generated by the inner checks from all $\mathcal{S}_i^{\mathrm{in}}$ and the lifts of the outer checks in $\mathcal{S}^{\mathrm{out}}$.}

\begin{proposition}[Distance with unequal inner codes]\label{rep:weights}\label{prop:selective-lift}
    For an outer stabilizer group $\mathcal{S}^{\mathrm{out}}$ and inner
    stabilizer groups $\mathcal{S}_i^{\mathrm{in}}$, fix compatible physical
    inner logical representatives
    $\bar P_i\in\cN(\mathcal{S}_i^{\mathrm{in}})$ of the single-qubit labels
    $P\in\{X,Y,Z\}$, with $\bar I_i=I$. Define
    \[
    \lambda_i(P)=
    \min_{R_i\in\bar P_i\mathcal{S}_i^{\mathrm{in}}}\wt(R_i),
    \qquad \lambda_i(I)=0.
    \]
    Thus the minimum is over physical Paulis in the inner logical
    coset corresponding to $P$.
    An unencoded {qubit} has $\lambda_i(P)=1$ for $P\ne I$. 
    Here $P$ is the single-qubit Pauli on outer physical qubit $i$; its inner representative acts as $P$ on the one logical qubit of inner block $i$.
    {The distance $d_{\rm enc}$ of the code with stabilizer group $\mathcal{S}^{\rm enc}$ is}
    \begin{equation}
        d_{\rm enc}=\min_{E\in\cN(\mathcal{S}^{\mathrm{out}})\setminus \mathcal{S}^{\mathrm{out}}}\sum_i\lambda_i(E_i).
    \label{rep:distance-formula}
    \end{equation}
\end{proposition}

\begin{proof}
    {Let $\widetilde E\in\cN(\mathcal{S}^{\rm enc})$ be a physical
    Pauli on the concatenated code. Its restriction
    $\widetilde E_i$ to inner block $i$ belongs to
    $\cN(\mathcal{S}_i^{\mathrm{in}})$. Let $E_i\in\{I,X,Y,Z\}$ be the
    single-qubit label of its inner logical coset, so
    $\widetilde E_i\in\overline{E_i}\mathcal{S}_i^{\mathrm{in}}$, where
    $\overline{E_i}$ denotes the fixed inner representative of
    that label. The labels define the outer physical Pauli
    $E=\bigotimes_iE_i$; $E_i$ is a single-qubit factor of $E$,
    whereas $\widetilde E_i$ acts on the whole inner block.}
    
    {For $g=\bigotimes_i g_i\in \mathcal{S}^{\mathrm{out}}$, its lift
    $\widetilde g=\bigotimes_i\overline{g_i}$ is a check of
    $\mathcal{S}^{\rm enc}$. The commutation sign of $\widetilde E$ with
    $\widetilde g$ equals that of $E$ with $g$. Hence the lifted
    outer-check constraints are exactly
    $E\in\cN(\mathcal{S}^{\mathrm{out}})$.
    Moreover, $\widetilde E\in \mathcal{S}^{\rm enc}$ exactly when
    $E\in \mathcal{S}^{\mathrm{out}}$: its label is then an outer stabilizer,
    and it differs from that stabilizer's lift by a product of
    inner stabilizers from the $\mathcal{S}_i^{\mathrm{in}}$.
    Thus $\widetilde E\in\cN(\mathcal{S}^{\rm enc})\setminus \mathcal{S}^{\rm enc}$
    exactly when
    $E\in\cN(\mathcal{S}^{\mathrm{out}})\setminus \mathcal{S}^{\mathrm{out}}$.}
    
    Because the inner blocks are disjoint,
    \[
    \wt(\widetilde E)=\sum_i\wt(\widetilde E_i)
    \geq\sum_i\lambda_i(E_i).
    \]
    Minimizing over nontrivial concatenated logical Paulis gives
    the lower bound in Eq.~\eqref{rep:distance-formula}.
    
    {Conversely, fix an outer logical Pauli
    $E=\bigotimes_iE_i\in
    \cN(\mathcal{S}^{\mathrm{out}})\setminus \mathcal{S}^{\mathrm{out}}$.
    In each inner block choose
    $R_i\in\overline{E_i}\mathcal{S}_i^{\mathrm{in}}$ with
    $\wt(R_i)=\lambda_i(E_i)$, taking $R_i=I$ when $E_i=I$.
    The choices are independent because the blocks are disjoint
    and multiplication by inner stabilizers does not change
    the labels $E_i$ or their outer commutation conditions.
    The physical Pauli $\widetilde E=\bigotimes_iR_i$ therefore
    belongs to $\cN(\mathcal{S}^{\rm enc})\setminus \mathcal{S}^{\rm enc}$ and has
    weight $\sum_i\lambda_i(E_i)$. Minimizing over the outer
    logical Paulis $E$ proves equality in
    Eq.~\eqref{rep:distance-formula}.}
\end{proof}

For the outer code used before and after the gate and the intermediate outer code, let us denote their stabilizer groups as $\mathcal{S}^{\mathrm{out}}$ and $\mathcal{S}_0^{\mathrm{out}}$, respectively.
Using the same inner blocks to lift these two groups gives
the concatenated stabilizers $\mathcal{S}^{\rm enc}$ and
$\mathcal{S}_0^{\rm enc}$, respectively. Both contain all inner
stabilizers $\mathcal{S}_i^{\mathrm{in}}$.

{Let $J\subseteq\{1,\ldots,n\}$ index the outer physical
qubits encoded in a common $\qcode{m}{1}{\Delta}$ inner code.
Leaving the other qubits unencoded gives
$n-|J|+m|J|=n+|J|(m-1)$ physical qubits. 
Write $d_J$ and $\delta_J$ for the distances of the
concatenated codes with stabilizers $\mathcal{S}^{\rm enc}$ and
$\mathcal{S}_0^{\rm enc}$, respectively.
Since every nonidentity
inner logical has weight at least $\Delta$,
Eq.~\eqref{rep:distance-formula} gives}
\begin{equation}
\begin{aligned}
    d_J&\geq
    \min_{E\in\cN(\mathcal{S}^{\mathrm{out}})\setminus \mathcal{S}^{\mathrm{out}}}
    w_J(E),\qquad
    \delta_J\geq
    \min_{E\in\cN(\mathcal{S}_0^{\mathrm{out}})\setminus \mathcal{S}_0^{\mathrm{out}}}
    w_J(E),\\
    w_J(E)&=\sum_{j\in\supp(E)}
    \bigl(\Delta\ind[j\in J]+\ind[j\notin J]\bigr).
\end{aligned}
\label{eq:weighted-weight}
\end{equation}
{Uniform concatenation encodes every outer qubit in the
same inner code, so $J=\{1,\ldots,n\}$. Then
$w_J(E)=\Delta\wt(E)$, and the two minima in
Eq.~\eqref{eq:weighted-weight} give $d_J\geq d\Delta$ and
$\delta_J\geq\delta\Delta$.}
{Here $d$ and $\delta$ are the distances of the
outer codes with stabilizers $\mathcal{S}^{\mathrm{out}}$ and
$\mathcal{S}_0^{\mathrm{out}}$, respectively.}

\subsection{Steane and Reed--Muller component codes}\label{app:component-codes}
Here we fix the signed component-code conventions used in the selective examples and the fixed code.
Use outer Steane supports $H_1=\{1,4,5,7\}$, $H_2=\{2,4,6,7\}$, and $H_3=\{3,5,6,7\}$. Write $x_i=X_{H_i}$ and $z_i=Z_{H_i}$, and let $\cC_0$ have stabilizer $\langle x_1,x_2,x_3,z_1,z_2,z_3\rangle$. Take

\begin{equation}\begin{gathered}
 L=Z_1Z_2Z_4,\quad A=X_1Z_2,\quad B=Y_1Z_4,\quad h=z_1,\\
 G=Y_1Z_2Z_4Z_5Z_7,\quad M=X_1Z_5Z_7,\quad
 \cS_0=\langle x_1,x_3,z_2,z_3,x_2z_1\rangle.
\end{gathered}\label{pm:steane}
\end{equation}

The intermediate code $\cD_0$ is stabilized by $\cS_0$. The signed product $x_2z_1$, rather than $x_2$ alone, is a retained generator.

We label the 15 physical qubits by the nonzero vectors of $\mathbb F_2^4$, using the binary expansion of each integer $1,\ldots,15$. This numbering is used to specify the supports of the checks and logical operators below.
For each $i=1,\ldots,4$, we define the length-$15$ binary string $u^{(i)}=(v_i)_{v\in\mathbb F_2^4\setminus\{0\}}$: its entry at qubit $v$ is the $i$th binary coordinate of $v$.
Then, let $V=\operatorname{span}_{\mathbb F_2} \{u^{(1)},u^{(2)},u^{(3)},u^{(4)}\}$ and $W=V+\langle\mathbf1\rangle$. 
The stabilizer has four independent $X$ checks from $V$ and ten independent $Z$ checks from $W^\perp$.
For the check-weight bounds below, choose the ten $Z$ checks with support strings $u^{(i)}$, $1\leq i\leq4$, and $u^{(ij)}=(v_iv_j)_{v\in\mathbb F_2^4\setminus\{0\}}$, $1\leq i<j\leq4$. These strings form a basis of $W^\perp$: their defining monomials are linearly independent, and their products with the constant and coordinate functions have even-weight evaluation strings, so they are orthogonal to $W$. The four coordinate strings have weight eight and the six quadratic strings have weight four.
\begin{equation}
 |b\rangle_{\rm RM}=\frac14\sum_{v\in V}|v+b\mathbf1\rangle,\qquad b\in\{0,1\}.
\label{rep:rm-basis}
\end{equation}
Compatible Hermitian logical representatives are
\begin{equation}
 \bar X=X_1\cdots X_7,\qquad \bar Z=Z_1Z_2Z_3,\qquad
 \bar Y=i\bar X\bar Z=-Y_1Y_2Y_3X_4X_5X_6X_7.
\label{rep:rm-reps}
\end{equation}

\begin{lemma}[Pauli-specific logical weights and transversal $T$]\label{rep:rm}
    The minimum logical weights in this convention are $(d_X,d_Y,d_Z)=(7,7,3)$, and $T^{\dagger\otimes15}$ implements the logical $T$.
\end{lemma}

\begin{proof}
We write a Pauli, up to phase, as $X^xZ^z$, where
$x,z\in\mathbb F_2^{15}$ specify its $X$ and $Z$ supports.
Commutation with the $Z$ checks requires
$x\in(W^\perp)^\perp=W$, and commutation with the $X$ checks
requires $z\in V^\perp$. The $X$ stabilizers have $x\in V$,
so the $X$ support of a logical $X$ or $Y$ lies in
$W\setminus V=\mathbf1+V$. Every nonzero string in $V$ has
weight eight, since a nonzero linear function on
$\mathbb F_2^4$ is one on eight vectors. The strings in
$\mathbf1+V$ therefore have weight seven or fifteen.
The weight of $X^xZ^z$ is at least the weight of $x$, so
$d_X,d_Y\geq7$. Both bounds are attained by the representatives
in Eq.~\eqref{rep:rm-reps}.

The $Z$ stabilizers have
$z\in W^\perp=V^\perp\cap\mathbf1^\perp$, namely the
even-weight strings orthogonal to $V$. Thus the $Z$ support
of a logical $Z$ or $Y$ is an odd-weight string in $V^\perp$.
Orthogonality to the four strings $u^{(i)}$ means that the
binary vector labels of its occupied qubits sum to zero.
A string of weight one cannot satisfy this condition because
the zero label is absent, and a string of weight two would
require two identical labels. Thus $d_Z\geq3$, including
representatives with both $X$ and $Z$ support. For distinct
nonzero labels $u,v$, the three qubits labelled $u,v,u+v$
give an odd-weight string whose labels sum to zero.
The choice $1,2,3$ gives the displayed $\bar Z$, so $d_Z=3$.
Together these bounds and representatives give
$(d_X,d_Y,d_Z)=(7,7,3)$.

Up to a common global phase, $T^\dagger$ multiplies
each $|1\rangle$ by $e^{-i\pi/4}$ and leaves $|0\rangle$
unchanged. A computational-basis string of weight $w$ therefore
acquires phase $e^{-i\pi w/4}$ under
$T^{\dagger\otimes15}$. This phase is one for the weights
zero and eight in $V$, and it is $e^{i\pi/4}$ for the weights
fifteen and seven in $\mathbf1+V$. Substitution in
Eq.~\eqref{rep:rm-basis} gives
\[
T^{\dagger\otimes15}|b\rangle_{\rm RM}
=e^{i\pi b/4}|b\rangle_{\rm RM},\qquad b\in\{0,1\},
\]
up to a common global phase. Thus the logical $T$ is implemented.
Since the physical gates act separately on each qubit, an error
on one qubit remains supported on that qubit after conjugation.
\end{proof}

For later circuit bounds, one explicit $Z$-check basis has supports $\{2,3,4,5\}$, $\{1,3,4,6\}$, $\{1,2,4,7\}$, $\{2,3,8,9\}$, $\{1,3,8,10\}$, $\{1,2,8,11\}$, $\{1,2,3,4,8,12\}$, $\{1,4,8,13\}$, $\{2,4,8,14\}$, and $\{3,4,8,15\}$. Each is even and its coordinate labels sum to zero. Together with the four weight-eight $X$ checks these give maximum check weight eight.

\subsection{Selective-concatenation examples}
For the corollary below, $\mathcal{S}^{\mathrm{out}}$ is the Steane stabilizer group of Eq.~\eqref{pm:steane}.
Its outer logical operator factors are $A=X_1Z_2$ and $B=Y_1Z_4$, where $AB=iL$ and $L=Z_1Z_2Z_4$. 
The retained outer stabilizer group is $\mathcal{S}_0^{\mathrm{out}}=\{g\in \mathcal{S}^{\mathrm{out}}:[g,A]=0\}$.

\begin{corollary}[{Two selective concatenations of the Steane code}]\label{cor:49}\label{cor:35}
    For each $i\in J=\{1,2,4\}$, take
    $\mathcal{S}_i^{\mathrm{in}}$ to be the stabilizer group of a separate
    $\qcode{15}{1}{3}$ Reed--Muller block, with the logical
    conventions of Eqs.~\eqref{rep:rm-basis}
    and~\eqref{rep:rm-reps}; set $\mathcal{S}_i^{\mathrm{in}}=\{I\}$
    on the unencoded qubits $i\notin J$.
    The concatenated codes with stabilizers $\mathcal{S}^{\rm enc}$ and $\mathcal{S}_0^{\rm enc}$ then have parameters $\qcode{49}{1}{5}$ and $\qcode{49}{2}{3}$, respectively.
    The same construction with $J=\{2,4\}$ gives $\qcode{35}{1}{3}$ and $\qcode{35}{2}{3}$ codes with stabilizers $\mathcal{S}^{\rm enc}$ and $\mathcal{S}_0^{\rm enc}$, respectively.
\end{corollary}

\begin{proof}
Every outer logical Pauli
$E\in\cN(\mathcal{S}^{\mathrm{out}})\setminus \mathcal{S}^{\mathrm{out}}$
has odd weight of at least three.
Every such weight-three Pauli $E$ satisfies
$\supp(E)\cap\{1,2,4\}\ne\varnothing$.
For $J=\{1,2,4\}$, each encoded nonidentity factor has minimum
weight at least three, so
\[
\sum_i\lambda_i(E_i)\geq
\begin{cases}
3+2=5,&\wt(E)=3,\\
\wt(E)\geq5,&\wt(E)\geq5.
\end{cases}
\]
Equation~\eqref{rep:distance-formula} therefore gives $d_J\geq5$.

The outer logical
$E_Z=Z_1Z_3Z_5\in
\cN(\mathcal{S}^{\mathrm{out}})\setminus \mathcal{S}^{\mathrm{out}}$
has the physical lift $\widetilde E_Z=\bar Z_1Z_3Z_5$.
Here $\bar Z_1\in
\cN(\mathcal{S}_1^{\mathrm{in}})\setminus \mathcal{S}_1^{\mathrm{in}}$
is a minimum-weight inner logical $Z$ on block $1$,
while $Z_3,Z_5$ act on unencoded physical qubits.
Lemma~\ref{rep:rm} gives $\wt(\bar Z_1)=3$.
\footnote{
Eq.~\eqref{rep:rm-reps} gives an explicit representative of $\bar Z_1$.}
Thus $\widetilde E_Z\in \cN(\mathcal{S}^{\rm enc})\setminus \mathcal{S}^{\rm enc}$ has weight $3+1+1=5$ and proves equality.

The set $\cN(\mathcal{S}_0^{\mathrm{out}})\setminus \mathcal{S}_0^{\mathrm{out}}$ contains no weight-one Pauli; its weight-two Paulis are $X_1Z_2,Y_1Z_4,Y_2X_4$.
Each has support intersecting both $J=\{1,2,4\}$ and $J=\{2,4\}$, so its lifted weight is at least $2+2=4$.
Every higher-weight outer intermediate logical has lifted weight at least three.
Equation~\eqref{rep:distance-formula} thus gives
$\delta_J\geq3$ for both choices of $J$.
{The outer intermediate logical
$E_0=Y_3X_5Z_6\in
\cN(\mathcal{S}_0^{\mathrm{out}})\setminus \mathcal{S}_0^{\mathrm{out}}$
avoids both selected sets. Its physical lift $\widetilde E_0$
acts by the same Paulis on the unencoded qubits $3,5,6$
and belongs to
$\cN(\mathcal{S}_0^{\rm enc})\setminus \mathcal{S}_0^{\rm enc}$.
It has weight three and proves equality.}

{For $J=\{2,4\}$, every
$E\in\cN(\mathcal{S}^{\mathrm{out}})\setminus \mathcal{S}^{\mathrm{out}}$
has weight at least three, and lifting cannot lower its
weight. The outer logical $E_Z=Z_1Z_3Z_5$ avoids $J$, so
its physical lift in
$\cN(\mathcal{S}^{\rm enc})\setminus \mathcal{S}^{\rm enc}$
acts on the unencoded qubits $1,3,5$ and has weight three.
Thus $d_J=3$.}
\end{proof}

The  49-qubit $\qcode{49}{1}{5}$ code above is the nonuniform Steane--Reed--Muller code of Refs.~\cite{nikahd2017nonuniform,chamberland2017overhead}. 

For the $J=\{2,4\}$ construction of Corollary~\ref{cor:35}, the lifted non-Clifford factor is $\widehat B=Y_1\bar Z_4$, with the pure-$Z$ representative of Eq.~\eqref{rep:rm-reps} on Reed--Muller block $4$.
A Clifford parity map conjugates this Pauli product to $Z_1$; explicitly, take \[C=\left(\prod_{j\in\supp(\bar Z_4)}\operatorname{CNOT}_{j\to1}\right)H_1S_1^\dagger.\]
Here $j$ labels a physical qubit in block $4$, and the target $1$ is the unencoded outer qubit.
The gates $S_1^\dagger$ and $H_1$ first map $Y_1$ to $Z_1$; the CNOTs then collect the $Z$ parity onto qubit $1$.
Thus $C\widehat B C^\dagger=Z_1$, and applying $C$, then a single-qubit rotation, then $C^\dagger$ implements 
\[
C^\dagger R_{Z_1}(\theta)C=R_{\widehat B}(\theta).
\]
A $Z_1$ fault immediately after its physical $T$ becomes $\widehat B$ after undoing $C$, and the  second rotation turns it into the logical fault of Proposition~\ref{prop:bare-failure}.
A pre-rotation check cannot detect this newly created error.
Proposition~\ref{pm:oneT} excludes exact repair using only additional stabilizer operations.
It does not prove that fifteen $T$-type gates are optimal, and does not exclude extra non-Clifford resources or restricted noise models.
\par

We next apply selective concatenation to the high-rate BCH family of Proposition~\ref{cor:bch}.
Only the seven qubits supporting the chosen logical Pauli are encoded, so the additional number of physical qubits is independent of $m$.
In the next corollary, $\mathcal{S}^{\mathrm{out}}$ denotes the BCH stabilizer group of Proposition~\ref{cor:bch}, and $\mathcal{S}_0^{\mathrm{out}}$ its retained subgroup for the  chosen factorization.
The inner and concatenated groups use the same notation as above.

\begin{corollary}[Selective BCH family]\label{cor:selective-bch}
{Choose a minimum-weight logical Pauli $L$ and a
 balanced factorization for the $\qcode{2^m-1}{2^m-1-6m}{7}$ BCH code of Proposition~\ref{cor:bch}. 
Use lifts of these outer factors in both concatenated constructions below.
Encoding each of the seven qubits in $J=\supp(L)$ in a separate $\qcode{15}{1}{3}$ Reed--Muller block gives codes with}
\begin{equation}
\begin{aligned}
    N_m&=2^m+97, & K_m&=2^m-1-6m,\\
    d&\geq7, & \delta&\geq3,
\end{aligned}
\label{eq:selective-bch}
\end{equation}
and $K_m/N_m\to1$. 
{Here $N_m$ and $K_m$ are the resulting numbers of physical and logical qubits of the  code used before and after the gate, while $d=d_J$ and $\delta=\delta_J$ are the  distances before and during the gate for the chosen  balanced factorization.
The intermediate code has the same $N_m$ physical qubits and $K_m+1$ logical qubits.}
Uniform concatenation with the same {$\qcode{15}{1}{3}$ Reed--Muller} inner code gives $\qcode{15(2^m-1)}{2^m-1-6m}{\geq21}$ codes with $\delta\geq9$ for lifts of the  balanced factorization of every minimum-weight outer-code logical and rate $\to1/15$.
\end{corollary}

\begin{proof}
The lifted outer factors commute with every inner stabilizer.
First, note that the intermediate stabilizer $\mathcal{S}_0^{\rm enc}$ retains all inner checks but has one fewer independent lifted outer check (i.e. the omitted outer stabilizer $h$).
It therefore has $N_m-K_m-1$ independent generators and encodes $N_m-(N_m-K_m-1)=K_m+1$ logical qubits.

Proposition~\ref{cor:bch} gives  original outer-code distance seven, and Theorem~\ref{thm:pure}(a) gives outer intermediate distance at least three for the  balanced factorization. These are the outer-code bounds used in Eq.~\eqref{rep:distance-formula}.
For selective concatenation, $\lambda_i(P)\geq1$ for every single-qubit Pauli $P\ne I$ on outer qubit $i$, so $\sum_i\lambda_i(E_i)\geq\wt(E)$. 
Taking the minima over the  logical operators of the original and intermediate outer codes gives $d_J\geq7$ and $\delta_J\geq3$, respectively.
Here $d_J$ and $\delta_J$ are the concatenated-code
distances defined before Eq.~\eqref{eq:weighted-weight}.
The number of physical qubits is $N_m=(2^m-1)+7(15-1)=2^m+97$; each inner block encodes one logical qubit, so $K_m=2^m-1-6m$ is unchanged and $K_m/N_m\to1$.

For uniform concatenation, every nonidentity factor has
$\lambda_i(E_i)\geq3$, so
$\sum_i\lambda_i(E_i)\geq3\wt(E)$. Consequently
Eq.~\eqref{rep:distance-formula} gives
\[
\begin{aligned}
    d_J&\geq
    3\min_{E\in\cN(\mathcal{S}^{\mathrm{out}})\setminus \mathcal{S}^{\mathrm{out}}}
    \wt(E)=3\cdot7 = 21
    \quad \text{and}\quad
    \delta_J&\geq
    3\min_{E\in\cN(\mathcal{S}_0^{\mathrm{out}})\setminus \mathcal{S}_0^{\mathrm{out}}}
    \wt(E)\geq3\cdot3 = 9.
\end{aligned}
\]
There are $15(2^m-1)$ physical qubits and the same $K_m$
logical qubits, giving rate $K_m/[15(2^m-1)]\to1/15$.
\end{proof}

The selective family implements a rotation about one chosen logical Pauli per set of seven active qubits. Encoding every qubit suffices for an arbitrary logical basis and gives the uniform family with rate $1/15$; a smaller union of active qubits may suffice for a chosen logical basis. In either case syndrome extraction inherits dense BCH checks of weight $\Theta(2^m)$, so the family concerns code rate rather than the cost of a syndrome round.

\subsection{The fixed code and its protected circuit}\label{app:rep}\label{app:combined22}

Here we define the fixed 22-qubit code, prove the distances of the codes used during its logical $T$ gate, and specify the protected correction and measurement circuits. We then prove single-fault correctness, give a 33-qubit serial schedule, and establish the recursive error bound.
This supplies the definitions and proofs for Sec.~\ref{subsec:gate22}, including Theorem~\ref{thm:main22} and the circuit in Fig.~\ref{fig:gate22}.

The protected gate and serial state-preparation schedule establish the claims in Theorem~\ref{thm:main22}.
The protected $H$, $S$ and CNOT gates of Lemma~\ref{rep:transfer}, together with the logical $T$ gate of Theorem~\ref{thm:main22}, form a universal fault-tolerant gate set. Corollary~\ref{c22:recursion} gives its recursive error bound.
The recursive construction uses this complete finite set of circuits rather than the distance of the original code alone.
\par

\subsubsection{Signed code definitions}\label{app:c22-signed}

Here we use the component conventions of \SM{}~\ref{app:component-codes} to define the fixed code $\cC'_{22}$ and its intermediate code $\cD'_{22}$.
The signed generators and logical representatives determine the distances and check weights used in Sec.~\ref{subsec:gate22} and Fig.~\ref{fig:gate22}.
\par

Encode outer qubit $1$ by this Reed--Muller code, qubit $2$ by $|0\rangle\mapsto|00\rangle$, $|1\rangle\mapsto|11\rangle$, and leave qubits $3,4,5,6,7$ unencoded. The repetition check and representatives are
\begin{equation}\label{eqn:repetition_check_representatives}
 Z_{2a}Z_{2b},\qquad \widehat X_2=X_{2a}X_{2b},\quad
 \widehat Z_2=Z_{2a},\quad \widehat Y_2=Y_{2a}X_{2b},
\end{equation}
where the subscripts $2a$ and $2b$ label the two physical qubits encoding outer qubit $2$.
Let $V_{\rm enc}$ be this encoding isometry, using the Reed--Muller basis of Eq.~\eqref{rep:rm-basis} on outer qubit $1$.
For an outer Pauli $P$, its lift $\widehat P$ replaces $X_1,Y_1,Z_1$ by $\bar X_1,\bar Y_1,\bar Z_1$ from Eq.~\eqref{rep:rm-reps}, uses the repetition representatives in Eq.~{\eqref{eqn:repetition_check_representatives}} on qubit $2$, and leaves the factors on qubits $3,\ldots,7$ unchanged, preserving the overall phase.
For an outer code $Q$, the notation $\widehat Q$ denotes its image under the same encoding:
\begin{equation}
\widehat Q=\{V_{\rm enc}|\psi\rangle:|\psi\rangle\in Q\},\qquad
\widehat P V_{\rm enc}=V_{\rm enc}P.
\label{c22:lift}
\end{equation}
In particular, $\widehat{C_B\cD_0}$ and $\widehat{C_B\cC_0}$ mean that the outer code is first transformed by $C_B$ and then encoded by $V_{\rm enc}$.
Their stabilizers consist of the inner checks and the lifts of the corresponding signed outer checks, as in \SM{}~\ref{app:selective-details}.

Recall the outer factor $B=Y_1Z_4$ defined in Eq.~\eqref{pm:steane}.
Define
\begin{equation}
 C_B=\operatorname{CNOT}_{4\to1}H_1S_1^\dagger,\qquad C_BBC_B^\dagger=Z_1.
\label{rep:codes}
\end{equation}
The retained outer generators, with their Hermitian signs, are
\begin{center}
\begin{tabular}{cc}
$\cD_0$&$C_B\cD_0$\\\hline
$+XIIXXIX$&$+ZIIYXIX$\\
$+IIXIXXX$&$+IIXIXXX$\\
$+IZIZIZZ$&$+IZIZIZZ$\\
$+IIZIZZZ$&$+IIZIZZZ$\\
$-ZXIYZXY$&$-IXIYZXY$
\end{tabular}
\end{center}
Lift the second column and append the fourteen Reed--Muller checks and the repetition check. This gives twenty independent generators. 
Explicitly, the five lifted outer generators, in the order of the second column, are
\[
\begin{gathered}
\bar Z_1Y_4X_5X_7,\quad X_3X_5X_6X_7,\\
Z_{2a}Z_4Z_6Z_7,\quad Z_3Z_5Z_6Z_7,\\
-X_{2a}X_{2b}Y_4Z_5X_6Y_7.
\end{gathered}
\]
The fourteen Reed--Muller checks and the repetition check are independent because the two inner blocks are disjoint.
A product of the five lifted outer generators can be an inner stabilizer only if the corresponding product of outer generators is the identity; their independence therefore supplies five further independent checks.
\par
The code used before and after the gate below adds one signed constraint. The repetition stabilizer has weight two, so these codes are degenerate.

Conjugate the outer operators by $C_B$ before lifting. Primes denote these outer Clifford images:
\begin{align}
 A'&=Y_1Z_2Z_4,&B'&=Z_1,&L'&=X_1Z_2Z_4,\nonumber\\
 h'&=X_1Z_4Z_5Z_7,&G'&=Z_1Z_2Z_5Z_7,&M'&=Y_1Z_4Z_5Z_7.
\label{c22:operators}
\end{align}
All signs in this display are positive in the outer convention. In particular $A'B'=iL'$, $G'=iA'h'$, and $M'=A'h'L'$. 
Physical $\bar Y$ retains the minus sign in Eq.~\eqref{rep:rm-reps}. 
Thus, the three physical representatives are
\begin{equation}
\begin{gathered}
\widehat G'=\bar Z_1Z_{2a}Z_5Z_7,\qquad
\widehat M'=\bar Y_1Z_4Z_5Z_7,\qquad
\widehat h'=\bar X_1Z_4Z_5Z_7.
\end{gathered}
\label{c22:measurement-lifts}
\end{equation}

Then, define
\begin{equation}
 \cD'_{22}=\widehat{C_B\cD_0},\qquad
 \cC'_{22}=\widehat{C_B\cC_0}
 =\{|\psi\rangle\in\cD'_{22}:\widehat h'|\psi\rangle=|\psi\rangle\}.
\label{c22:codes}
\end{equation}
A convenient generating set for $\cC'_{22}$ is the twenty generators just specified, and $\widehat h'$. 
This basis choice matters for the following weight bound. The logical operators
\begin{equation}
 Z_{\rm log}=Z_{2a}Z_5Z_7=\widehat{h'L'},\qquad
 X_{\rm log}=X_{2a}X_{2b}X_5X_7
\label{c22:logical}
\end{equation}
commute with these stabilizers and anticommute with each other. On $\cC'_{22}$, $\widehat L'$ equals $Z_{\rm log}$. Thus the target rotation is a $T$ gate in this explicitly specified logical basis.

\begin{proposition}[Distances and check weights]\label{c22:distances}
The codes $\cC'_{22}$ and $\cD'_{22}$ have parameters $\qcode{22}{1}{3}$ and $\qcode{22}{2}{3}$, respectively.
Fixing either sign of $\widehat G'$, $\widehat M'$, or $\widehat h'$ from Eq.~\eqref{c22:measurement-lifts} in $\cD'_{22}$ gives a $\qcode{22}{1}{3}$ code.
With the representatives and generators just specified, the maximum retained-check weight is eight, the maximum check weight for $\cC'_{22}$ is ten, and
\begin{equation}
 \mathrm{wt}(\widehat G')=6,\qquad
 \mathrm{wt}(\widehat M')=\mathrm{wt}(\widehat h')=10.
\label{c22:weights}
\end{equation}
\end{proposition}

\begin{proof}
    There are $n=15+2+5=22$ physical qubits and twenty independent retained generators, so $\cD'_{22}$ encodes $k=22-20=2$ logical qubits.
    Each of $\widehat h'$, $\widehat G'$, and $\widehat M'$ commutes with the retained generators and is independent of them: $\widehat h'$ anticommutes with $\widehat A'$, whereas $\widehat G'$ and $\widehat M'$ anticommute with $\widehat h'$, and all four of these operators commute with the retained generators.
    Fixing either sign of any of the three constraints therefore gives $k=22-21=1$; the positive $\widehat h'$ constraint gives $\cC'_{22}$.
    It remains to show that all these codes have distance three.
    Put $\cS'_0=C_B\cS_0C_B^\dagger$ for the retained outer stabilizer group.
    \par
    For the five generators in the $C_B\cD_0$ column, the complete single-qubit Pauli syndrome table is
    \begin{center}
    \begin{tabular}{c|ccccccc}
     &1&2&3&4&5&6&7\\\hline
    $X$&10000&00100&00010&10101&00011&00110&00111\\
    $Y$&10000&00101&01010&00100&11011&01111&11110\\
    $Z$&00000&00001&01000&10001&11000&01001&11001
    \end{tabular}
    \end{center}
    Write the five signed outer generators as $g_1,\ldots,g_5$ and define the syndrome bits by $g_jE=(-1)^{s_j(E)}Eg_j$.
    For any Paulis $E,F$,
    \[
    \begin{gathered}
    g_jEF=(-1)^{s_j(E)+s_j(F)}EFg_j,\\
    \syn_{\cS'_0}(EF)=\syn_{\cS'_0}(E)+\syn_{\cS'_0}(F),
    \end{gathered}
    \]
    with addition in $\mathbb F_2^5$.
    A zero syndrome therefore means that the Pauli commutes with every generator and belongs to $\cN(\cS'_0)$.
    In the table, $\syn_{\cS'_0}(Z_1)=00000$, giving a weight-one normalizer.
    The equal entries $\syn_{\cS'_0}(X_2)=\syn_{\cS'_0}(Y_4)=00100$ give $\syn_{\cS'_0}(X_2Y_4)=00100+00100=00000$; the product has weight two because its two nonidentity factors act on distinct qubits.
    These operators are outside $\cS'_0$: $Z_1$ anticommutes with the normalizer $h'$, and $X_2Y_4$ anticommutes with the normalizer $W=Z_2Z_3Z_6$.
    \par
    Thus the only outer logicals of weight at most two are $Z_1$ and $X_2Y_4$. Their lifted weights are three and $2+1=3$. Every other nontrivial outer logical has weight at least three, and each nonidentity inner logical has weight at least one. Equation~\eqref{rep:distance-formula} proves that $\cD'_{22}$ has distance at least three. Adding any of the stated signed constraints cannot decrease distance.
    
    For equality, $W=Z_2Z_3Z_6$ is an outer Steane logical, is unchanged by $C_B$, and its lift has physical weight three. It commutes with all retained and added checks. 
    Since $W$ lies outside the Steane stabilizer and is unchanged by $C_B$, it lies outside the conjugated outer stabilizer, so its lift lies outside the stabilizer of $\cC'_{22}$ by Proposition~\ref{prop:selective-lift}.
    It is also outside $\cS'_0G'$ and $\cS'_0M'$: every element of either coset anticommutes with $h'$, whereas $W$ commutes with $h'$. Thus it remains a nontrivial logical for both signs of each constraint, establishing all upper bounds. 
    Together with the physical-qubit and independent-generator counts at the start of the proof, the distance equalities give all the stated code parameters.
    
    The fourteen Reed--Muller checks have weight at most eight, with the four $X$ checks each of weight eight, the repetition check has weight two, and the five lifted outer checks have weights $6,4,4,4,6$.
    Hence the maximum retained-check weight is eight.
    {The code $\cC'_{22}$ adds} $\widehat h'=\bar X_1Z_4Z_5Z_7$, of weight ten.
    Hence the maximum check weight of this generating set for $\cC'_{22}$ is ten.
    The disjoint supports in Eq.~\eqref{c22:measurement-lifts} give $\mathrm{wt}(\widehat G')=3+1+1+1=6$ and $\mathrm{wt}(\widehat M')=\mathrm{wt}(\widehat h')=7+1+1+1=10$, proving Eq.~\eqref{c22:weights}.
\end{proof}

\begin{remark}[Role of the repetition block]\label{rep:minimal}\label{rep:distances}
Before conjugation by $C_B$, the only weight-two intermediate logicals are $X_1Z_2,Y_1Z_4,Y_2X_4$, with no weight-one logical. 
The same lift gives each weight at least three; $Z_2Z_3Z_6$ again attains three. Without the repetition block, the two respective intermediate codes retain $Y_2X_4$ and $X_2Y_4$ at weight two. Among choices adding one two-qubit repetition block with Reed--Muller fixed at qubit $1$ and this fixed $C_B$, precisely a $Z$-basis repetition at qubit $2$ or $4$ raises both remaining weights. Its $X$ and $Y$ logical weights must both be two, so its weight-one logical Pauli must be $Z$. This is a restricted code-design statement, not a lower bound on all logical-$T$ circuits. In the chosen code, $Z_{2a}$ and $Z_{2b}$ differ by a stabilizer; their undetectability by the repetition check alone causes no ambiguity after the retained outer checks are measured.
\end{remark}

{Lemma~\ref{rep:rm} supplies $\lambda_i(X)=\lambda_i(Y)=7$ and $\lambda_i(Z)=3$ for a Reed--Muller inner block in Eq.~\eqref{rep:distance-formula}.
These values enter the distance proof of
Proposition~\ref{c22:distances}, while the representatives in Eq.~\eqref{rep:rm-reps} give the measurement weights in Eq.~\eqref{c22:weights}. The transversal action supplies the fifteen-$T^\dagger$ layer in 
Theorem~\ref{thm:main22}; its preservation of one-qubit error support is used in that theorem's single-fault proof.}

\subsubsection{Protected correction and logical measurement}\label{app:c22-correction}

Here we specify the correction and logical-measurement circuits represented by the EC and measurement boxes in Fig.~\ref{fig:gate22}.
The proofs include correction of an incoming error, the range on arbitrary inputs, and the decoded logical measurement instrument required for the protected gate.
\par

We use the local stochastic circuit noise model of Eq.~\eqref{eq:local-stochastic-budget}, with reliable classical processing and flexible two-qubit connectivity. A  single measurement of a Pauli check uses a verified cat, with one distinct cat qubit controlling the Pauli factor at each data {qubit}.

{For an arbitrary local fault, the one-error space is the span of states obtained from the code by the identity 
or weight-one Pauli errors; the error may remain correlated with a reference.}

\begin{definition}[Located failure and abort]\label{def:located-failure}
The \emph{classical record} $\omega$ of a circuit procedure consists of its observed measurement outcomes, including ancilla verification, syndrome and logical-measurement outcomes, together with their circuit-step and attempt labels.
A \emph{located failure} is a reported unsuccessful termination: a specified stopping rule computes a flag $f(\omega)=1$, identifies the procedure invocation that failed, and certifies no successful output from that invocation.
When this failure terminates a gate invocation, we call it a \emph{located abort}.
The qualifier \emph{located} refers to the identified procedure or gate invocation; it does not identify a faulty elementary operation or assert that the remaining quantum state is correct.
Rejection of one candidate or one gate attempt is not an abort if the prescribed procedure continues with a replacement or recovery and retry.
\end{definition}
\par

Here we describe a bounded serial realization of the correction circuits. 
For a weight-$w$ Pauli check, prepare a cat by $H$ on its first qubit followed by the chain of CNOTs $1\to2,2\to3,\ldots,w-1\to w$. 
Verify each adjacent $Z_jZ_{j+1}$ parity once with a freshly reset verification ancilla in $|0\rangle$, using the two cat qubits as controls and the ancilla as target. 
Reject a candidate if any parity is odd. 
After verification, each cat qubit controls a single data Pauli, and all cat qubits are measured in $X$. 
Try at most two sequential candidates per request, resetting the first before a replacement. 
Both rejections are a located failure (Definition~\ref{def:located-failure}).
Figure~\mbox{\ref{fig:verified-cat}} shows the cat-state preparation and parity verification used in each check measurement.

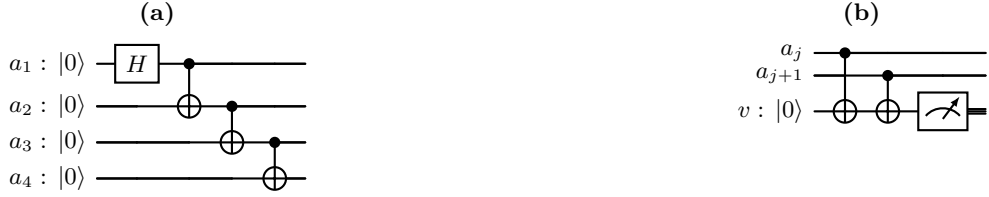
\begin{figure*}[t]
\centering
\begin{minipage}[t]{0.48\textwidth}
\centering
{\textbf{(a)}}\par\smallskip
{\begin{quantikz}[row sep=0.15cm,column sep=0.24cm]
\lstick{$a_1:\,|0\rangle$} & \gate{H} & \ctrl{1} & \qw & \qw & \qw \\
\lstick{$a_2:\,|0\rangle$} & \qw & \targ{} & \ctrl{1} & \qw & \qw \\
\lstick{$a_3:\,|0\rangle$} & \qw & \qw & \targ{} & \ctrl{1} & \qw \\
\lstick{$a_4:\,|0\rangle$} & \qw & \qw & \qw & \targ{} & \qw
\end{quantikz}}
\end{minipage}\hfill
\begin{minipage}[t]{0.48\textwidth}
\centering
\textbf{(b)}\par\smallskip
\begin{quantikz}[row sep=0.15cm,column sep=0.24cm]
\lstick{$a_j$} & \ctrl{2} & \qw & \qw & \qw \\
\lstick{$a_{j+1}$} & \qw & \ctrl{1} & \qw & \qw \\
\lstick{$v:\,|0\rangle$} & \targ{} & \targ{} & \meter{} & \cw
\end{quantikz}
\end{minipage}
\caption{
Verified cat state preparation. (a) A four-qubit cat state is prepared by a Hadamard gate followed by a CNOT chain; the same chain is extended to $w$ qubits for a weight-$w$ check. (b) Each adjacent parity $Z_{a_j}Z_{a_{j+1}}$ is measured with the verification ancilla $v$, which is measured in the computational basis and reset to $|0\rangle$ before each pair. The cat state is accepted only if every measured parity is positive. A rejected candidate is replaced once, and a second rejection is a located failure  (Definition~\ref{def:located-failure}). Controlled Pauli gates between the cat state qubits and the data are applied only after verification.
}
\label{fig:verified-cat}
\end{figure*}

Let $Q$ be an $n$-qubit stabilizer code of distance at least three, with independent signed generators $g_1,\ldots,g_m$.
Use the serial cat-state preparation and verification of Fig.~\ref{fig:verified-cat}, with at most two candidates per check and the local fault model of Eq.~\eqref{eq:local-stochastic-budget}.
An accepted cat state is a candidate for which every verification parity is reported as $+1$.
 A single measurement of a check means one use of an accepted cat state, the controlled single-qubit Pauli couplings, and the product of the final $X$-measurement outcomes. Repetition and comparison of these outcomes are separate steps.

The complete syndrome is $s=(s_1,\ldots,s_m)\in\mathbb F_2^m$, where the measured eigenvalue of $g_j$ is $(-1)^{s_j}$.
For each $s$, choose a Hermitian Pauli $R_s$ satisfying $g_jR_s=(-1)^{s_j}R_sg_j$ for every $j$.
Set $R_0=I$ and choose a weight-one Pauli whenever the syndrome has such a representative; for all other syndromes, fix any solution of these binary commutation equations.
This choice can be made by first tabulating the syndromes of $I$ and the $3n$ Paulis $X_j,Y_j,Z_j$, and then solving the remaining cases by binary Gaussian elimination.
Measure every generator once per complete round, without applying a recovery between rounds, and stop when two consecutive syndrome vectors agree.
For the concatenated codes, each round includes all inner and lifted outer checks.
Apply the corresponding $R_s$ as a product of single-qubit Paulis.
Use at most four rounds, declaring a located failure (Definition~\ref{def:located-failure}) if no consecutive pair agrees; a located failure is a classical flag that this correction procedure has not completed successfully.
Rejecting both cat candidates likewise gives a located failure.

\begin{lemma}[Serial cat and correction properties]\label{c22:cat}
For the code $Q$ and serial correction procedure specified above, the following properties hold.
\begin{enumerate}[label=\arabic*.]
    \item With at most one fault, the accepted cat state has at most one $X$ component modulo $X^{\otimes w}$, and a  single check measurement introduces at most one data error, although its reported sign may be wrong.
    \item With no internal fault, the procedure corrects one incoming physical error.
    With one internal fault on an encoded input without incoming errors, it preserves the logical state and leaves at most one residual physical error.
    \item For an arbitrary input, its output lies in $Q$ with no fault and in the one-error space of $Q$ with at most one fault.
    \item {With at most one fault, at least one of the two cat candidates is accepted for every check, and two consecutive complete syndrome records agree within four rounds; hence neither stopping rule declares a located failure.}
\end{enumerate}
\end{lemma}

\begin{proof}
First consider a Pauli fault; arbitrary local faults are treated at the end.
Write the bit-error pattern of a $w$-qubit cat state as $x\in\mathbb F_2^w$, so verification of $Z_{a_j}Z_{a_{j+1}}$ gives $(-1)^{x_j+x_{j+1}}$ when verification is ideal.
If the fault occurs in the cat encoder, all subsequent verification operations are ideal, and acceptance requires $x_j=x_{j+1}$ for every $j$.
Thus $x=0$ or $x=\mathbf1$; the latter is equivalent to zero because $X^{\otimes w}$ stabilizes the ideal cat state.
If the fault occurs during verification, the encoder is ideal.
A verification CNOT touches only one cat qubit and the verification ancilla, so it can introduce an $X$ component on at most that cat qubit.
An $X$ on the target ancilla does not propagate to a later cat control, whereas a $Z$ on that ancilla propagates only $Z$ components to cat controls.
Resetting the ancilla after each parity prevents propagation to later pairs.
Preparation, measurement, and idle faults obey the same bit-error bound.
Phase errors on the cat state can change the product of its final $X$ outcomes, so its reported check sign need not be correct.

Each cat qubit couples to only one data qubit.
An $X$ component on that cat qubit can therefore propagate only the corresponding single-qubit Pauli to the data.
A fault in a controlled Pauli likewise touches only that cat qubit and its data target, and no later coupling uses that cat qubit.
A data Pauli that anticommutes with the controlled Pauli propagates a $Z$ to the cat control, which cannot spread to another data qubit.
Consequently one  single check measurement introduces at most one data error, and an existing data error does not spread to additional data qubits.
This proves statement~1 for a Pauli fault.
A candidate prepared and verified without faults is always accepted; hence rejecting two candidates requires at least two faults.
This establishes the cat-candidate bound in statement~4.

Let $\Pi_s=\prod_{j=1}^m(I+(-1)^{s_j}g_j)/2$ be the projector onto syndrome $s$.
In a fault-free round the commuting check measurements project any input onto one of these sectors.
All subsequent fault-free rounds give the same syndrome until a data error occurs.
For a Pauli $E$, write $\syn(E)$ for its commutation syndrome, so that $E\Pi_s=\Pi_{s+\syn(E)}E$.
The chosen recovery satisfies $R_s\Pi_s=\Pi_0R_s$, and therefore maps the measured sector into $Q$.
This proves the arbitrary-input range claim in statement~3 with no fault.
If the input is $E|\psi\rangle$ with $|\psi\rangle\in Q$ and $\wt(E)\leq1$, the observed syndrome is $s=\syn(E)$ and $\wt(R_s)\leq1$.
The Pauli $R_sE$ has zero syndrome and weight at most two, so distance at least three implies that it is a stabilizer up to phase.
Thus the incoming error is corrected without changing the logical state, including when several weight-one errors share a syndrome.
This proves the fault-free part of statement~2.

Now suppose the sole fault occurs during syndrome extraction.
Only its containing round can have an incorrect or mixed syndrome record; the fault introduces at most one data Pauli $E$, and subsequent clean rounds have a fixed syndrome.
If the first round contains the fault, rounds two and three agree unless the procedure has already stopped.
If the second round contains the fault, rounds three and four agree unless it has already stopped.
If neither of the first two rounds contains a fault, their records agree and the procedure stops after round two.
A fault in a wait between rounds is assigned to the following round, while a fault after the stopping decision is an output or recovery fault.
This proves the four-round bound and completes statement~4 for a Pauli fault.

At least one round in the accepted equal pair is fault-free.
If its last round is fault-free, the final state has the reported syndrome $s$, so $R_s$ returns it to $Q$.
If its last round contains the fault, the preceding clean round has already projected onto syndrome $s$.
The final state is then $E|\phi_s\rangle$ with $\Pi_s|\phi_s\rangle=|\phi_s\rangle$ and $\wt(E)\leq1$, so $R_sE|\phi_s\rangle=\pm E R_s|\phi_s\rangle$ lies in $EQ$.
The same argument covers a fault after the last syndrome readout.
This proves the arbitrary-input range claim in statement~3 for a Pauli fault.
On a clean encoded input, the agreed syndrome is either $0$ before the error or $\syn(E)$ after it.
In the first case $R_0=I$ leaves only $E$; in the second, $R_{\syn(E)}E$ is a stabilizer by the distance-three argument above.
Thus neither case changes the logical state.
If the sole fault is instead in the final product of single-qubit recovery gates, all syndrome rounds are clean and that fault leaves at most one physical error after the ideal recovery.
This completes statements~2 and~3 for a Pauli fault.

Finally, expand each Kraus operator of an arbitrary local fault in Paulis on the qubits at its location.
The preceding arguments hold for each term and each recorded branch; by linearity their outputs lie in the span of $Q$ and its single-qubit-error spaces, with the stated logical action whenever the input is correctable.
Tensoring every step with the identity on a reference gives the same conclusions for entangled inputs.
This extends statements~1--4 to arbitrary local faults.
\end{proof}

Logical measurement records are compared only after accounting for every intervening Pauli recovery. Ordinary error correction during these measurements uses only the twenty retained checks; it must not fix the additional logical measurement outcome.

\begin{lemma}[Protected Pauli measurement]\label{pm:measurement}
For a stabilizer code $Q$ of distance at least three and Hermitian $P\in\cN(\cS_Q)$, use three measurements of $P$ using verified cat states, with $Q$ error correction as in Lemma~\ref{c22:cat} before the first and after each measurement. Take the majority of the three outcomes in the corrected Pauli frame. This implements the ideal decoded logical measurement, with at most one residual physical error, whenever the number of incoming physical errors plus internal faults is at most one.
\end{lemma}
\begin{proof}
Write $\Pi_\lambda(P)=(I+\lambda P)/2$, with $\lambda\in\{+1,-1\}$.
Since $P$ commutes with the stabilizers of $Q$, each projected state $\Pi_\lambda(P)|\psi\rangle$ remains in $Q$ for $|\psi\rangle\in Q$.
Hence the ordinary correction circuits of Lemma~\ref{c22:cat} apply after each measurement without selecting a sign of $P$.
The ideal measurement has probability $\operatorname{Tr}(\Pi_\lambda(P)\rho)$ and unnormalized output $\Pi_\lambda(P)\rho\Pi_\lambda(P)$ for every encoded density operator $\rho$.

First suppose all circuit locations are ideal.
The leading correction removes a possible incoming error.
The  first measurement projects onto one eigenspace of $P$, the intervening corrections act trivially on the encoded state, and the next two measurements give the same sign $\lambda$.
Thus the majority and the quantum output agree with the ideal measurement.

Now suppose the input is clean and there is one Pauli fault in a  single check measurement.
Propagating it through the remaining Clifford gates gives an output data Pauli $F$ of weight at most one and possible flips of the cat readouts.
For reported sign $r$, the corresponding branch is proportional to $F\Pi_{\epsilon r}(P)$, where $\epsilon\in\{+1,-1\}$ accounts for these flips and Pauli commutations.
The following clean correction removes $F$, leaving the ideal projected state with sign $\lambda=\epsilon r$.
If this is the first  single check measurement, the next two clean measurements both give $\lambda$.
If it is the second or third extraction, the earlier clean measurement has already fixed $\lambda$, and the other  clean measurement record has the same value.
Only the record of the faulty extraction can disagree, so the majority is correct and the final state has the required projection.

If the fault is in the leading correction or a correction  between these measurements, Lemma~\ref{c22:cat} leaves at most one data Pauli $F$ and preserves the logical state.
Define $\epsilon_F\in\{+1,-1\}$ by $PF=\epsilon_FFP$.
The next  clean measurement obeys $\Pi_r(P)F=F\Pi_{\epsilon_F r}(P)$, and the following clean correction removes $F$.
Thus only that  next measurement record can be reversed relative to the logical sign $\lambda$; the other two records and the corrected quantum output have sign $\lambda$.
A fault in a wait has the same effect as a one-qubit error immediately before the next operation.
A fault in the final correction leaves at most one residual error and cannot change the already determined majority.

The comparisons above are made in a common Pauli frame.
If a known Pauli $E$ is tracked rather than physically removed, define $\epsilon_E\in\{+1,-1\}$ by $PE=\epsilon_E EP$, and replace the  measured sign $r$ by $\epsilon_E r$ before comparison.
Indeed, $\Pi_r(P)E=E\Pi_{\epsilon_E r}(P)$.
This accounts for every intervening tracked recovery and does not require knowledge of the unknown fault.

For an arbitrary local fault, expand each Kraus operator in local Paulis.
For each fixed full record and final majority $\lambda$, the corrected logical action of every term above is a scalar multiple of $\Pi_\lambda(P)$, up to a correctable final error.
After ideal decoding, the branch operators therefore have the form $c_{\lambda,\omega}\Pi_\lambda(P)$, where $\omega$ includes the other recorded outcomes and any decoder syndrome.
With at most one fault the finite procedure has no located failure (Definition~\ref{def:located-failure}), so trace preservation gives $\sum_\omega|c_{\lambda,\omega}|^2=1$ for each nonzero eigenspace.
Summing these branches yields exactly $\Pi_\lambda(P)\rho\Pi_\lambda(P)$ and its Born probability, with the possible single-error residual already bounded above.
This also proves the statement for inputs entangled with a reference, because the same operator identities hold after tensoring with the identity on that reference.
\end{proof}

\subsubsection{Verified stabilizer resources and Clifford gates}

Here we construct verified ancillary stabilizer states and use them for protected logical Clifford operations, including transfers between encodings.
These operations supply the Clifford part of the universal fault-tolerant gate set stated in Sec.~\ref{sec:intro}.
\par

Here verification means testing all independent stabilizer generators of a known target state, together with the flag checks specified below, and accepting the candidate only when every generator has its prescribed sign and every flag is zero.

\begin{lemma}[Verification of a known stabilizer state]\label{rep:verification}
An $m$-qubit stabilizer state can be prepared and verified with two extra qubits so that a single fault either rejects the candidate or leaves at most one physical error on an accepted candidate. 
Prepare and verify one candidate, and, only if it is rejected, reset the same qubits and prepare and verify a second candidate; accept the first successful candidate and declare a located failure  (Definition~\ref{def:located-failure}) if both are rejected.
If $L_{\rm cand}$ is the maximum number of locations in one prescribed preparation-and-verification attempt, including resets and all data and ancillary idles, then each branch has at most $2L_{\rm cand}$ locations, and rejecting both candidates requires at least two faults.

\end{lemma}
\begin{proof}
Prepare the candidate by a fixed Clifford encoder. For each independent signed generator $P=\prod_jP_j$, prepare $a$ in $|+\rangle$ and $f$ in $|0\rangle$, apply $\operatorname{CNOT}_{a\to f}$, then $\prod_j\operatorname{controlled}_a(P_j)$, then $\operatorname{CNOT}_{a\to f}$. Measure $a$ in $X$ and $f$ in $Z$, accepting only the prescribed generator sign and zero flag. Reset both verification ancillas between checks. A fault in the encoder is followed by fault-free full stabilizer verification, which rejects its propagated Pauli unless it stabilizes the target state. With a clean encoder, a sole verification fault can spread an $X_a$ error to later data, but also flips the flag. A fault at the first flag CNOT can hide that flip only while propagating the complete measured stabilizer, harmless on the ideal state. A last-flag-CNOT fault has no remaining data interaction. A data error propagates only a $Z_a$ component back to the control and does not spread to other data. Thus every accepted case leaves at most one data error modulo the measured stabilizer. Later clean checks cannot spread it.
A fault-free attempt is always accepted, and a full reset prevents a fault in the first attempt from spoiling a fault-free second attempt; therefore two rejections require at least two faults.
 Pauli expansion proves the arbitrary-fault statement. This verifier applies to a known ancillary stabilizer state; it is not a rule for postselecting unknown data.
\end{proof}

\begin{lemma}[Protected Clifford transfer]\label{rep:transfer}
Let $Q$ and $Q'$ be stabilizer codes with parameters $\qcode{n}{k}{d}$ and $\qcode{n'}{k}{d'}$, respectively, where $d,d'\geq3$.
For Clifford encodings $E_Q$ and $E_{Q'}$, let $U:Q\to Q'$ be any encoded Clifford isometry, meaning that $UE_Q=E_{Q'}C$ for a $k$-qubit Clifford unitary $C$.
Such a $U$ has a finite stabilizer-only teleportation implementation that corrects one incoming error with no internal fault, or tolerates one fault on a clean input, leaving at most one output error.
\end{lemma}
\begin{proof}
For a Clifford encoding $E_Q$, prepare and verify the stabilizer state
\begin{equation}
 |\Omega_U\rangle=2^{-k/2}\sum_{x\in\mathbb F_2^k}E_Q^*|x\rangle\otimes UE_Q|x\rangle.
\label{rep:resource}
\end{equation}
Here $*$ denotes entrywise complex conjugation in the physical and logical computational bases, and
\[
Q^*=\{|\psi\rangle^*:|\psi\rangle\in Q\}=\operatorname{im}E_Q^*.
\]
If $Q$ has signed stabilizer generators $g_j$, then $Q^*$ has generators $g_j^*$; a generator with an odd number of $Y$ factors changes sign because $Y^*=-Y$.
Thus $Q^*$ has the same code parameters as $Q$, although its signed stabilizers can differ.
\par

Resource state $|\Omega_U\rangle$ is obtained from $k$ Bell pairs by the Clifford encodings $E_Q^*$ and $E_{Q'}C$, so it is a known stabilizer state to which Lemma~\ref{rep:verification} applies.
The first half of $|\Omega_U\rangle$ is encoded in $Q^*$, as required by the Bell contraction. 
At each coordinate, apply a CNOT from the input qubit to the corresponding qubit of this half, then a Hadamard on the input qubit and a $Z$ measurement of both qubits. A single input error, measured-half resource error, or fault in one such pairwise circuit affects only that coordinate's two Bell bits. These changes are equivalent to one Pauli error in the complete Bell record, whose logical outcome can therefore be recovered because $d\geq3$.

The full Bell measurement outcomes, together with a distance-three decoder, determine the logical byproduct uniquely: any two error hypotheses of weight at most one with the observed syndrome differ by a stabilizer, so they give the same corrected logical Bell measurement outcomes.
For syndromes without such an error hypothesis, fix any Pauli representative with that syndrome, as in Lemma~\ref{c22:cat}, so the decoder and its logical Pauli correction are defined for every set of Bell measurement outcomes.
 Conjugate the decoded logical byproduct by $U$ and apply its physical Pauli representative on $Q'$. An output-half resource error or a fault in that Pauli correction leaves one output error. Lemma~\ref{rep:verification} bounds the total resource error across both halves. These cases prove the claim by Pauli linearity. With $U=I$ this is Knill error correction; ideal Bell contraction projects arbitrary inputs into the output code, and one faulty component leaves that output in the span of the code and its single-qubit-error spaces, supplying the standard error-correction range property.
Pauli expansion covers arbitrary local faults, and tensoring these branch identities with the identity on a reference system proves the same statements for inputs entangled with that system.
\end{proof}

The same transfer construction implements an encoded Clifford on several blocks: verify its encoded Bell resource, decode the Bell record separately on each input block, and apply the resulting Pauli byproducts to the output blocks. With no internal fault, one incoming error in each input block is corrected. On clean encoded inputs, one internal fault leaves at most one error in each output block; in this verifier it leaves at most one error across the entire resource. The output also has the stated range property for arbitrary inputs. This supplies protected $H$, $S$, and CNOT gates on the fixed code, including stabilizer generators with complex matrix entries. Every resource is a known stabilizer state, so preparation and verification use only Clifford operations~\cite{gottesman1999demonstrating,knill2005quantum}. There is no direct fanout from an unencoded control to unverified encoded data. Pauli byproducts are resolved before a non-Clifford layer, physically or by adapting its rotation signs.

\subsubsection{Proof of the protected fixed-code gate (Theorem~\ref{thm:main22})}\label{app:gate_single_fault_proof}

Here we assemble the transversal layer, error correction using the retained checks and adaptive Pauli measurements into the circuit of Fig.~\ref{fig:gate22}.
We prove all claims of Theorem~\ref{thm:main22}: the logical action, fifteen physical $T$-type gates, single-fault guarantee, 33-qubit space bound, and state preparation.

The proof combines the transversal action in Lemma~\ref{rep:rm} and the ideal completion in Theorem~\ref{pm:adaptive} with the correction and measurement guarantees in Lemmas~\ref{c22:cat} and~\ref{pm:measurement}.
\par

On an input state encoded in $\cC'_{22}$, we consider a circuit performing the following steps.
\begin{enumerate}[label=\arabic*.]
    \item Apply error correction using the full stabilizer for $\cC'_{22}$.
    \item Apply the fifteen physical $T^\dagger$ gates on Reed-Muller block $1$.
    \item Apply error correction using the retained checks for $\cD'_{22}$.
    \item Measure $\widehat G'$ and record its outcome $y$.
    \item If $y=-1$, measure $\widehat M'$ and record its outcome $r$.
    \item Measure $\widehat h'$ and record its outcome $z$.
    \item Apply $\widehat A'^{(1-z)/2}$, and additionally apply $\widehat L'$ when $y=-1$ and $rz=-1$.
    \item Apply error correction using the full stabilizer for $\cC'_{22}$.
\end{enumerate}
Every displayed logical measurement is the full protected measurement of Lemma~\ref{pm:measurement}. 
All ordinary syndrome extraction during these measurements uses only the twenty retained checks. 
In particular $h'$ is not fixed prematurely.
\par

\begin{proof}[Proof of Theorem~\ref{thm:main22}]
\emph{Logical action and single-fault correctness.}
The transversal layer in step~2 implements $R_{\widehat B'}(\pi/4)$ by Lemma~\ref{rep:rm}. It preserves $\cD'_{22}$ because the lifted outer retained generators have only $I$ or $Z$ at qubit $1$. The exact gate then follows from Theorem~\ref{pm:adaptive} applied in steps~4--7 to $\widehat A'$, $\widehat B'$, $\widehat L'$, $\widehat h'$, $\widehat G'$, and $\widehat M'$, obtained by conjugating the outer operators as in Eq.~\eqref{c22:operators} and then lifting them as in Eq.~\eqref{c22:lift}.
For fault tolerance, fault-free error correction in step~1 removes one incoming error. 
If step~1 contains the sole fault, its residual error stays on one qubit during the transversal layer in step~2, including an error on an unencoded or repetition qubit, which is untouched by the transversal layer. 
A fault at a physical $T^\dagger$ in step~2 also remains supported on one qubit. The following clean error correction using the retained checks in step~3 removes it. 
If that correction is faulty, the fault-free correction at the start of the protected $\widehat G'$ measurement in step~4 removes the residual. Lemma~\ref{pm:measurement} handles each measurement in steps~4--6 and its record; successive components have clean leading corrections if the earlier component contains the sole fault. 
Final Pauli corrections in step~7 do not spread errors. Trailing {error correction using the full stabilizer} in step~8 removes any preceding residual, or leaves at most one if it is itself faulty. 
A single fault cannot cause both allowed ancilla-preparation attempts to fail. Pauli expansion covers arbitrary single-location noise and entangled inputs.

\emph{Space and state preparation.}
Every retained check, check of $\cC'_{22}$, or displayed logical check has weight at most ten by Proposition~\ref{c22:distances}. All check requests, correction rounds, and candidate attempts are serial. The verification ancilla is reused after reset;  no extra cat state or output resource is prepared for a later step. Repetition of a measurement increases time but does not increase this simultaneous space.
The maximum simultaneous occupancy is therefore $22+10+1=33$ qubits: the data block, at most ten cat qubits and one verification ancilla.
For a logical $|+\rangle$, prepare the $22$-qubit stabilizer state specified by the checks of $\cC'_{22}$ and $X_{\rm log}=+1$ using a fixed Clifford encoder. 
Verify its complete known stabilizer with the same cats and reject any incorrect outcome. The added logical check has weight four. A fault in the encoder is followed by fault-free complete verification, so an accepted propagated Pauli is a stabilizer. A sole verification fault either rejects or leaves one data error, removed by the fault-free correction in step~1. Two sequential complete state candidates make this step finite without increasing peak space. Thus the accepted logical state followed by the gate is $T|+\rangle$, with the same single-fault guarantee.
\end{proof}

The $33$-qubit bound concerns a serial gate or state preparation on the fixed code. Conversion to another memory code, routing on restricted connectivity, and parallel preparation have separate costs. For comparison, the generic verified-teleportation construction of Lemma~\ref{rep:transfer} transfers Steane to Reed--Muller and back using a $7+15=22$-qubit resource: the two transfer peaks are $7+22+2=31$ and $15+22+2=39$, while final Steane correction needs $7+14+2=23$. These give a $39$-qubit upper bound for that schedule only. Specialized code-switching circuits use fewer qubits; Ref.~\cite{heussen2024efficient} reports a 24-qubit implementation with flag-qubit reuse. The $33$-qubit bound therefore does not establish a qubit advantage over specialized code switching.

To count circuit locations, let $C_Q,C_D$ denote the elementary-location costs of the stated correction circuits for $\cC'_{22}$ and $\cD'_{22}$, and let $C_{\rm meas}(w)=4C_D+3C_{\rm raw}(w)$. A padded longest branch has
\begin{equation}
 N\leq 2C_Q+C_D+15+C_{\rm meas}(6)+C_{\rm meas}(10)+C_{\rm meas}(10)+C_{\rm other}.
\label{c22:cost}
\end{equation}
This deliberately retains separate neighboring corrections. The optional weight-ten measurement occurs with probability $1/2$ ideally. Its three check measurements, together with those of the other measurements, use $48$ or $78$ data couplings, averaging $63$, excluding all syndrome extraction and cat verification. In a raw weight-$w$ cat attempt, the encoder uses $w-1$ CNOTs, verification uses $2(w-1)$ CNOTs, and data coupling uses $w$ controlled Paulis; hence the fault-free count is $4w-3$ before local basis changes, measurements and idles. Capped retries and delays must be included in $C_{\rm raw}$ and $C_{\rm other}$. The gate-failure-or-abort probability is bounded by $\binom N2p^2$, but no useful numerical coefficient is assigned without counting the full schedule.

\subsubsection{A finite set of circuits and error suppression by recursive concatenation}\label{app:ft-budget}

Here we recursively concatenate the fixed code $Q_{22}=\cC'_{22}$: each data qubit is replaced by a 22-qubit block, and each elementary circuit location by the corresponding protected encoded circuit.
Corollary~\mbox{\ref{c22:recursion}} derives the error recurrence in Eq.~\mbox{\eqref{eq:recursive-error}}, proving the arbitrary accuracy below a positive threshold stated in Sec.~\mbox{\ref{subsec:gate22}}.

The same argument gives the $22^\ell$ data-qubit and $15^\ell$ physical $T$-type counts stated there, excluding ancilla and Clifford overhead.
The $T$-type counts refer to the scheduled locations; a branch stopped after failure executes at most this number of gates.
Equations~\mbox{\eqref{eq:computation-error-budget}} and~\mbox{\eqref{eq:target-level}} then relate the concatenation level to a target failure-or-abort probability for a computation.

\par

The primitive properties used in the concatenation argument of Corollary~\mbox{\ref{c22:recursion}} below are summarized in the following table.

\begin{center}
\begin{tabular}{p{0.39\columnwidth}p{0.53\columnwidth}}
Required property & Proof\\\hline
Correction and arbitrary-input range & Lemmas~\ref{c22:cat} and~\ref{rep:transfer}\\
Verified stabilizer preparation & Lemma~\ref{rep:verification}\\
Clifford gates and blockwise decoding & Lemma~\ref{rep:transfer} and the following paragraph\\
Logical measurement and output & Lemma~\ref{pm:measurement}\\
Logical $T$ and its inverse & Theorems~\ref{thm:main22} and~\ref{pm:adaptive}, with the inverse rule\\
Bounded preparation and correction & Two-candidate and four-round caps in Lemma~\ref{c22:cat}
\end{tabular}
\end{center}

To implement the inverse logical unitary $R_{Z_{\rm log}}(-\pi/4)=R_{Z_{\rm log}}(\pi/4)^\dagger$ on the same fixed code $\cC'_{22}$, replace fifteen physical $T^\dagger$ gates in the circuit implementing $R_{Z_{\rm log}}(\pi/4)$ by $T$ gates, and use the additional $\widehat L'$ correction when $y=-1$ and $rz=+1$.
The inverse rule following Theorem~\ref{pm:adaptive} proves the action. Complex conjugating the entire circuit would also conjugate the stabilizer generators of the fixed code with complex matrix entries, so it would not by itself prove an inverse gate on this fixed code.  All other circuits in this set, including resource encoders and verification, are Clifford circuits. Replacing their elementary locations by the corresponding encoded Clifford circuits  gives a closed set of circuits without extra non-Clifford resources.

\emph{Completion after a reported failure.}\label{app:recursive-failures}
Fix each invocation's output registers and return time in advance. A gate or correction performed on the input register returns that register. Reserve the prescribed output half for teleportation even when resource preparation fails; preparation and reset likewise have designated output registers. A destructive measurement returns only its bit, while a nondestructive measurement also retains its data. Hold all ordinary output bits until the return time.

On reaching its own stopping failure, an invocation stops its remaining intended operations, passes its designated output registers to the padding idles, sets each required outgoing bit to zero (sign $+1$), and uses the identity for any unreported Pauli correction. Other registers are discarded. No additional attempt is made to repair or replace the output, whose state need not be encoded or correct. The parent continues using these ordinary outputs; the child failure flag is recorded only for diagnosis and does not enter a parent stopping rule. Only an outermost invocation's own stopping failure aborts the computation. A discarded register has no subsequent coupling before its next scheduled preparation or reset, itself subject to this rule.

At the elementary level padding uses physical idles; at higher levels it uses the protected idle circuits in this set, with the same completion rule at the lower level. This ends at physical idles and is finite at every level. Every initialization, preparation attempt, verification, reset, conditional operation and wait is included in the schedule, including waits on reserved outputs and the full padding circuits after failure. All these slots contribute to $N_\star$; no additional protected reset or recovery procedure is assumed.

A rectangle is a protected logical operation followed by error correction; its extended rectangle also includes leading error correction.
Shared corrections are assigned by processing extended rectangles backwards in circuit order.
When a later bad extended rectangle includes a shared correction, remove that correction and its assigned lower-level locations from the earlier extended rectangle.
Apply this assignment at every lower level, excluding from an earlier truncated extended rectangle every elementary location assigned to a later one.
At level one, call an assigned extended rectangle bad if it contains at least two physical faults; recursively, call it bad if it contains at least two bad child extended rectangles with disjoint assigned supports.
A rectangle is called bad when its assigned, possibly truncated extended rectangle is bad.
For every set $R$ of level-$\ell$ rectangle positions fixed independently of the fault pattern, $\eta_\ell$ is an upper bound satisfying $\Pr(\text{all positions in }R\text{ are bad after level reduction})\leq\eta_\ell^{|R|}$.
\par

\begin{corollary}[Arbitrary accuracy on the fixed code]\label{c22:recursion}
Use the fixed code $Q_{22}=\cC'_{22}$, the gate of Theorem~\ref{thm:main22}, {the logical $T^\dagger$ gate on the same code specified above}, and the protected Clifford, preparation, measurement, and correction circuits specified here. 
For each requested cat state or known stabilizer resource, use at most two sequential preparation-and-verification attempts as in Lemmas~\ref{c22:cat} and~\ref{rep:verification}, and limit repeated syndrome extraction to four complete rounds as in Lemma~\ref{c22:cat}.
Declare a failure if both candidates reject or no consecutive pair of complete syndrome records agrees, and pad shorter branches with idles.
Below the outermost level, handle this failure by the local completion rule of \SM{}~\ref{app:recursive-failures}.
Fix a padded finite schedule for  each circuit in this set, with a specified location slot for every possible preparation attempt, verification, conditional operation, and required idle, and let $N_\star\geq2$ bound the number of lower-level location slots in every extended rectangle.
This includes all initialization, reset and failure-branch slots specified in \SM{}~\ref{app:recursive-failures}.
 Set
\begin{equation}
 A=\binom{N_\star}{2},\qquad p_{\rm th}=A^{-1},\qquad\eta_0=p.
\label{eq:pair-bound}
\end{equation}
Under Eq.~\eqref{eq:local-stochastic-budget}, the joint local-stochastic strength of bad truncated level-$\ell$ rectangles obeys
\begin{equation}
 \eta_{\ell+1}\leq A\eta_\ell^2,\qquad
 \eta_\ell\leq A^{-1}(Ap)^{2^\ell}
 =p_{\rm th}(p/p_{\rm th})^{2^\ell},\qquad p<p_{\rm th}.
\label{eq:recursive-error}
\end{equation}
The level-$\ell$  code used before and after the gate has $22^\ell$ data qubits and distance at least $3^\ell$. Its logical-$T$ circuit contains $15^\ell$ scheduled physical $T$-type gates. A bad level-$\ell$ rectangle requires at least $2^\ell$ physical faults; this is a lower bound, not an assertion that a malignant set of exactly that size exists.
\end{corollary}
\begin{proof}
Lemma~\ref{c22:cat} supplies both correction and arbitrary-input range: a fault-free correction maps arbitrary inputs into the code, and one faulty correction maps them into its one-error space. Lemma~\ref{rep:transfer} supplies the same properties for Knill correction and protected Clifford gates, with separate decoding on each input block. Verified known logical stabilizer states give preparation rectangles; Lemma~\ref{pm:measurement}, followed by discarding the block when appropriate, gives measurement rectangles. Idles are enclosed by correction. These finite circuits supply the standard preparation, measurement, gate, and error-correction properties for distance-three extended rectangles. Theorem~\ref{thm:main22} supplies the non-Clifford gate property, including one incoming error or one internal fault.

The local completion acts only on an invocation's inputs and reserved ancillary registers, followed by tracing out unused registers. Summing its recorded branches gives a trace-preserving quantum operation with the prescribed quantum and classical outputs, also for inputs entangled with a reference. Faults in its local padding cannot enlarge these supports, and neither an early output bit nor a failure flag controls another invocation. Thus a failed child is one arbitrary faulty operation on its participating blocks after level reduction. Its output need not be encoded; the arbitrary-input correction and range properties apply. The completion branch is never reached with at most one fault by Lemmas~\ref{c22:cat}--\ref{rep:transfer}, so their single-fault properties are unchanged.

Apply the standard backwards truncation and level-reduction procedure for overlapping extended rectangles~\cite{aliferis2006quantum}. A correction shared by neighboring rectangles is assigned in that procedure rather than counted as two independent sources of noise.
Choose ideal decoders with the same syndrome-to-Pauli recovery convention as the corrections above.
In checking decoded correctness, restore removed trailing corrections level by level by ideal corrections; inserting the corresponding ideal correction immediately before an ideal decoder does not change its decoded output.
The correction and range properties above therefore apply to the completed rectangle, and the restored corrections add no bad child locations.
 The stated rectangle properties imply that an incorrect truncated parent rectangle contains at least two bad child rectangles. Exhaustion of two candidate preparations also requires two bad children; failure of the complete-record stopping rule is likewise counted as a located failure  (Definition~\ref{def:located-failure}) and cannot occur with one bad child. 
These are the parent's own stopping tests: a child flag triggers neither, and an unsuccessful child is already included as one local faulty operation. Hence an outermost abort also requires a bad parent rectangle.
Fix $r$ parent-rectangle positions independently of the fault pattern.
On the event that all are bad after level reduction, the procedure supplies disjoint assigned child supports.
 Choose a pair of bad children in each parent. The joint lower-level local-stochastic bound is at most $\eta_\ell^{2r}$ for each such choice, and there are at most $A^r$ choices. Thus the joint upper bound is $(A\eta_\ell^2)^r$, establishing the recurrence without assuming independent failures of overlapping rectangles. Counting all pairs also covers fault sets of size three or more; a count of malignant pairs alone would need a separate treatment of larger minimal malignant sets.

Induction gives Eq.~\eqref{eq:recursive-error}. Concatenation replaces each data qubit by a distance-three 22-qubit block, giving $22^\ell$ data qubits and distance at least $3^\ell$. Each non-Clifford gate has exactly fifteen scheduled child $T$-type gates, and all remaining circuits are Clifford, giving $15^\ell$ scheduled physical $T$-type locations. Applying the two-child argument at each level gives the physical-fault lower bound.
\end{proof}

For a computation with $G$ logical locations, using local completion and counting only outermost stopping failures as aborts, the failure-or-abort probability is at most
\begin{equation}
 P_{\rm fail}\leq G\eta_\ell.
\label{eq:computation-error-budget}
\end{equation}
For $0<\epsilon<Gp_{\rm th}$ and $0<p<p_{\rm th}$, it suffices to choose
\begin{equation}
 \ell=\max\!\left\{0,\left\lceil\log_2\!\left(
 \frac{\ln(Gp_{\rm th}/\epsilon)}{\ln(p_{\rm th}/p)}\right)\right\rceil\right\}.
\label{eq:target-level}
\end{equation}
If $Q$ is the maximum number of logical data qubits and $G_T$ is the number of logical $T$-type gates, the corresponding counts are $Q22^\ell$ data qubits and $G_T15^\ell$ scheduled physical $T$-type gates, excluding ancillas and Clifford volume. If each rectangle has at most $V$ child locations, the complete circuit volume is at most a constant times $GV^\ell$. The $33$-qubit peak is a level-one serial bound; it does not imply total space $33^\ell$. No numerical threshold, malignant-pair coefficient, or cycle time is determined without the complete elementary schedule.

The signed generators, syndrome table, witnesses, correction rule and branch identities above specify the data used in the analytical code and circuit proofs.

\subsubsection{A protected 35-qubit example}

Here we give the related selective-concatenation constructions mentioned in Sec.~\ref{sec:selective}, using Reed--Muller blocks on the support of $B$.
Remark~{\ref{pm:ft}} encodes outer Steane qubits $1$ and $4$ to obtain a protected 35-qubit gate, and Corollary~{\ref{pm:support}} extends the circuit to a $w$-qubit rotation support under its intermediate-distance hypothesis.

For the 35-qubit example, Fig.~\mbox{\ref{fig:support35}} shows the two transversal CNOT layers and intervening inner correction that distinguish its rotation circuit from the single transversal layer in Fig.~\mbox{\ref{fig:gate22}}.

\begin{remark}[Encoding both qubits of the rotation support]\label{pm:distances}\label{pm:ft}
Encode outer qubits $1,4$ in Reed--Muller blocks, using the fixed basis $V_1=V_{\rm RM}R_X(\pi/2)$ at qubit $1$ and the ordinary basis at qubit $4$. Thus $(X,Y,Z)_1$ maps to $(\bar X,\bar Z,-\bar Y)_1$, and $\widehat B=\bar Z_1\bar Z_4$. This gives  $\qcode{35}{1}{3}$ and $\qcode{35}{2}{3}$ codes before and during the gate, respectively. Each weight-two outer intermediate logical listed above meets $\{1,4\}$; Eq.~\eqref{rep:distance-formula} gives the lower bound, and $Z_2Z_3Z_6$ gives equality, also for the $G$-, $M$-, and $h$-fixed codes by the witness argument above. The fixed encoding signs are part of the code definition.

Apply a transversal CNOT layer $4\to1$, inner correction on both blocks, fifteen $T^\dagger$ gates on block $1$, inner correction, the inverse CNOT layer, and inner correction. A faulty CNOT leaves at most one error in each block, removed by the next clean inner corrections. A faulty $T$ stays on one qubit; an error from a faulty correction remains single-qubit through the next transversal rotation or spreads to at most one qubit per block through the next CNOT, and is then corrected. Unencoded qubits are outside the rotation support. Complete the gate with leading and trailing  error correction using the full stabilizer and Theorem~\ref{pm:adaptive}, using Lemma~\ref{pm:measurement} on the retained code. This proves the one-fault claim with fifteen $T$-type gates and thirty rotation-circuit CNOTs. These counts exclude extraction and ancillary resources.

\begin{figure}[t]
\centering
\resizebox{0.45\columnwidth}{!}{%
\begin{quantikz}[row sep=0.24cm,column sep=0.18cm]
\lstick{$\mathrm{RM}_1$} & \targ{} & \gate{\mathrm{EC}} & \gate{(T^\dagger)^{\otimes15}} & \gate{\mathrm{EC}} & \targ{} & \gate{\mathrm{EC}} & \qw \\
\lstick{$\mathrm{RM}_4$} & \ctrl{-1} & \gate{\mathrm{EC}} & \qw & \gate{\mathrm{EC}} & \ctrl{-1} & \gate{\mathrm{EC}} & \qw
\end{quantikz}}
\caption{Protected non-Clifford rotation for the 35-qubit example. The wires $\mathrm{RM}_1$ and $\mathrm{RM}_4$ represent the 15-qubit Reed--Muller blocks encoding outer qubits $1$ and $4$, respectively. Each displayed CNOT is a transversal layer of fifteen CNOTs, and each $\mathrm{EC}$ is protected inner-code error correction. The five unencoded outer qubits are untouched by this rotation. Leading and trailing  error correction using the full stabilizer and  completion by protected Pauli measurements and corrections are specified in Remark~\ref{pm:ft}. 
}
\label{fig:support35}
\end{figure}
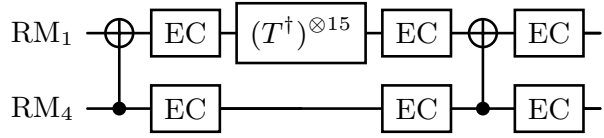

\end{remark}

\begin{corollary}[Encoding the non-Clifford support]\label{pm:support}
Consider Pauli $B$, which is the factor in Eq.~\eqref{eq:factor} used for the non-Clifford rotation $R_B(\pi/4)$, and $w=\wt(B)$ for the $n$-qubit outer code.
Encode precisely the $w$ qubits of $B$ in Reed--Muller blocks, choosing fixed bases mapping its factors to logical $Z$.
If the lifted intermediate code has distance at least three, the measurement construction has a one-fault implementation on $n+14w$ data qubits, with fifteen physical $T$-type gates and $30(w-1)$ physical CNOTs in its rotation circuit. The distance hypothesis holds automatically if the original intermediate distance is at least three.
\end{corollary}
\begin{proof}
Begin with  error correction using the full stabilizer, use inner-code correction after every transversal CNOT layer and after the transversal non-Clifford layer, and finish the  completion by adaptive Pauli measurements and corrections with  error correction using the full stabilizer.
Only inner checks are measured inside the parity circuit, because the intermediate parity maps need not preserve the lifted outer checks.

A parity map uses $w-1$ transversal CNOT layers and its inverse uses another $w-1$, each followed by inner correction. A faulty transversal two-block location leaves at most one error in each touched block, which the following clean inner corrections remove. A sole fault in correction leaves at most one residual in its block; each later transversal layer sends it to at most one qubit per touched block, again followed by clean correction. The transversal non-Clifford layer preserves single-qubit support. Thus every single fault is corrected before an unprotected spread can occur. All remaining operations are protected measurements on the lifted intermediate code. The distance formula of Proposition~\ref{rep:weights} shows that a lift cannot decrease distance when every nonidentity inner logical has weight at least one.
\end{proof}

For the  Golay factorization, $w=4$ gives $79$ data qubits, $90$ CNOTs in the rotation circuit, and fifteen $T$-type gates. 
This support-encoding example is distinct from the protected 23-data-qubit Golay circuit of Sec.~\ref{sec:golay}, which we will discuss in detail in \SM{}~\mbox{\ref{app:golay-completion}} below.

\section{Golay error detection and protected implementation}\label{app:golay}

Here we give the filtering and circuit proofs supporting Sec.~\ref{sec:golay}, using the Golay code data already specified in \SM{}~\ref{app:golay-details}.
We first prove ideal filtering, exact restart recovery and the calibration limitation.
The direct circuit is then specified completely, including  completion by Pauli measurements and corrections, resource counts and rejected-branch recovery, before its single-fault theorem and probability bound.
\par

\subsection{Exact filtering and restart recovery}\label{app:golay-extension}

All measurements in this subsection are ideal. Let $\cC$ be a pure stabilizer code of distance $d$, let $B$ have nonzero syndrome and support $\Omega$ of size $w<d$, and choose $h\in\cS$ anticommuting with $B$. The retained stabilizer $\cS_0$ consists of the elements of $\cS$ commuting with $B$.
Set $\cD=\cC\oplus B\cC$, the code stabilized by $\cS_0$.
 Write

\begin{equation}
\begin{gathered}
    U_\theta=R_B(\theta),\qquad
    G_\theta=U_\theta hU_\theta^\dagger
           =h\cos\theta-iBh\sin\theta,\\
    P_0=P_{\cD},\qquad
    Q_\theta=P_0\frac{I+G_\theta}{2}
       =U_\theta\PC U_\theta^\dagger.
\label{ge:full-projector}
\end{gathered}
\end{equation}

\phantomsection\label{app:proof-transported}
\begin{proof}[Proof of Lemma~\ref{lem:transported}]
    Here $P_0$ imposes all retained checks. Since $\{B,h\}=0$, $G_\theta$ is Hermitian, $G_\theta^2=I$, and $G_\theta U_\theta\PC=U_\theta\PC$. At $\theta=\pi/4$ this gives Lemma~\ref{lem:transported} and
    \begin{equation}
        G_B=(h-iBh)/\sqrt2.
    \label{gc:transported}
    \end{equation}
    Hermiticity and $G_\theta^2=I$ follow from $h=h^\dagger$ and $h^2=I$ by unitary conjugation, while
    \[
    G_\theta U_\theta\PC=U_\theta h\PC=U_\theta\PC.
    \]
    Every $g\in\cS_0$ commutes with $B$ and hence with $U_\theta$, so $gU_\theta\PC=U_\theta g\PC=U_\theta\PC$.
    At $\theta=\pi/4$, the identity $G_B=R_B(\pi/2)h$ also shows that $G_B$ is Clifford, and its two distinct nonzero Pauli terms in Eq.~\eqref{gc:transported} show that it is not a Pauli.
\end{proof}

\phantomsection\label{app:proof-local-filter}
\begin{proof}[Proof of Restricted rank and centralizer (Theorem~\ref{thm:local-filter}).]
    If a Pauli $F$ on $\Omega$ commutes with $\cS_0$, either $F$ or $BF$ normalizes $\cS$, according to its commutation sign with $h$.
    Indeed, $\cS=\langle\cS_0,h\rangle$: use $F$ if $[F,h]=0$, and use $BF$ if $\{F,h\}=0$, because both $B$ and $F$ then anticommute with $h$. 
    That normalizer has weight at most $w<d$; purity forces it to be the identity. The restricted centralizer is therefore $\{I,B\}$ modulo phase. 
    Its binary dimension is one, so the restricted check span has rank $2w-1$.
    Explicitly, if $W\subseteq\mathbb F_2^{2w}$ is the span of the restricted checks, then $W^\perp=\operatorname{span}\{B\}$ and nondegeneracy of the symplectic form gives $\dim W=2w-\dim W^\perp=2w-1$. 
    The retained checks reject every other Pauli on $\Omega$, and $\{B,G_\theta\}=0$ rejects $B$. 
\end{proof}
    
For Golay codes the same argument gives
\begin{equation}
\begin{aligned}
&\supp(F)\subseteq\Omega_B\cup J,\quad |J|\leq2,\quad
F\in\cN(\cS_0)\\
&\hspace{25mm}\Longrightarrow\quad F\in\{I,B\}\quad\text{up to phase},
\end{aligned}
\label{gc:centralizer}
\end{equation}
because $|\Omega_B|+|J|\leq6<7$. Let $\mathcal F_{\Omega,r}$ denote the Hermitian Pauli basis operators with arbitrary action on $\Omega$ and nonidentity action on at most $r$ outside {qubits}. These statements concern errors within the qubit Hilbert space.

\begin{theorem}[Exact accepted channel]\label{ge:channel-filter}
Suppose $w+r<d$. Let a channel after $U_\theta$ have Kraus operators in $\operatorname{span}\mathcal F_{\Omega,r}$, written as $K_a=\sum_{F\in\mathcal F_{\Omega,r}}k_{a,F}F$ in the Hermitian Pauli basis. Acceptance of all the checks in Eq.~\eqref{ge:full-projector} gives
\begin{equation}
Q_\theta K_aU_\theta\PC=k_{a,I}U_\theta\PC.
\label{ge:kraus-filter}
\end{equation}
Its probability is $q=\sum_a|k_{a,I}|^2$, independently of the encoded input. When $q>0$, the accepted output is exactly the ideal output, also for an input entangled with an external reference. For $r=0$, the same statement holds if the full retained syndrome is replaced by any retained checks whose restricted centralizer on $\Omega$ is $\{I,B\}$.
\end{theorem}

\begin{proof}
For a nonidentity $F\in\mathcal F_{\Omega,r}$, the operator $U_\theta^\dagger F U_\theta$ is traceless and supported on at most $w+r<d$ qubits. Purity implies $\PC E\PC=0$ for every nonidentity Pauli $E$ on such a support. Expanding $U_\theta^\dagger F U_\theta$ in Paulis therefore gives $\PC U_\theta^\dagger F U_\theta\PC=0$. For $F=I$ the result is $\PC$. Equation~\eqref{ge:kraus-filter} follows by linearity. Summing the accepted Kraus maps gives $qU_\theta\rho U_\theta^\dagger$ for every encoded density operator $\rho$, including a reference system.

For the smaller set of retained checks and $r=0$, every $F\notin\{I,B\}$ is killed by a retained-check projector. The remaining $B$ is killed by $(I+G_\theta)/2$, because it anticommutes with $G_\theta$. This proves the same Kraus identity.
\end{proof}

For a channel confined to the four Golay qubits in the rotation support, write its Kraus operators as $K_{a,\Omega_B}\otimes I_{\Omega_B^c}$. Then
\begin{equation}
q=\frac{1}{4^4}\sum_a\bigl|\operatorname{Tr}K_{a,\Omega_B}\bigr|^2.
\label{ge:trace-acceptance}
\end{equation}
In particular, $K=R_B(\varepsilon)$ gives $q=\cos^2(\varepsilon/2)$ and the ideal accepted rotation.
Any sequence of operations confined to $\Omega_B$ is included by composition; acceptance can be zero.
The full $21$ retained checks and $G_B$ give the same exact accepted channel for errors on $\Omega_B$ and at most two outside qubits.

\phantomsection\label{app:proof-angle-filter}
\begin{proof}[Proof of Lemma~\ref{lem:angle-filter}]
Set $U_B=R_B(\pi/4)$ and $\Pi_+=(I+G_B)/2=U_B(I+h)U_B^\dagger/2$.
Since $h\PC=\PC$ and $hB\PC=-B\PC$, we have
\begin{align*}
\Pi_+R_B(\pi/4+\varepsilon)\PC
&=U_B\frac{I+h}{2}R_B(\varepsilon)\PC\\
&=U_B\frac{I+h}{2}\left(\cos(\varepsilon/2)I-i\sin(\varepsilon/2)B\right)\PC\\
&=\cos(\varepsilon/2)U_B\PC.
\end{align*}
Applying this identity to $|\psi\rangle\in\cC$ proves Eq.~\eqref{eq:angle-filter}, and its squared norm is the acceptance probability $\cos^2(\varepsilon/2)$ for a normalized input.
\end{proof}

\begin{theorem}[Exact restart recovery]\label{ge:restart}\label{lem:safe-interruption}
Under the assumptions preceding Theorem~\ref{ge:channel-filter}, suppose $w+2t<d$. Let the rotation followed by the error have Kraus operators $K_aU_\theta$, with $K_a\in\operatorname{span}\mathcal F_{\Omega,t}$. A full syndrome measurement of $\cC$ followed by a decoder restricted to $\mathcal F_{\Omega,t}$ recovers the original encoded input exactly. The conclusion remains valid after any outcome of ideal retained-check and $G_\phi$ measurements, for arbitrary measurement angle $\phi$. It concerns recovery of the input for a restart, rather than completion of the desired non-Clifford gate.
\end{theorem}

\begin{proof}
The set $\mathcal F_{\Omega,t}$ is closed, up to phase, under multiplication by $B$, so $K_aU_\theta\PC$ is a linear combination of $F\PC$ with $F\in\mathcal F_{\Omega,t}$. Two distinct such Paulis $F,F'$ have a nonidentity product supported on at most $w+2t<d$ qubits. By purity this product cannot be a normalizer. Their full $\cC$ syndromes are therefore distinct. The full syndrome measurement selects a unique $F\PC$ in each nonzero branch, and applying $F^\dagger$ recovers the input. The coefficients can depend on the noise and measured outcomes but not on the logical input. This is the usual error-and-erasure application of the Knill--Laflamme condition~\cite{knill1997theory}.

An ideal retained check $g$ acts on $F\PC$ by its commutation sign with $F$, since $g\PC=\PC$. Also, $hF\PC=\pm F\PC$ and $BhF\PC=\pm BF\PC$. Consequently every projector $(I\pm G_\phi)/2$ preserves $\operatorname{span}\{F\PC:F\in\mathcal F_{\Omega,t}\}$. Conditioning on any sequence of those ideal outcomes cannot leave this correctable space. The same syndrome measurement and recovery therefore apply after monitor rejection.
\end{proof}

For Golay, $w=4$, $d=7$, and $t=1$. The decoder uses
\begin{equation}
\begin{aligned}
\mathcal F_{\Omega_B,1}
=\{P_{\Omega_B}P_q:\;&P_{\Omega_B}\in\cP(\Omega_B),\ q\notin\Omega_B,\\
&P_q\in\{I,X_q,Y_q,Z_q\}\}.
\end{aligned}
\label{gc:recovery-family}
\end{equation}
After removing duplicates there are
\begin{equation}
4^4[1+3(23-4)]=14\,848
\label{ge:restart-count}
\end{equation}
hypotheses, each with a distinct original-code syndrome. Thus the same recovery covers arbitrary errors on all four qubits in the rotation support and one unknown outside qubit, including coherent superpositions. A minimum-weight decoder uninformed of $\Omega_B$ need not recover this family. 
Theorem~\ref{ge:restart} assumes ideal recovery; faults in its circuit implementation are included in \SM{}~\ref{app:gc-recovery} and the proof of Theorem~\ref{thm:golay-complete}.

\begin{proposition}[A monitor does not remove its own angle error]\label{ge:calibration}
If the actual rotation angle is $\theta+\varepsilon$ and the ideal transported-check measurement is calibrated at $\theta+\delta$, then
\begin{equation}
Q_{\theta+\delta}R_B(\theta+\varepsilon)\PC
=\cos\!\left(\frac{\varepsilon-\delta}{2}\right)
 R_B(\theta+\delta)\PC.
\label{ge:miscalibration}
\end{equation}
The accepted rotation, when its acceptance probability is nonzero, is exactly $R_B(\theta+\delta)$, where $\theta+\delta$ is the calibration angle defining the measured transported check $G_{\theta+\delta}$.
Its residual angle error relative to the desired rotation $R_B(\theta)$ is therefore $\delta$, independently of $\varepsilon$.
\end{proposition}

\begin{proof}
Using Eq.~\eqref{ge:full-projector}, the left side is $R_B(\theta+\delta)\PC R_B(\varepsilon-\delta)\PC$. Since $B$ has nonzero syndrome, $\PC B\PC=0$, and expansion of the remaining rotation gives Eq.~\eqref{ge:miscalibration}.
\end{proof}

\subsection{\texorpdfstring{A protected direct Golay gate}{A protected direct Golay gate}}
\label{app:golay-completion}

The direct Golay circuit specified in \SM{}~\ref{app:gc-control}--\ref{app:gc-recovery} and illustrated in Fig.~\ref{fig:golay-complete} protects $G_B$ with two verified cat controls and replaces $R_A(\pi/2)$ by the completion by Pauli measurements and corrections of Theorem~\ref{pm:adaptive}. 
It acts on the original $23$-qubit block without an encoded magic-state resource and requires only one- and two-qubit Clifford+$T$ gates. 
Transporting the check through the circuit follows Ref.~\cite{dasu2026flagging}; {the single-fault correctness proof of Theorem~\ref{thm:golay-complete} uses} the full retained syndrome and Theorem~\ref{ge:restart}.

\begin{figure}[!htbp]
\centering
\begin{tikzpicture}[font=\small,box/.style={draw,rounded corners=2pt,align=center,inner sep=5pt},flow/.style={-{Stealth[length=2mm]},thick}]
\node[anchor=west] at (-1.7,1.55) {(a) Complete gate, with attempt number $j\in\{1,2\}$};
\node[box,text width=2.7cm,minimum height=1.55cm] (lead) at (0,0) {Full Golay correction};
\node[above=2pt of lead] {Encoded input; $j=1$};
\node[box,text width=3.4cm,minimum height=1.55cm] (monitor) at (4.05,0) {Prepare cats $a,c$\\with verification\\ Rotation and monitor in (b)};
\node[box,text width=6.0cm,minimum height=1.55cm] (finish) at (10.25,0) { Correction using retained checks\\Protected $G\to y$; if $y=-1$, $M\to r$\\Protected $h\to z$; $A$ if $z=-1$\\$L$ if $y=-1$ and $rz=-1$\\Full Golay correction; $R_L(\pi/4)$ output};
\node[box,text width=3.1cm] (recover) at (0,-2.25) {Full Golay syndrome\\Apply $F^\dagger$ for its modeled hypothesis $F$; set $j=2$};
\node[box,text width=2.0cm] (abort) at (8.15,-2.25) {Located abort};
\node[box,text width=3.5cm] (caps) at (11.85,-2.25) {Cat or correction cap without success, or no recovery hypothesis};
\draw[flow] (lead) -- (monitor);
\draw[flow] (monitor) -- node[above]{accept} (finish);
\draw[flow] (monitor.south) |- node[pos=.7,above]{reject, $j=1$} (recover.east);
\draw[flow] (recover.north) -- node[left]{restart} (lead.south);
\draw[flow] (monitor.south) -- (4.05,-1.6) -| (abort.north);
\node[above] at (6.0,-1.6) {reject, $j=2$};
\draw[flow] (caps.west) -- (abort.east);
\node[anchor=west] at (-1.7,-3.55) {(b)  Rotation and monitor: chronological order follows the arrows};
\node[box,text width=3.1cm] (raw) at (0,-4.45) {$U=R_B(\pi/4)$};
\node[box,text width=3.1cm] (ha) at (4,-4.45) {$C_a(h)$};
\node[box,text width=3.1cm] (hc1) at (8,-4.45) {$C_c(h)$};
\node[box,text width=3.1cm] (w) at (12,-4.45) {$W=C_{a_1}(U^2)$};
\node[box,text width=3.1cm] (hc2) at (12,-5.65) {$C_c(h)$};
\node[box,text width=3.1cm] (terminal) at (8,-5.65) {$V_1,V_2,V_3,V_4$};
\node[box,text width=3.1cm] (phase) at (4,-5.65) {$CS^\dagger_{a_1,c_1}$};
\node[box,text width=3.1cm] (flags) at (0,-5.65) {Six $ZZ$ checks\\Eight $X$ readouts\\Flags $f_a,f_c$};
\node[box,text width=3.1cm] (retained) at (0,-7.65) {Reset/reuse ancillas\\Measure $g_1,\ldots,g_{21}$\\Record $\boldsymbol t$};
\node[box,text width=7.3cm] (accept) at (6.7,-7.65) {Accept iff all six $ZZ$ results, $f_a,f_c$,\\and $t_1,\ldots,t_{21}$ are $+1$};
\draw[flow] (raw) -- (ha);
\draw[flow] (ha) -- (hc1);
\draw[flow] (hc1) -- (w);
\draw[flow] (w) -- (hc2);
\draw[flow] (hc2) -- (terminal);
\draw[flow] (terminal) -- (phase);
\draw[flow] (phase) -- (flags);
\draw[flow] (flags) -- (retained);
\draw[flow] (retained) -- (accept);
\end{tikzpicture}%
\caption{Complete protected direct Golay gate of Theorem~\ref{thm:golay-complete}. Panel (a) includes the accepted and rejected branches;  panel (b) expands the rotation and monitor using Eqs.~\eqref{gc:pairs}--\eqref{gc:terminal}, with the initial controlled-$h$ of Eq.~\eqref{gc:sequence} omitted. The $21$ generators $g_i$ are listed in Table~\ref{tab:gc-retained}. Each logical measurement of $G$, $M$ or $h$ uses three verified extractions and  error correction using the retained checks before the first and after each, as specified in \SM{}~\ref{app:gc-completion}; the conditional Pauli corrections are also specified there. On first rejection, the full syndrome selects the recovery hypothesis of Eq.~\eqref{gc:recovery-family}, as explained in \SM{}~\ref{app:gc-recovery}. All branches use the same ancillas serially, with at most two candidates per requested cat and four rounds per correction; reaching a cap without success or obtaining no recovery hypothesis gives the located abort of Definition~\ref{def:located-failure}.}
\label{fig:golay-complete}
\par
\end{figure}
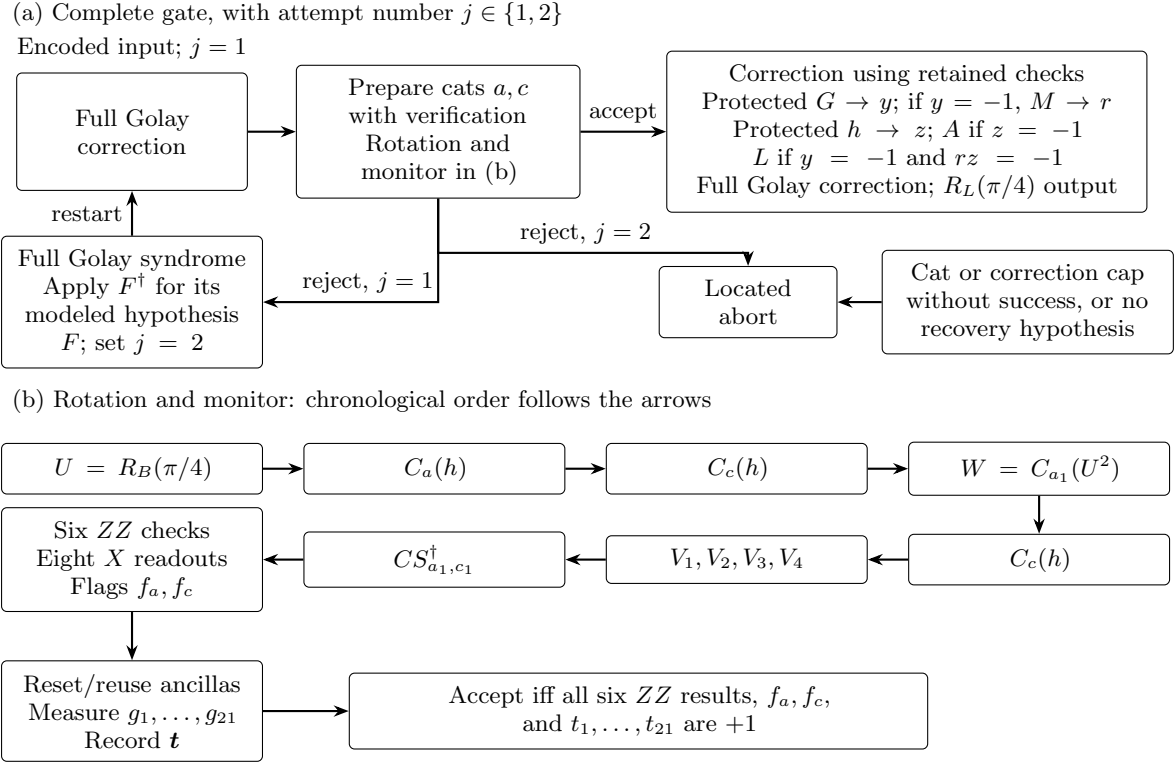

\subsubsection{Control preparation and circuit order}
\label{app:gc-control}

Prepare two four-qubit cat registers $a=(a_1,\ldots,a_4)$ and $c=(c_1,\ldots,c_4)$ in $(|0000\rangle+|1111\rangle)/\sqrt2$. Use the same assignment of data pairs for both registers:
\begin{equation}
(1,2),\qquad(14,9),\qquad(15,16),\qquad(21,17).
\label{gc:pairs}
\end{equation}
Each  cat-state qubit therefore couples to one qubit in the rotation support and one qubit outside $\Omega_B$. Define $C_a(h)$ by the eight CZ gates from each $a_j$ to its two assigned data qubits, and define $C_c(h)$ similarly. On the cat repetition-code spaces these implement controlled $h$. Prepare each cat with a fixed Clifford encoder and verify its three adjacent $ZZ$ parities before any data coupling. A parity uses one reset ancilla in $|0\rangle$, two CNOTs from the corresponding  cat-state qubits, and a $Z$ readout. Allow two serial candidates for each requested cat. 
Two failed candidates are a located abort (Definition~\ref{def:located-failure}) and require at least two faults. 
Both monitor cats are ready before the raw rotation.
For this Golay circuit, the flag in Definition~\ref{def:located-failure} declares a located abort if both candidates for any requested cat are rejected, if a correction or recovery procedure reaches four syndrome rounds without two consecutive identical records, if the full syndrome used for restart recovery has no representative in Eq.~\eqref{gc:recovery-family}, or if the second monitor attempt is rejected.
These stopping rules apply also to leading and trailing correction and to the correction steps in the  completion by protected Pauli measurements and corrections.
A first monitor rejection instead initiates the input recovery of \SM{}~\ref{app:gc-recovery} and, if that recovery completes, a second attempt.
\par

In chronological order, consider the sequence whose implemented portion is shown in Fig.~\ref{fig:golay-complete}(b)
\begin{equation}
\begin{aligned}
C_a(h),\ U,\ C_a(h),\ C_c(h),\ W,\ C_c(h),\ V,\\
W=C_{a_1}(U^2),\qquad V=C_{a,c}(U^4).
\end{aligned}
\label{gc:sequence}
\end{equation}
The first $C_a(h)$ is only needed to state an operator identity on arbitrary data; omit it in the gate, since it acts trivially on a fault-free input in $\cC$. The raw rotation is supported on $\Omega_B$. Implement $W$ entirely within $\Omega_B\cup\{a_1\}$ using
\begin{equation}
W=R_B(\pi/4)R_{Z_{a_1}B}(-\pi/4).
\label{gc:W}
\end{equation}
Each rotation has a parity-ladder implementation with one physical $T$-type gate. By Lemma~\ref{lem:compiled-support}, a fault within either completed unitary block is equivalent at its output to an arbitrary error on its stated support; no protected non-Clifford primitive is assumed.

For $V$, write $(B_1,B_2,B_3,B_4)=(Y_1,Z_{14},Z_{15},Z_{21})$. Apply four separate doubly controlled rotations
\begin{equation}
V_j=CC_{a_j,c_j}R_{B_j}(\pi),\qquad j=1,\ldots,4,
\label{gc:terminal}
\end{equation}
and then apply $CS^\dagger$ to $a_1,c_1$. Each $V_j$ touches only one data qubit. On the two cat code spaces the four rotations supply $B^{ac}$, and the final control phase supplies $(-i)^{ac}$; their product is $U^{4ac}$. This compilation must not use a shared data parity ladder for the four terminal targets. The control phase occurs after all terminal data couplings. These support and ordering conditions are used in the recovery proof.

\begin{lemma}[Ideal coherent-control identity]\label{lem:gc-identity}
The full sequence in Eq.~\eqref{gc:sequence} equals $U$ on the data and the identity on the two logical cat controls. After omission of its first $C_a(h)$, the portion after the raw rotation implements $C_a(G_B)$, independently of $c$.
\end{lemma}
\begin{proof}
For fixed cat values $a,c\in\{0,1\}$, the data operation is
\begin{equation}
U^{4ac}h^cU^{2a}h^ch^aUh^a
=U^{(-1)^a+2a(-1)^c+4ac}=U,
\label{gc:identity}
\end{equation}
where $hUh=U^{-1}$ and the exponent equals one in all four cases. Removing the rightmost $h^a$ leaves $Uh^a$. Factoring off the raw rotation on the right gives the operation after the rotation $Uh^aU^\dagger=G_B^a$.
\end{proof}

After $V$, remeasure all adjacent $ZZ$ parities of both cats, then measure their eight qubits in $X$. Let $f_a,f_c$ be the products within the two registers. 
Accept the cat stage only when all six parity checks and both products are positive. 
Reset the eight cat-state qubits and reuse them for one complete verified-cat extraction of $\cS_0$, with one-to-one cat--data couplings. 
Let $g_1,\ldots,g_{21}$ be the signed generators in Table~\ref{tab:gc-retained}.
For each $g_i$, prepare a verified eight-qubit cat, apply the eight controlled single-qubit Pauli factors of $g_i$ with one  cat-state qubit per supported data qubit, and measure all eight  cat-state qubits in $X$.
Their product $t_i\in\{\pm1\}$ is the reported outcome of measuring $g_i$; reset and reuse the ancillas between generators.
The retained measurement record is $\boldsymbol t=(t_1,\ldots,t_{21})$.
\par
Accept this stage precisely when $t_i=+1$ for every $i$, equivalently when all $21$ retained syndrome bits are zero, as shown in Fig.~\ref{fig:golay-complete}(b). 
A negative final cat parity, $f_a=-1$, $f_c=-1$, or any $t_i=-1$ rejects the current attempt.
After the first rejection, measure a full original-code syndrome, including $h$, and apply the inverse of its hypothesis in Eq.~\eqref{gc:recovery-family} to recover the original input before restarting, as specified in \SM{}~\ref{app:gc-recovery} and Fig.~\ref{fig:golay-complete}(a).
The second rejection gives a located abort as in Definition~\ref{def:located-failure}.

\subsubsection{\texorpdfstring{Measurement-based completion and resource counts}{Measurement-based completion and resource counts}}
\label{app:gc-completion}

After acceptance, correct using only $\cS_0$, then follow the accepted branch in Fig.~\ref{fig:golay-complete}(a). Representatives for the  completion by protected Pauli measurements and corrections are
\begin{align}
G=iAh&\equiv Y_1Z_4Z_5Z_8Z_{16}\pmod{\cS_0},\nonumber\\
M=AhL=-iBh&=X_1Z_2Z_9Z_{16}Z_{17},\nonumber\\
h&=Z_1Z_2Z_9Z_{14}Z_{15}Z_{16}Z_{17}Z_{21}.
\label{gc:completion}
\end{align}
Measure $G$ with result $y$. If $y=-1$, measure $M$ with result $r$. Measure $h$ with result $z$, apply $A$ when $z=-1$, and additionally apply $L$ when $y=-1$ and $rz=-1$. Theorem~\ref{pm:adaptive} proves that every corrected branch implements $R_L(\pi/4)$.

Protect each logical measurement by three verified  check measurements, with  error correction using the retained checks before the first and after each, and compare signs in the corrected Pauli frame as in Lemma~\ref{pm:measurement}. Correction measures only $\cS_0$ and does not fix $h$ prematurely. Repeat  complete syndrome measurement rounds until two consecutive records agree, with a cap of four rounds, as in Lemma~\ref{c22:cat}; declare a located abort (Definition~\ref{def:located-failure}) if no agreement is found. With at most one fault, rounds before that fault have one fixed syndrome, rounds after it have a fixed syndrome, and only its containing round can have a mixed record; agreement is obtained by round four. Verified one-to-one extraction leaves at most one data error and at most one  unreliable individual measurement record. The intervening clean correction removes the data error, and the other two logical-measurement records determine the correct decoded result. This procedure replaces the  unprotected second rotation entirely.

The monitor uses $23$  data qubits, two four-qubit cats, and one verification ancilla. After the monitor, reuse the eight  cat-state qubits for weight-eight syndrome extraction and the weight-five or weight-eight logical checks in Eq.~\eqref{gc:completion}.  All generators of the original Golay code can also be chosen weight eight. Thus
\begin{equation}
n_{\rm peak}=23+4+4+1=32.
\label{gc:peak}
\end{equation}
This count assumes serial ancilla reuse with reset and flexible two-qubit connectivity. Routing, parallel ancilla preparation, and conversion to another memory code have separate costs. The $23$  data qubits are unchanged, although the data alone need not lie in a fixed stabilizer code during a partial controlled-check operation.

For a terminal $Z$ target with computational bit $t$, the phase identity
\begin{equation}
t-(a\mathbin\oplus t)-(c\mathbin\oplus t)
 +(a\mathbin\oplus c\mathbin\oplus t)=(-2+4t)ac
\label{gc:phase-polynomial}
\end{equation}
gives an exact four-$T$ circuit for $CC\,R_Z(\pi)$: in chronological order apply $T_t$, $\mathrm{CX}_{a,t}$, $T_t^\dagger$, $\mathrm{CX}_{c,t}$, $T_t$, $\mathrm{CX}_{a,t}$, $T_t^\dagger$, $\mathrm{CX}_{c,t}$. For the $Y$ target, put $F=R_X(-\pi/2)$ and apply $F^\dagger$ before this circuit and $F$ after it. The last $CS^\dagger$ uses three $T$-type gates: apply $T^\dagger$ to both controls, compute their XOR into one control, apply $T$ there, and uncompute. Therefore
\begin{equation}
N_T=1+2+4\cdot4+3=22
\label{gc:T-count}
\end{equation}
per attempt, counting the raw rotation as its one-$T$ compilation. Two attempts use at most $44$ $T$-type gates; retries of cat preparation consume none. No optimality or space--time advantage is asserted.

\subsubsection{\texorpdfstring{Recovery and a finite schedule}{Recovery and a finite schedule}}
\label{app:gc-recovery}

Every rejected single-fault branch has a Kraus operator in the span of Paulis on $\Omega_B$ and at most one unknown outside qubit, acting on the original input. To see this, move complete factors of $h$ to the input, where they act trivially. All rotations and terminal data operations are supported on $\Omega_B$. A cat bit fault contributes at most the one outside qubit assigned to its wire in Eq.~\eqref{gc:pairs}. A fault in $V_j$ may affect both cat registers, but all controlled-$h$ passes have already finished. It cannot propagate to further outside data qubits. A sole retained-extraction fault adds at most one data qubit. An encoder fault that fails initial cat verification is discarded before data coupling and uses the allowed replacement.

Theorem~\ref{ge:restart} therefore applies to every rejected branch, using the $14\,848$ hypotheses of Eq.~\eqref{gc:recovery-family}. Full original-code syndrome extraction and the inverse of its unique hypothesis recover the input, including entanglement with a reference.
For arbitrary observed records, complete the decoder by declaring a located abort (Definition~\ref{def:located-failure}) if the extracted syndrome has no hypothesis in Eq.~\eqref{gc:recovery-family}; the single-fault recovery guarantee ensures that this case cannot occur in the one-fault regime. This recovery uses the known rotation support; it is not unrestricted minimum-weight decoding.

A fault-free history never reaches restart recovery. With at most one fault, a reached recovery branch has already used that fault, so recovery is ideal in this part of the proof. Recovery operations remain locations in the noise model; a fault during recovery together with the earlier rejection is a two-fault event. The same verified-cat and repeated-syndrome procedure can be used for full Golay recovery. Allow two gate attempts, as shown in Fig.~\ref{fig:golay-complete}(a): recover after the first rejection, and declare a located abort (Definition~\ref{def:located-failure}) after the second. With one fault, the second attempt is ideal whenever it is needed. Together with the two-candidate cap for every cat in \SM{}~\ref{app:gc-control} and the four-round correction cap in \SM{}~\ref{app:gc-completion}, this gives a fixed finite schedule.

\subsubsection{\texorpdfstring{Single-fault correctness}{Single-fault correctness}}
\label{app:gc-single-fault}

Use local stochastic circuit noise, including preparations, measurements, resets, and waits, with reliable classical processing. A specified set of $r$ faulty locations has probability at most $p^r$, and a faulty location may act arbitrarily on its participating qubits. All displayed gates can be compiled into one- and two-qubit Clifford+$T$ gates.

\begin{theorem}[Protected direct Golay gate]\label{thm:golay-complete}
Use leading Golay correction, the rotation and monitor of \SM{}~\ref{app:gc-control}, error correction using the retained checks, the completion by protected Pauli measurements and corrections of \SM{}~\ref{app:gc-completion}, and trailing Golay correction. 
Here the leading and trailing Golay corrections use all $22$ independent stabilizer checks of the original $\qcode{23}{1}{7}$ code, before the rotation and after the final Pauli corrections, respectively.
The completion measures $G$ with outcome $y$, measures $M$ with outcome $r$ if $y=-1$, and measures $h$ with outcome $z$, then applies $A$ if $z=-1$ and additionally $L$ if $y=-1$ and $rz=-1$.
Each logical measurement is protected by three verified check measurements, with correction using $\cS_0$ before the first and after each, as specified in \SM{}~\ref{app:gc-completion}.

If the number of incoming physical errors plus internal faulty locations is at most one, an attempt either accepts the ideal $R_L(\pi/4)$ output up to at most one correctable physical error, or recovers the unknown encoded input exactly for a restart. A fault-free attempt is never rejected. With the stated serial schedule, one attempt uses at most $32$ simultaneous physical qubits and, when fully compiled, $22$ single-qubit $T$-type gates. Two gate attempts and at most two candidates for each fresh cat give a finite circuit whose logical failure or located abort requires at least two faults.
\end{theorem}
\begin{proof}
A clean leading correction removes one incoming error. If leading correction contains the sole fault, it leaves at most one physical error, which is covered by the following cases. First suppose the fault precedes final retained extraction, so the subsequent retained checks are ideal.

\emph{Cat bit errors.} A cat-encoder fault followed by ideal verification is rejected unless the bit pattern is constant. The all-one pattern is the cat stabilizer. If the only fault is in parity verification, the preceding cat is ideal; a parity CNOT can introduce at most one cat bit error. A verifier phase can propagate to two cat-state qubits but cannot create two dangerous bit errors. All completed data-control blocks are diagonal in their cat computational controls. A nonconstant bit pattern persists until final parity verification and is rejected. The same argument removes the $X$ and $Y$ components of a fault on $a_1$ inside $W$, and nonconstant bit components from a terminal $V_j$.

\emph{Raw rotation and preceding data errors.} Expand a fault on $\Omega_B$ in Paulis. With no cat bit error, the subsequent circuit changes a component $E$ only within the data span of $E$ and $BE$, with possible control phases. These two terms have the same retained syndrome. Equation~\eqref{gc:centralizer} leaves only $I$ and $B$; the $B$ component flips $f_a$ and is rejected. A single error before the rotation or on a waiting outside data qubit obeys the same support argument.

\emph{Faults inside $W$.} Absorb any fault in its compilation into an error after the completed block, supported on $\Omega_B\cup\{a_1\}$. Final cat parity verification removes its bit components. Retained extraction then leaves only $B^bZ_{a_1}^s$. Their flag pairs are
\begin{center}
\begin{tabular}{ccc}
Component&$f_a$&$f_c$\\\hline
$I$&$+$&$+$\\
$Z_{a_1}$&$-$&$+$\\
$B$&$+$&$-$\\
$Z_{a_1}B$&$-$&$-$
\end{tabular}
\end{center}
Thus only the identity survives. The second cat is essential because the $B$ component has $f_a=+1$. Neither $B$ nor the correlated component $Z_{a_1}B$ passes both tests.

\emph{Controlled-$h$ and terminal faults.} In the absence of a cat bit error, a faulty controlled-$h$ gate gives a single-qubit data operator $E_q$ and cat phases. Its suffix produces only $E_q$ and $BE_q$. Every nonidentity $E_q$ has nonzero retained syndrome because $\cD$ has distance four; the same is true of $BE_q$. A pure control phase either changes a flag or leaves the data ideal. A terminal $V_j$ fault without cat bit errors affects only its one data target, so retained extraction detects its nonidentity data components. The other terminal blocks never touch that target. A fault in the last $CS^\dagger$ acts only on controls.

\emph{Extraction, readout, and waits.} If final cat verification or readout contains the sole fault, the data entering that stage are ideal and disentangled from the cats. The fault can therefore cause only harmless acceptance or false rejection. A sole fault in final retained extraction acts on ideal data and leaves at most one physical residual, possibly with an incorrect syndrome record. If accepted, the following clean  error correction using the retained checks removes it. Wait faults on data or controls are included in the preceding support and cat-bit cases.

Consequently an accepted monitor gives $U|\psi\rangle$, possibly with one correctable residual from retained extraction. The argument applies termwise to every fault Kraus operator and therefore includes coherent single-location faults and inputs entangled with a reference system. 
The recovery argument above in \SM{}~\ref{app:gc-recovery} proves exact restoration in rejected branches.
Protected Pauli measurements then implement Theorem~\ref{pm:adaptive}; leading corrections remove any earlier residual, and trailing correction either removes the sole preceding error or leaves at most one error if it is itself faulty.
\end{proof}

This theorem establishes fault tolerance of the monitor for arbitrary unknown inputs in $\cC$, whose ideal state after the non-Clifford rotation lies in the $G_B=+1$ eigenspace.
It does not assert a protected classical eigenvalue for an independent measurement of $G_B$ on arbitrary states in $\cD$: a pure phase fault on the readout cat can flip that label. In the gate, a nominal negative eigenvalue already requires another error.

\subsubsection{{Failure probability and scope}}
\label{app:gc-failure}
Here we convert the single-fault guarantee of Theorem~\ref{thm:golay-complete} into the quadratic failure-or-abort bound quoted in Sec.~\ref{subsec:golayprotected}.
The counting includes the complete two-attempt circuit and all capped preparation, correction and recovery branches of \SM{}~\ref{app:gc-control}--\ref{app:gc-recovery}.
We also distinguish this analytical circuit bound from the algebraic checks of the construction and explain why it does not imply a numerical threshold, a recursive Golay architecture or protection against a shared calibration error as in Eq.~\eqref{ge:miscalibration}.
\par

For an ideally encoded input, pad the adaptive branches into a fixed finite set of $N_G$ elementary locations, including recovery, ancilla candidates, classical-feed-forward waits, and idles.
To determine $N_G$, first fix the elementary Clifford+$T$ decomposition of every operation and the numbers of idle locations required during measurement and classical feed-forward.
Include slots for both gate attempts, both candidates for every requested cat state, all four allowed rounds of each correction, restart recovery, and every conditional Pauli measurement or correction, padding unused branches with idles as needed.
Count one location for each elementary one- or two-qubit gate, preparation, measurement or reset, and for each required one-qubit idle; a two-qubit gate counts as one location.
Summing these slots gives $N_G$, with shared prefix operations counted once and all padding counted.
The displayed circuit identities and retry caps bound this count, but do not assign it a unique numerical value without these decomposition and timing choices.
A union bound over fault pairs gives
\begin{equation}
\Pr(\text{logical failure or located abort})
\leq \binom{N_G}{2}p^2.
\label{gc:failure-bound}
\end{equation}
Let $\mathsf{F}_i$ be the event that location $i\in\{1,\ldots,N_G\}$ is faulty, and let $\mathsf{Bad}$ denote logical failure or located abort.
For an ideally encoded input, Theorem~\ref{thm:golay-complete} implies that $\mathsf{Bad}$ requires at least two faulty locations, so
\[
\mathsf{Bad}\subseteq\bigcup_{1\leq i<j\leq N_G}(\mathsf{F}_i\cap\mathsf{F}_j).
\]
The union bound and the local-stochastic condition of Eq.~\eqref{eq:local-stochastic-budget} therefore give
\begin{align*}
\Pr(\mathsf{Bad})
&\leq\sum_{1\leq i<j\leq N_G}\Pr(\mathsf{F}_i\cap\mathsf{F}_j)
\leq\sum_{1\leq i<j\leq N_G}p^2
=\binom{N_G}{2}p^2.
\end{align*}
No independence assumption is used, and fault sets containing three or more locations are included because they contain at least one pair.
\par
Correctable output residuals are excluded from logical failure. 
This coarse bound counts the complete two-attempt circuit of Theorem~\ref{thm:golay-complete}, including the recovery and finite caps in \SM{}~\ref{app:gc-recovery}; its numerical coefficient follows after counting $N_G$ as specified above.
A malignant-pair analysis instead identifies pairs of elementary locations for which some allowed faults cause logical failure or abort when every other location is ideal.
It can refine the two-fault contribution, but fault sets of three or more locations require a separate bound; no such pair classification is performed here.

Checks of the Golay gate of Theorem~\ref{thm:golay-complete} use the signed code in \SM{}~\ref{app:golay-details}, the retained centralizer in Eq.~\eqref{gc:centralizer}, the recovery family in Eq.~\eqref{gc:recovery-family}, and the Pauli representatives in Eq.~\eqref{gc:completion}.
The coherent-control and elementary-gate identities are given in Lemma~\ref{lem:gc-identity} and Eqs.~\eqref{gc:W} and~\eqref{gc:phase-polynomial}.
They verify the algebra and the stated supports, and do not constitute a noisy-circuit simulation. 
The single-fault proof uses those supports together with verified-cat extraction and completion by protected Pauli measurements and corrections; a useful error coefficient requires the scheduled location count or a malignant-pair analysis.

This is a single-fault guarantee. Golay distance seven alone does not give three-fault correctness for this circuit, and the monitor alone supplies neither a numerical threshold nor a recursive construction. Recursion additionally requires a complete set of protected Clifford, preparation, measurement, and  correction circuits with the necessary rectangle properties. Comparisons with other methods must include ancillas, retries, extraction time, and Clifford locations. Single-location coherent faults are included, whereas calibration errors shared across several non-Clifford locations are multi-location noise. Equation~\eqref{ge:miscalibration} shows why acceptance does not remove a common angle error shared by the rotation and its check circuit.

\section{{Code design}}\label{app:resources}

Here we quantify the rate and syndrome-extraction costs discussed in Sec.~\ref{sec:discussion} for the BCH codes of Proposition~\ref{cor:bch} and their selective concatenations.
Equation~\eqref{eq:active-union-rate} gives the rate when several chosen logical supports are encoded, and Proposition~\ref{prop:logical-access-rate} bounds the cost of encoding the supports of a complete commuting logical basis.
Equation~\eqref{eq:bch-check-interactions} gives the number of direct data--ancilla interactions needed to measure an independent generator basis.
These bounds explain the implementation costs accompanying the high-rate family in Sec.~\ref{sec:intermediate} and the seven-site selective construction of Corollary~\ref{cor:selective-bch}.
\par

\subsection{High-rate code examples and extraction cost}\label{app:rate-design}

\subsubsection{\texorpdfstring{How many logical rotations share the selective protection?}{How many logical rotations share the selective protection?}}
We ask whether selectively encoding several logical Pauli supports can preserve a high code rate. The following bound shows why protecting a complete commuting logical basis costs more than protecting a few selected rotations.
Let an outer $\qcode{n}{K}{d}$ code have a chosen collection of logical Pauli axes $L_1,\ldots,L_r$, and actively encode their combined support $J=\bigcup_a\supp(L_a)$ in 15-qubit inner blocks. Writing $h=|J|$, the data-code rate is exactly
\begin{equation}
R_J=\frac{K}{n+14h},\qquad
1-R_J=\frac{n-K+14h}{n+14h}.
\label{eq:active-union-rate}
\end{equation}
For the BCH family, $n-K=6m$. Therefore $h=o(n)$ preserves $R_J\to1$. In particular, $r=o(n)$ chosen weight-seven axes have $h\leq7r=o(n)$. This counts the code hosting those rotations; it neither establishes their fault tolerance nor guarantees that the chosen axes form a full logical basis.

\begin{proposition}[Rate cost of a full commuting logical basis]
\label{prop:logical-access-rate}
Suppose the same active set $J$ contains the supports of $K$ independent, mutually commuting logical Pauli representatives, one for each member of a complete logical $Z$ basis. Then $h\geq K$ and the selective construction obeys
\begin{equation}
R_J\leq\frac{K}{n+14K}.
\label{eq:full-logical-access-rate}
\end{equation}
For a family with $K/n\to1$, its asymptotic upper bound is $1/15$.
\end{proposition}

\begin{proof}
Modulo phases, Paulis supported on $J$ form a $2h$-dimensional binary symplectic space. Independence of the logical representatives modulo the stabilizer implies their independence in this physical space. Their pairwise commutation makes their span isotropic, whose dimension is at most $h$. Hence $K\leq h$, and substitution into Eq.~\eqref{eq:active-union-rate} proves the bound.
\end{proof}

The restriction to full-support active encoding is essential: a construction that leaves some of a selected axis unencoded is outside this argument. The proposition is a limitation of this selective implementation rule, not a general bound on fault-tolerant computation with high-rate codes.

\subsubsection{\texorpdfstring{Dense checks and the cost of a syndrome round}{Dense checks and the cost of a syndrome round}}

For the BCH CSS code, let $d_m^\perp$ denote the minimum distance of its classical dual. Every nonidentity stabilizer has weight at least $d_m^\perp$: its nonzero $X$ or $Z$ component belongs to that dual, and the Pauli support contains the support of that component. There are $n-K=6m$ independent stabilizer generators. 

Choose independent generators $g_1,\ldots,g_{6m}$.
A \emph{direct controlled-Pauli coupling} is the two-qubit gate $|0\rangle\langle0|_a\otimes I_i+|1\rangle\langle1|_a\otimes P_i$, with syndrome ancilla $a$ as control and data qubit $i$ as target.
There is one such coupling per supported data qubit in each generator measurement.
For one measurement of each generator by these direct couplings, the \emph{interaction count} is defined as $W_m=\sum_{j=1}^{6m}\wt(g_j)$.
It excludes ancilla preparation and verification, readout, idles and repeated measurements.
For each $j$, the nonidentity generator $g_j$ has a nonzero $X$ or $Z$ component in the classical dual, so $\wt(g_j)\geq d_m^\perp$.
It acts on only $n$ data qubits, so $\wt(g_j)\leq n$.
Adding these inequalities for $j=1,\ldots,6m$ gives
\[
    \begin{aligned}
    6m\,d_m^\perp&=\sum_{j=1}^{6m}d_m^\perp\leq\sum_{j=1}^{6m}\wt(g_j)=W_m,\\
    W_m&=\sum_{j=1}^{6m}\wt(g_j)\leq\sum_{j=1}^{6m}n=6mn.
    \end{aligned}
\]
Thus, we have
\begin{equation}
    6m\,d_m^\perp\leq W_m\leq6mn.
\label{eq:bch-check-interactions}
\end{equation}

The upper bound counts one coupling per supported data qubit in one measurement of each generator; verification and repetition can increase the actual circuit cost. 
The Reed--Muller bound $d_m^\perp\geq2^{m-2}=(n+1)/4$ in the proof of Proposition~\ref{cor:bch} gives $d_m^\perp=\Theta(n)$~\cite{macwilliams1977theory}, so this interaction count $W_m$ is $\Theta(n\log n)$.
Explicitly, substituting $d_m^\perp\geq(n+1)/4$ and $m=\log_2(n+1)$ into Eq.~\eqref{eq:bch-check-interactions} gives
\[
    \frac32(n+1)\log_2(n+1)\leq W_m\leq6n\log_2(n+1).
\]
Since $K=n-6m=\Theta(n)$, this is $W_m/K=\Theta(\log n)$ couplings per encoded qubit.
If a data qubit participates in at most one two-qubit coupling per time step, the coupling depth is at least $W_m/n=\Omega(\log n)$, before connectivity, preparation, or verification costs.

The seven-site selective lift preserves this asymptotic count: lifting an outer check cannot reduce its support, increases it by at most a fixed factor, and adds only a fixed number of inner checks. Thus $K/N\to1$ does not make a syndrome round constant-cost per encoded qubit. Equation~\eqref{eq:bch-check-interactions} applies to separate direct generator measurements; it is not an architecture-independent lower bound against subsystem measurements or other extraction schemes.

\end{document}